\newcommand*{\TECHREP}{}%

\newcommand*{\APPENDIX}{}%
 
\ifdefined\PAPER
\documentclass[conference]{IEEEtran}
\fi
\ifdefined\TECHREP
\documentclass[conference,10pt,onecolumn,final,a4paper]{IEEEtran}
\IEEEoverridecommandlockouts
\fi

\usepackage{mathptmx}

\usepackage{epsfig,endnotes}
\usepackage{color, colortbl}
\usepackage{tikz}
\usepackage{graphicx}
\usepackage{amssymb,amsmath}
\usepackage{booktabs}
\usepackage{verbatim}
\usepackage{multirow}
\usepackage{subcaption}
\usepackage{hyperref}
\usepackage{cleveref}
\usepackage{url}
\usepackage[linesnumbered,algoruled,vlined,noresetcount]{algorithm2e}
\SetNlSty{bfseries}{\color{black}}{}
\usepackage[utf8]{inputenc}
\usepackage{amsthm}
\usepackage{mathtools}
\usepackage{ifthen}
\usepackage{xargs}
\usepackage{thmtools, thm-restate}

\usepackage{times}

\hypersetup{
    colorlinks=true,
    linkcolor=black,
    citecolor=black,
    filecolor=black,
    urlcolor=black,
}

\ifdefined\PAPER 
\theoremstyle{remark}
\fi
\newtheorem{theorem}{Theorem}[section]
\newtheorem{corollary}{Corollary}[theorem]
\newtheorem{lemma}[theorem]{Lemma}
\newtheorem{prop}[theorem]{Proposition}

\crefname{section}{§}{§§}
\Crefname{section}{§}{§§}

\definecolor{DarkGray}{gray}{0.85}

\newcommand{\AC}[0]{AllConcur}
\newcommand{\ACwEA}[0]{AllConcur-w/EA}
\newcommand{\ACplus}[0]{AllConcur+}

\newcommand{\broadcast}[1]{\emph{broa\-dcast}$(\mathit{#1})$}
\newcommand{\tobroadcast}[1]{\emph{A}-\emph{broa\-dcast}$(\mathit{#1})$}
\newcommand{\todeliver}[1]{\emph{A}-\emph{deliver}$(\mathit{#1})$}
\newcommand{\rbroadcast}[1]{\emph{R}-\emph{broa\-dcast}$(\mathit{#1})$}
\newcommand{\rdeliver}[1]{\emph{R}-\emph{deliver}$(\mathit{#1})$}
\newcommand{\sender}[1]{$\mathit{sender}(\mathit{#1})$}

\newcommand{\acState}[2]{[ {#1}, {#2} ]}
\newcommand{\acFState}[2]{[ {#1}, {#2} ]_\triangleright}
\newcommand{\acRState}[2]{[\![  {#1}, {#2} ]\!]}

\newcommand{\tuu}[0]{\text{T}_{\text{UU}}}
\newcommand{\trf}[0]{\text{T}_{\text{R}\triangleright}}
\newcommand{\tur}[0]{\text{T}_{\text{UR}}}
\newcommand{\tfr}[0]{\text{T}_{\triangleright\text{R}}}
\newcommand{\trr}[0]{\text{T}_{\text{RR}}}
\newcommand{\tsk}[0]{\text{T}_{\text{Sk}}}

\newcommandx{\yaHelper}[2][1=\empty]{%
\ifthenelse{\equal{#1}{\empty}}%
  { \ensuremath{ \scriptstyle{ #2 } } } 
  { \raisebox{ #1 }[0pt][0pt]{ \ensuremath{ \scriptstyle{ #2 } } } }  
}   

\newcommandx{\yrightarrow}[4][1=\empty, 2=\empty, 4=\empty, usedefault=@]{%
  \ifthenelse{\equal{#2}{\empty}}
  { \xrightarrow{ \protect{ \yaHelper[ #4 ]{ #3 } } } } 
  { \xrightarrow[ \protect{ \yaHelper[ #2 ]{ #1 } } ]{ \protect{ \yaHelper[ #4 ]{ #3 } } } } 
}

\newcommand{\trans}[0]{\rightarrow}
\newcommand{\fTrans}[0]{\yrightarrow[][]{\text{\scriptsize{fail}}}[-1pt]}
\newcommand{\sTrans}[0]{\yrightarrow[]{\text{\scriptsize skip}}[-1pt]}

\DeclareMathAlphabet{\mathpzc}{OT1}{pzc}{m}{it}

\newlength\mylen

\ifdefined\PAPER 
\usepackage{caption} 
\fi

\begin{document}

\ifdefined\PAPER
\bstctlcite{IEEEexample:BSTcontrol}
\fi

\sloppy


\ifdefined\TECHREP  
\title{A Dual Digraph Approach for  Leaderless \\ Atomic Broadcast\\ \Large  (Extended Version)*
\thanks{\textsuperscript{*}CC-BY 4.0. This is the author's version of the work. 
The definitive version is published in the proceedings of the 2019 38th International Symposium on Reliable Distributed Systems (SRDS 2019).
Please refer to that publication when citing \ACplus{}.}}

\author{
\IEEEauthorblockN{Marius Poke}
\IEEEauthorblockA{Faculty of Mechanical Engineering}
\IEEEauthorblockA{Helmut Schmidt University}
\IEEEauthorblockA{marius.poke@hsu-hh.de}
\and
\IEEEauthorblockN{Colin W. Glass}
\IEEEauthorblockA{Faculty of Mechanical Engineering}
\IEEEauthorblockA{Helmut Schmidt University}
\IEEEauthorblockA{glassc@hsu-hh.de}
}

\maketitle
\fi

\ifdefined\PAPER
\title{A Dual Digraph Approach for Leaderless \\ Atomic Broadcast}

\author{
\IEEEauthorblockN{Marius Poke}
\IEEEauthorblockA{\textit{Faculty of Mechanical Engineering} \\ 
\textit{Helmut Schmidt University} \\
marius.poke@hsu-hh.de}
\and
\IEEEauthorblockN{Colin W. Glass}
\IEEEauthorblockA{\textit{Faculty of Mechanical Engineering} \\ 
\textit{Helmut Schmidt University} \\
glassc@hsu-hh.de}
}

\maketitle

\fi

\begin{abstract}
%
Many distributed systems work on a common shared state;
in such systems, distributed agreement is necessary for consistency.
With an increasing number of servers, these systems become more susceptible to single-server failures, increasing the relevance of fault-tolerance. 
Atomic broadcast enables fault-tolerant distributed agreement, yet it is costly to solve.
Most practical algorithms entail linear work per broadcast message.
\AC{}---a leaderless approach---reduces the work, by connecting the servers via 
a sparse resilient overlay network; yet, this resiliency entails redundancy, limiting the reduction of work. 
%
In this paper, we propose \ACplus{}, an atomic broadcast algorithm that lifts this limitation: 
During intervals with no failures, it achieves minimal work by using a redundancy-free overlay network. 
When failures do occur, it automatically recovers by switching to a resilient overlay network.
%
In our performance evaluation of non-failure scenarios, 
\ACplus{} achieves comparable throughput to AllGather---a non-fault-tolerant distributed agreement algorithm---and
outperforms \AC{}, LCR and Libpaxos both in terms of throughput and latency.
Furthermore, our evaluation of failure scenarios shows that \ACplus{}'s expected performance is robust with regard to occasional failures. 
%
%
%
Thus, for realistic use cases, leveraging redundancy-free distributed agreement during intervals with no failures improves performance significantly.
\end{abstract}


\ifdefined\PAPER
%

\fi

%

\section{Introduction}

Many distributed systems work on a common shared state, 
e.g., distributed-ledger systems~\cite{Androulaki:2018:HFD:3190508.3190538} and databases of travel reservation systems~\cite{Unterbrunner2014}.
To guarantee consistency, distributed agreement is necessary---all the servers sharing the state need to agree on the ordering of updates. 
Moreover, the updates are propagated to all servers, which will then apply them sequentially to their states, i.e., active replication~\cite{Correia2010}.
We consider applications for which the state updates are well distributed among the servers and cannot be reduced.

Most distributed agreement algorithms are designed for use cases, where replicating the state is a means to achieve 
high availability~\cite{Lamport:1998:PP:279227.279229, Liskov2012, Junqueira:2011:ZHB:2056308.2056409, Ongaro2014, ringpaxos, Guerraoui:2010:TOT:1813654.1813656}.
\ifdefined\TECHREP
These algorithms are typically used to provide coordination services
to large distributed systems.
Although the nature of these services varies, ranging from configuration management~\cite{Hunt:2010:ZWC:1855840.1855851} to locking~\cite{Burrows:2006:CLS:1298455.1298487}, 
the usage patterns usually are similar~\cite{Ongaro2014, Corbett:2013:SGG:2518037.2491245, Hunt:2010:ZWC:1855840.1855851}:
\fi
\ifdefined\PAPER
These algorithms are typically used to provide coordination services to large distributed systems (e.g.,~\cite{Hunt:2010:ZWC:1855840.1855851, Burrows:2006:CLS:1298455.1298487})
and have usually the following usage pattern~\cite{Ongaro2014, Hunt:2010:ZWC:1855840.1855851, Corbett:2013:SGG:2518037.2491245}:
\fi
The servers of the distributed system coordinate their activities by sending requests to a group consisting of a handful of replicas (i.e., three to five usually suffice~\cite{Corbett:2013:SGG:2518037.2491245});
once a request is executed, a reply is sent back to the server that sent the request.
Yet, this pattern does not apply for use cases, 
where having multiple consistent replicas is a requirement of the application, such as distributed ledgers, and replicating the state across hundreds of servers is not uncommon~\cite{Androulaki:2018:HFD:3190508.3190538}. 
In such cases, the servers of the distributed system and the group of replicas are actually the same. 
Thus, for consistency, every state update must be propagated to all servers, which entails an all-to-all exchange.

A straightforward way to implement an all-to-all exchange is to use one of the dissemination schemes typical for unrooted collectives~\cite{hoefler-moor-collectives}.
\ifdefined\TECHREP
For instance, dissemination through a circular digraph provides minimal work\footnote{In the absence of encryption, we assume that the main work performed by a server is sending and receiving messages.} 
per broadcast message~\cite{Guerraoui:2010:TOT:1813654.1813656}. 
\fi
Yet, such implementations provide no fault tolerance and with an increasing number of servers in the system, failures become more likely. 
\ifdefined\TECHREP
For example, in a non-fault-tolerant implementation of a distributed ledger, a failure may lead to inconsistencies, such as the servers disagreeing on the validity of a transaction.

\fi
Atomic broadcast enables fault-tolerant distributed agreement, yet it is costly to solve. 
To provide total order, most practical algorithms rely either on leader-based approaches 
(e.g.,~\cite{Lamport:1998:PP:279227.279229,Liskov2012,Junqueira:2011:ZHB:2056308.2056409,Ongaro2014,Burrows:2006:CLS:1298455.1298487}), 
which entail heavy workloads on the leader,
or on message timestamps that reflect causal ordering~\cite{Lamport:1978:TCO:359545.359563} (e.g.~\cite{Keidar2000, Guerraoui:2010:TOT:1813654.1813656}),
which means they must contain information on every server.
%
As a result, 
the work required for broadcasting one message
is linear in the number of 
\ifdefined\PAPER
servers (in the absence of encryption, we assume that the main work performed by a server is sending and receiving messages).
\fi
\ifdefined\TECHREP
servers.
\fi
The aforementioned algorithms were designed mainly for providing high availability and thus, they are not well suited for large-scale distributed agreement.


\AC{}~\cite{poke2017allconcur} adopts a leaderless approach, where the servers are connected via a sparse overlay network 
and the size of the messages is constant (as no timestamps, reflecting causal ordering, are required). As a result, the work per broadcast message is sublinear.
Yet, to reliably disseminate messages, the overlay network needs to be resilient. 
This resiliency comes at the cost of redundancy, which introduces a lower bound on the work per broadcast message.


In this paper, we present \ACplus{}, a leaderless concurrent atomic broadcast algorithm that adopts a dual digraph approach,
with the aim of lifting the lower bound on work imposed by \AC{}.
In general, a resilient overlay network is necessary, since the frequency of failures in distributed systems 
makes non-fault-tolerant services unfeasible~\cite{tsubame}.
Yet, for many use cases, intervals with no failures are common enough to motivate a less conservative approach.
Thus, servers in \ACplus{} communicate via two overlay networks described by 
two digraphs---an \emph{unreliable} digraph, with a vertex-connectivity of one, 
and a \emph{reliable} digraph, with a vertex-connectivity larger than the maximum number of 
tolerated failures. 
The unreliable digraph enables minimal-work distributed agreement during intervals with no failures---every 
server both receives and sends every broadcast message at most once.
When failures do occur, \ACplus{} falls back to the reliable digraph.
The fault tolerance is given by the reliable digraph's vertex-connectivity and can be adapted to system-specific requirements.
Thus, similarly to \AC{}, \ACplus{} trades off reliability against performance. 

We designed \ACplus{} as a round-based algorithm. The dual digraph approach entails two modes, which means two types of rounds---reliable and unreliable~(\cref{sec:round_based}).
However, due to rollbacks caused by failures, the sequence of rounds is not ordered. 
%
Therefore, to keep better track of the state of each server, we introduce epochs, which, in a nutshell, consist each of one reliable round followed by a sequence of zero or more unreliable rounds. 
This enables us to model the execution of \ACplus{} as an ordered sequence of states, that can be reached through a well-defined set of transitions~(\cref{sec:state_machine}).
%

Our evaluation of \ACplus{}'s performance is based on a discrete-event simulator~\cite{Varga:2008:OOS:1416222.1416290}~(\cref{sec:evaluation}).
When no failures occur, \ACplus{} achieves between $79\%$ and $100\%$ of the throughput of AllGather~\cite{mpi-3.1}, a non-fault-tolerant distributed agreement algorithm. 
When comparing with other fault-tolerant algorithms, \ACplus{} outperforms them in terms of both throughput and latency: 
It achieves up to $6.5\times$ higher throughput and up to $3.5\times$ lower latency than \AC{}; 
up to $6.3\times$ higher throughput and up to $3.2\times$ lower latency than LCR~\cite{Guerraoui:2010:TOT:1813654.1813656}; 
and up to $318\times$ higher throughput and up to $158\times$ lower latency than Libpaxos~\cite{libpaxos}.
Moreover, our evaluation of failure scenarios shows that \ACplus{}'s expected performance is robust with regard to occasional failures. 
For example, if every time between successive failures a sequence of nine rounds completes, then \ACplus{} has 
up to $3.5\times$ higher throughput and up to $1.9\times$ lower latency than \AC{}.

In summary, our work makes three key contributions:
\begin{itemize}
\setlength{\itemsep}{0pt}
  \item the design of \ACplus{}, a leaderless concurrent atomic broadcast algorithm that 
  leverages redundancy-free agreement during intervals with no failures~(\cref{sec:acp_design});
  \item an informal proof of \ACplus{}'s correctness~(\cref{sec:corectness_proof});
  \item an evaluation of \ACplus{}'s performance~(\cref{sec:evaluation}).
\end{itemize}

\section{System model}
\label{sec:sys_model}

We consider $n$ servers that are subject to a maximum of $f$ crash failures. 
The servers communicate through messages according to an overlay network 
described by a digraph---server $p$ sends messages to 
server $q$ if there is a directed edge $(p,q)$ in the digraph.
\ifdefined\PAPER
The edges describe FIFO reliable channels~\cite{Pedone:2003:OAB:795635.795644},
which are easily implemented in practice with sequence numbers and retransmissions.
\fi
\ifdefined\TECHREP
The edges describe FIFO reliable channels, i.e., we assume the following properties~\cite{Pedone:2003:OAB:795635.795644}: 
\begin{itemize}
\setlength{\itemsep}{0pt}
  \item (no creation) if $q$ receives a message $m$ from $p$, then $p$ sent $m$ to $q$;
  \item (no duplication) $q$ receives every message at most once; 
  \item (no loss) if $p$ sends $m$ to $q$, and $q$ is non-faulty, then $q$ eventually receives $m$;
  \item (FIFO order) if $p$ sends $m$ to $q$ before sending $m^\prime$ to $q$, then $q$ will not receive $m^\prime$ before receiving $m$.
\end{itemize}
In practice, FIFO reliable channels are easy to implement, using sequence numbers and retransmissions.
\fi
We consider two modes: 
(1) an unreliable mode, described by an unreliable digraph $G_U$ with vertex-connectivity $\kappa(G_U) = 1$;
and (2) a reliable mode, described by a reliable digraph $G_R$ with vertex-connectivity $\kappa(G_R) > f$. 
Henceforth, we use the terms vertex and server interchangeably.

\textbf{Atomic broadcast.}
Atomic broadcast is a communication primitive that ensures all messages are delivered in the same order by all non-faulty servers. 
We distinguish between receiving and delivering a message~\cite{Birman:1991:LCA:128738.128742}, i.e., servers can delay the delivery of received messages.
\ifdefined\TECHREP
To formally define atomic broadcast, let $m$ be a message (uniquely identified);
let \tobroadcast{m} and \todeliver{m} be communication primitives for broadcasting and delivering messages atomically;
and let \sender{m} be the server that A-broadcasts~$m$.
Then, any \emph{non-uniform} atomic broadcast
algorithm must satisfy four properties~\cite{Chandra:1996:UFD:226643.226647,Hadzilacos:1994:MAF:866693,Defago:2004:TOB:1041680.1041682}:
\begin{itemize}
\setlength{\itemsep}{0pt}
  \item (Validity) If a non-faulty server A-broadcasts $m$, then it eventually A-delivers~$m$.
  \item (Agreement) If a non-faulty server A-delivers $m$, then all non-faulty servers eventually 
  A-deliver~$m$.
  \item (Integrity) For any message $m$, every non-faulty server A-delivers $m$ at most once, and only if 
  $m$ was previously A-broadcast by \sender{m}.
  \item (Total order) If two non-faulty servers $p_i$ and $p_j$ A-deliver messages $m_1$ and $m_2$, 
  then $p_i$ A-delivers $m_1$ before $m_2$, if and only if $p_j$ A-delivers $m_1$ before~$m_2$. 
\end{itemize}
Integrity and total order are safety properties---they must hold at any point during execution.
Validity and agreement are liveness property---they must eventually hold (to ensure progress).
\fi
\ifdefined\PAPER
Any non-uniform atomic broadcast algorithm must satisfy four properties:
validity, agreement, integrity, and total order~\cite{Chandra:1996:UFD:226643.226647,Defago:2004:TOB:1041680.1041682}.
Integrity and total order are safety properties; validity and agreement are liveness property.
We use \tobroadcast{m} and \todeliver{m} to denote the communication primitives for broadcasting and delivering messages atomically.
\fi
If only validity, agreement and integrity hold, the broadcast is \emph{reliable}; we use \rbroadcast{m} and \rdeliver{m} 
to denote the communication primitives of reliable broadcast. 
Also, we assume a message can be A-broadcast multiple times; yet, it can be A-delivered only once. 

The agreement and the total order properties are non-uniform---they apply only to non-faulty servers. 
\ifdefined\TECHREP
As a consequence, it is not necessary for all non-faulty servers to A-deliver the messages A-delivered by faulty servers.
This may lead to inconsistencies in some applications, such as a persistent database 
(i.e., as a reaction of A-delivering a message, a faulty server issues a write on disk). 
Uniformity can facilitate the development of such applications, yet, it comes at a cost. 
In general, applications can protect themselves from non-uniformity (e.g., any persistent update that is contingent on the A-delivery of $m$, 
must wait for $m$ to be A-delivered by at least $f+1$ servers).
\fi
In Appendix~\ref{app:uniformity} we discuss the modifications required by uniformity.

\textbf{Failure detection.}
Failure detectors (FD) have two properties---\emph{completeness}, i.e., all failures are eventually detected,
and \emph{accuracy}, i.e., no server is falsely suspected to have failed~\cite{Chandra:1996:UFD:226643.226647}.
We consider a heartbeat-based FD: servers send heartbeat messages to their successors in $G_R$;
once a server fails, its successors detect the lack of heartbeat messages and R-broadcast a failure notification to the other servers.
Clearly, such an FD guarantees completeness. 
For deployments, where some assumptions of synchrony~\cite{Dwork:1988:CPP:42282.42283} can be practical (e.g., within a single datacenter), 
accuracy can also be guaranteed, i.e., the FD is perfect (denoted by $\mathcal{P}$)~\cite{Chandra:1996:UFD:226643.226647}.
However, to widen \ACplus{}'s applicability (e.g., deployments over multiple datacenters), we assume an \emph{eventually perfect} FD 
(denoted by $\Diamond\mathcal{P}$)~\cite{Chandra:1996:UFD:226643.226647}, which guarantees accuracy only eventually.
As a result, \ACplus{} has many of the properties of other atomic broadcast algorithms, such as Paxos~\cite{Lamport:1998:PP:279227.279229}:
it guarantees safety even under asynchronous assumptions and liveness under weak synchronous assumptions.

\section{A dual digraph approach}
\label{sec:dual_digraph_approach}

\ACplus{} switches between two modes---unreliable and reliable. 
During intervals with no failures, \ACplus{} uses the unreliable mode, which enables minimal-work distributed agreement.
When failures occur, \ACplus{} automatically switches to the reliable mode, that uses the early termination mechanism of \AC{}~\cite{poke2017allconcur}.

In this section, we first give an overview of \AC{}'s early termination mechanism~(\cref{sec:allconcur_overview}). 
Then, we describe the two modes of \ACplus{}~(\cref{sec:round_based}) and specify the possible transitions from one mode to another~(\cref{sec:state_machine}).
Afterwards, we discuss the conditions necessary to A-deliver messages~(\cref{sec:adeliv_msg}).
Finally, we reason about the possible concurrent states \ACplus{} servers can be in~(\cref{sec:conc_states}).


\subsection{Overview of \AC{}'s early termination}
\label{sec:allconcur_overview}

\ifdefined\TECHREP
\begin{figure*}[!tp]
\captionsetup[subfigure]{justification=centering}
\centering
\subcaptionbox{$G_S(9,3)$\label{fig:gs_n9_d3}} {
\includegraphics[width=.18\textwidth]{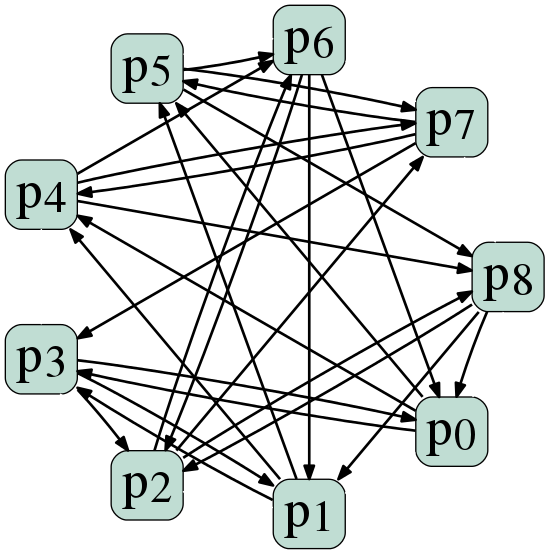}
}
\qquad
\subcaptionbox{\label{fig:tracking_digraph}} {
\input{figures/tracking_digraph_small.tex}
}
\caption{(a) Overlay network connecting nine servers. 
(b) Server $p_6$ tracking $p_0$'s message $m_0$ within a $G_S(9,3)$ digraph. Dotted red nodes indicate failed servers.
Dashed gray nodes indicate servers from which $p_6$ received failures notifications (i.e, dashed red edges). 
Solid green nodes indicate servers suspected to have $m_0$~\cite{allconcur_tla}.}
\label{fig:tracking_in_gs}
\end{figure*}
\fi

\AC{} is a round-based algorithm that, through an early termination mechanism, allows every server to A-deliver 
a round as soon as it knows it has all the respective messages any other non-faulty server has~\cite{poke2017allconcur}.
As a result, it avoids waiting unnecessarily for the worst-case bound of 
\ifdefined\TECHREP
$f+1$ communication steps~\cite{lamport1982byzantine,Aguilera:1998:SBP:866987} 
\fi
\ifdefined\PAPER
$f+1$ communication steps~\cite{lamport1982byzantine} 
\fi
(or more generally, $f+D_f(G)$ communication steps, with $D_f(G)$ being $G$'s fault diameter~\cite{Krishnamoorthy:1987:FDI:35064.36256}, 
i.e., $G$'s maximum diameter after removing any $f$ vertices).
The early termination mechanism uses failure notifications to track A-broadcast messages.
For every A-broadcast message $m$, every server $p$ maintains a \emph{tracking digraph} $g$, i.e., a representation of $p$'s suspicion of $m$'s whereabouts. 
Thus, $g$'s vertices indicate the servers suspected of having $m$ and $g$'s edges indicate the paths on which $m$ is suspected of having been transmitted. 
The tracking stops when $p$ either receives $m$ or suspects only failed servers of having $m$. 
To safely complete a round, $p$ must finish tracking all messages A-broadcast in that round. 

\ifdefined\TECHREP
To illustrate the early termination mechanism we consider the example from the formal specification of \AC{}~\cite{allconcur_tla}:
$n=9$ servers $(p_0,\ldots,p_8)$ connected through a $G_S(n,d)$ digraph~\cite{Soneoka:1996:DDC:227095.227101} with both the degree and the vertex-connectivity equal to three, 
i.e., optimally connected~\cite{Meyer:1988:FFG:47054.47067} (see Figure~\ref{fig:gs_n9_d3}).
We consider the following scenario: $p_0$ sends $m_0$ only to $p_5$ and then it fails; $p_5$ receives $m_0$, but it fails before sending it further.
Thus, $m_0$ is lost; yet, $p_6$ (for example) is not aware of this and to avoid waiting unnecessarily for the worst case bound, it tracks $m_0$ (see Figure~\ref{fig:tracking_digraph}). 
Let $p_4$ be the first to detect $p_0$'s failure and consequently, R-broadcasting a notification. 
When $p_6$ receives this notification, it first marks $p_0$ as failed 
and then starts suspecting that, before failing, $p_0$ sent $m_0$ to its successors, i.e., $p_3$ and $p_5$. 
Note though that $p_6$ does not suspect $p_4$. Had $p_4$ received $m_0$ from $p_0$, then  $p_4$ would have relayed $m_0$, 
which would therefore have arrived to $p_6$ before the subsequent failure notification (due to assumption of FIFO reliable channels).
This procedure repeats for the other failure notifications received, e.g., $p_6$ detects $p_5$'s failure directly, $p_8$ detects $p_5$'s failure, $p_3$ detects $p_0$'s failure, 
and $p_7$ detects $p_5$'s failure. 
In the end, $p_6$ suspects only $p_0$ and $p_5$ of  having $m_0$ and, since both are faulty, $p_6$ stops tracking $m_0$. 
\fi

The early termination mechanism relies on the following proposition~\cite{poke2017allconcur,allconcur_tla}:
\ifdefined\TECHREP
\begin{prop}
\label{prop:early_termination}
Let $p_i$, $p_j$ and $p_k$ be three servers. Then, $p_i$ receiving a notification of $p_j$'s failure sent by $p_k$
indicates that $p_i$ has all the messages $p_k$ received directly from $p_j$.
\end{prop}
\fi
\ifdefined\PAPER
Let $p_i$, $p_j$ and $p_k$ be three servers, then $p_i$ receiving a notification of $p_j$'s failure sent by $p_k$
indicates that $p_i$ has all the messages $p_k$ received directly from $p_j$.
\fi
For this proposition to hold for $\Diamond\mathcal{P}$, 
a server must ignore any subsequent messages (except failure notifications) it receives from a predecessors it has suspected of having failed.
This is equivalent to removing an edge from $G$ and thus, may lead to a disconnected digraph even if $f<\kappa(G)$ (since $\Diamond\mathcal{P}$ may be inaccurate).
To avoid inconsistencies, \AC{} uses a \emph{primary partition} membership approach~\cite{Defago:2004:TOB:1041680.1041682}:
Only servers from the \emph{surviving partition}---a strongly connected component that contains at least a majority of 
servers---are allowed to make progress and A-deliver messages.
To decide whether they are part of the surviving partition, the servers use a \emph{forward-backward mechanism} 
(based on Kosaraju's algorithm to find strongly connected components~\cite{Aho:1983:DSA:577958}):
Once a round completes, before A-delivering its messages, every server R-broadcasts two control messages---a forward message in $G$, and a backward message in $G$'s transpose.
Then, each server A-delivers the round's messages only if it receives both forward and backward messages from $\lceil \frac{n-1}{2} \rceil$ other servers. 

\ifdefined\TECHREP
Allowing only servers from the surviving partition to A-deliver guarantees safety.
Yet, the servers that are not part of the surviving partition cannot make progress and thus, they must be removed (i.e., process controlled crash).
This also includes scenarios where there is no surviving partition, which entails that the algorithm must be stopped and started again from the latest A-delivered round. 
\fi

As an optimization, if assuming $\mathcal{P}$ is practical, the forward-backward mechanism is not necessary, 
which entails faster rounds.

\subsection{Round-based algorithm}
\label{sec:round_based}

\ACplus{} is a round-based algorithm that distinguishes between \emph{unreliable} and \emph{reliable} rounds.
%
Every round is described by its round number~$r$. 
In round~$r$, every server A-broadcasts a (possibly empty) message and collects (in a set) all the messages received for this round (including its own).
The goal is for all non-faulty servers to eventually agree on a common set of messages; we refer to this as the \emph{set agreement} property. 
Then, all non-faulty servers A-deliver the messages in the common set in a deterministic order. 
For brevity, we say a round is A-delivered if its messages are A-delivered.


\textbf{Unreliable rounds.}
Unreliable rounds enable minimal-work distributed agreement while no failures occur. The overlay network is described by $G_U$. 
A server completes an unreliable round once it has received a message from every server (A-broadcast in that round). 
Since $G_U$ is connected, the completion of an unreliable round is guaranteed under the condition of no failures.
To ensure set agreement, the completion of an unreliable round does not directly lead to it being A-delivered.
Yet, completing two successive unreliable rounds guarantees the first one can be A-delivered (see Section~\ref{sec:adeliv_msg} for details).

\textbf{Reliable rounds.}
%
Receiving a failure notification in an unreliable round triggers both a rollback to the latest A-delivered round and a switch to the reliable mode.
Thus, the first round that has not yet been A-delivered (i.e., succeeding the latest A-delivered round) is rerun reliably. 
To complete a reliable round, every server uses \AC{}'s early termination mechanism~(\cref{sec:allconcur_overview}).
In addition to completion, early termination guarantees set agreement~\cite{poke2017allconcur,allconcur_tla}. 
Thus, once a server completes a reliable round, it can safely A-deliver it~(\cref{sec:adeliv_msg}).
\ifdefined\TECHREP
In summary, a rollback entails (reliably) rerunning unreliable rounds, which have not been previously A-delivered. 
For validity to hold, it is necessary for the same messages to be A-broadcast when a round is rerun.
In practice, this requirement can be dropped~(\cref{sec:evaluation}).
\fi

Once a reliable round completes, all servers for which no messages were A-delivered are removed.
Thus, all non-faulty servers have a consistent view of the system 
\ifdefined\TECHREP
(similar to group membership services~\cite{Chandra:1996:IGM:248052.248120,Cristian1991}). 
\fi
\ifdefined\PAPER
(similar to group membership services~\cite{Chandra:1996:IGM:248052.248120}). 
\fi
After removing servers from the system, every non-faulty server needs to update both $G_U$, in order for it to be connected, 
and the set $F$ of received failure notifications; also, the servers may choose to update $G_R$.
We defer the discussion on updating the digraphs to 
\ifdefined\TECHREP
Sections~\ref{sec:algo_description} and~\ref{sec:update_Gf}, respectively. 
\fi
\ifdefined\PAPER
Section~\ref{sec:acp_design}. 
\fi
Updating~$F$ entails removing all the \emph{invalid failure notifications}, 
i.e., notifications that are targeting, or were detected by, removed servers. 
Invalid notifications are no longer required for further message tracking, 
since removed servers cannot be suspected of having any messages. 

\subsection{State machine approach}
\label{sec:state_machine}
It is common to describe a distributed algorithm as a state machine: Every server's behavior is defined by an ordered sequence of states.
As long as no failures occur, we can model \ACplus{}'s execution as an ordered sequence of unreliable rounds with strictly increasing round numbers. 
\ifdefined\TECHREP
This sequence consists of a sub-sequence of A-delivered rounds followed (potentially) by a round not yet A-delivered~(\cref{sec:adeliv_msg}) and by the ongoing round.
\fi
The occurrence of a failure leads to a rollback. This breaks the order of the sequence, i.e., the round numbers are no longer strictly increasing.
%
To enforce a strict order, we introduce \emph{epochs}. Every epoch starts with a reliable round and contains only one completed reliable round and at most one sequence of unreliable rounds.
An epoch is described by an epoch number~$e$, which corresponds to the number of reliable rounds completed so far, plus the ongoing round, if it is reliable.

The state of each server is defined by the epoch number, the round number, the round type and, for unreliable rounds, the type of the previous round.
We denote states with unreliable rounds by $\acState{e}{r}$ and 
states with reliable rounds by $\acRState{e}{r}$. 
Also, when necessary, we use~$\triangleright$ to indicate the first unreliable round following a reliable round. 
Initially, all servers  are in state $\acRState{1}{0}$, essentially assuming a reliable round $0$ has already been completed, without any messages being A-broadcast. 
State $\acRState{1}{0}$ is required by the definition of an epoch. 

Servers can move from one state to another; 
we define three types of state transitions.
First, \emph{no-fail transitions}, denoted by~$\trans$, 
move servers to the next unreliable round without increasing the epoch. 
We identify two no-fail transitions:
\ifdefined\TECHREP
\begin{align*}
 &\text{(}\tuu{}\text{)} \quad \acState{e}{r} \trans \acState{e}{r+1}; &
 &\text{(}\trf{}\text{)} \quad \acRState{e}{r} \trans \acFState{e}{r+1}.
\end{align*}
\fi
\ifdefined\PAPER
\begin{align*}
 &\text{(}\tuu{}\text{)} \,\, \acState{e}{r} \trans \acState{e}{r\!+\!1}; &
 &\text{(}\trf{}\text{)} \,\, \acRState{e}{r} \trans \acFState{e}{r\!+\!1}.
\end{align*}
\fi
$\tuu{}$ continues a sequence of unreliable rounds, while $\trf{}$, 
in the absence of failure notifications, starts a sequence of unreliable rounds.
\ifdefined\TECHREP
Note that the initial state $\acRState{1}{0}$ is always followed by a $\trf{}$ transition.
\fi

Second, \emph{fail transitions}, denoted by~$\fTrans$, move servers to a reliable round while increasing the epoch. 
Fail transitions are caused by failure notifications: In unreliable rounds, any failure notification immediately 
triggers a fail transition. In reliable rounds,  remaining failure notifications that are still valid at the end of the round result in a fail transition to the next round. 
Overall, we identify three fail transitions:
\ifdefined\TECHREP
\begin{align*}
 &\text{(}\tur{}\text{)} \quad \acState{e}{r} \fTrans \acRState{e+1}{r-1}; &
 &\text{(}\tfr{}\text{)} \quad \acFState{e}{r} \fTrans \acRState{e+1}{r}; &
 &\text{(}\trr{}\text{)} \quad \acRState{e}{r} \fTrans \acRState{e+1}{r+1}. &
\end{align*}
\fi
\ifdefined\PAPER
\begin{align*}
 &\text{(}\tur{}\text{)} \,\, \acState{e}{r} \fTrans \acRState{e\!+\!1}{r\!-\!1}; &
 &\text{(}\tfr{}\text{)} \,\, \acFState{e}{r} \fTrans \acRState{e\!+\!1}{r}; \\
 &\text{(}\trr{}\text{)} \,\, \acRState{e}{r} \fTrans \acRState{e\!+\!1}{r\!+\!1}. &
\end{align*}
\fi
$\tur{}$ and $\tfr{}$ both interrupt the current unreliable round and rollback to the latest A-delivered round.
The difference is $\tur{}$ is preceded by $\tuu{}$, while $\tfr{}$ is preceded by $\trf{}$, and thus, the latest A-delivered round differs~(\cref{sec:adeliv_msg}).
$\trr{}$ continues a sequence of reliable rounds, due to failure notifications in~$F$ that remain valid at the end of the round and thus, cannot be removed~(\cref{sec:round_based}).

Third, \emph{skip transitions}, denoted by~$\sTrans$, 
move servers to a reliable round without increasing the epoch, i.e.,
\ifdefined\TECHREP
\begin{align*}
 &\text{(}\tsk{}\text{)} \quad \acRState{e}{r} \sTrans \acRState{e}{r+1}.
\end{align*}
\fi
\ifdefined\PAPER
\begin{align*}
 &\text{(}\tsk{}\text{)} \,\, \acRState{e}{r} \sTrans \acRState{e}{r\!+\!1}.
\end{align*}
\fi
A server $p_i$ performs a skip transition if it receives, in a reliable round, 
a message A-broadcast by a server $p_j$ in the same epoch, but in a subsequent reliable round. 
This happens only if $p_j$ had one more $\tuu{}$ transition than $p_i$ 
before receiving the failure notification that triggered the fail transition to the current epoch.
Figure~\ref{fig:skip_trans} illustrates the skip transition of $p_i$ after receiving $m_j^{e,r+1}$, 
the message A-broadcast by $p_j$ while in state $\acRState{e}{r+1}$.
This message indicates (to $p_i$) that $p_j$ rolled back to round $r$, the latest A-delivered round; 
thus, it is safe (for $p_i$) to also A-deliver $\acState{e}{r}$  and \emph{skip} to $\acRState{e}{r+1}$.
Note that $\acRState{e}{r}$ is not completed (we say a state is completed or A-delivered, 
if its round is completed or A-delivered, respectively).

\begin{figure}[!tp]
\centering
\ifdefined\PAPER
\scalebox{0.70}{\input{figures/skip_transition.tex}}
\fi
\ifdefined\TECHREP
\scalebox{0.80}{\input{figures/skip_transition.tex}}
\fi
\caption{Skip transition (i.e., $\tsk{}$). Empty rectangles indicate unreliable rounds; filled (gray) rectangles indicate reliable rounds.
Rectangles with solid edges indicate completed rounds, while dashed edges indicate interrupted rounds.}
\label{fig:skip_trans}
\end{figure}


\ifdefined\TECHREP
In summary, $\tuu{}$ and $\trf{}$ lead to unreliable rounds, while $\tur{}$, $\tfr{}$, $\trr{}$ and $\tsk{}$ lead to reliable rounds. 
\fi

\subsection{A-delivering messages}
\label{sec:adeliv_msg}

In Figure~\ref{fig:skip_trans}, $p_j$ already A-delivered round $r$ when it receives 
the failure notification that triggers the transition from $\acState{e-1}{r+2}$ to $\acRState{e}{r+1}$.
To explain the intuition behind $p_j$ A-delivering $r$, we first introduce the following proposition:
\begin{prop}
\label{prop:one_end_all_start}
A necessary condition for a server to complete a state is for all non-faulty servers to start the state.
\end{prop}
\begin{IEEEproof}
For a server to complete a state (with either an unreliable or a reliable round), 
it must receive from every non-faulty server a message A-broadcast in that state~(\cref{sec:round_based}).
\end{IEEEproof}

Since $p_j$ started $\acState{e-1}{r+2}$, it completed $\acState{e-1}{r+1}$ and thus, 
every non-faulty server started $\acState{e-1}{r+1}$ (cf. Proposition~\ref{prop:one_end_all_start}).
This entails that every non-faulty server completed $\acState{e-1}{r}$.
Thus, $p_j$ knows round $r$ was safely completed by all non-faulty servers and it can A-deliver it.
Moreover, in the example, the message  $m_j^{e,r+1}$ A-broadcast by $p_j$ while in state $\acRState{e}{r+1}$ 
carries (implicitly) the information that all non-faulty servers completed $\acState{e-1}{r}$. 
Thus, upon receiving $m_j^{e,r+1}$, $p_i$ can also A-deliver $\acState{e-1}{r}$. 

Thus, in \ACplus{}, a server can A-deliver an unreliable round~$r$ in two ways: 
(1) it completes the subsequent unreliable round $r+1$;
or (2) after it interrupts the subsequent round $r+1$ due to a failure notification, 
it receives a message A-broadcast by another server in the same epoch, but in a subsequent reliable round, i.e., a skip transition.
Reliable rounds use early termination; thus, they can be A-delivered directly after completion (i.e., before either $\trf{}$ or $\trr{}$).

\subsection{Concurrent states}
\label{sec:conc_states}

The necessary condition stated in Proposition~\ref{prop:one_end_all_start} enables us to reason about the possible concurrent states, 
i.e., the states a non-faulty server $p_i$ can be, given that a non-faulty server $p_j$ is either in $\acState{e}{r}$ or in $\acRState{e}{r}$ 
(see Appendix~\ref{app:conc_states} for a description of all the possible concurrent states).
Knowing the possible concurrent states enables us to deduce a set of properties that aid the design of \ACplus{}:
Theorem~\ref{prop:unique_states} asserts the uniqueness of a state, i.e., a state is uniquely identified by only its epoch and round; 
Theorems~\ref{prop:unrel_msg}--\ref{prop:rel_msg_same_epoch} specify what messages a non-faulty server can receive. 
The informal proofs of these theorems are available in Appendix~\ref{app:conc_proofs}. 

\begin{restatable}{theorem}{propUniqueState}
\label{prop:unique_states}
Let $p_i$ and $p_j$ be two non-faulty servers, both in epoch $e$ and round $r$. 
Then, both are either in~$\acState{e}{r}$ or in~$\acRState{e}{r}$. 
\end{restatable}

\begin{restatable}{theorem}{propUnrelMsg}
\label{prop:unrel_msg}
Let $p_i$ and $p_j$ be two non-faulty servers. 
Let $m_j^{(e,r)}$ be the message A-broadcast by $p_j$ while in $\acState{e}{r}$ 
and received by $p_i$ in epoch $\tilde{e}$ and round $\tilde{r}$.
Then, $\tilde{e} \geq e$. Also, $(\tilde{e} = e \wedge \tilde{r} < r) \Rightarrow \tilde{r}=r-1$.
\end{restatable}

\begin{restatable}{theorem}{propRelMsg}
\label{prop:rel_msg} 
Let $p_i$ and $p_j$ be two non-faulty servers. 
Let $m_j^{(e,r)}$ be the message A-broadcast by $p_j$ while in $\acRState{e}{r}$ 
and received by $p_i$ in epoch $\tilde{e}$ and round $\tilde{r}$.
Then, $\tilde{e} < e \Rightarrow (\tilde{e} = e-1 \wedge \acRState{\tilde{e}}{\tilde{r}} \wedge \tilde{r} = r-1)$.
\end{restatable}

\begin{restatable}{theorem}{propRelMsgSameEpoch}
\label{prop:rel_msg_same_epoch}
Let $p_i$ and $p_j$ be two non-faulty servers. Let $m_j^{(e,r)}$ be a message 
A-broadcast by $p_j$ while in $\acRState{e}{r}$. 
If $p_i$ receives $m_j^{(e,r)}$ in epoch $\tilde{e}=e$ and round $\tilde{r} \leq r$, 
then it is in either $\acRState{e}{r-1}$ or~$\acRState{e}{r}$.
\end{restatable}

\section{The design of \ACplus{}}
\label{sec:acp_design}

In this section, we outline the design of \ACplus{} through a concise event-based description~(\cref{sec:algo_description}).
For a more detailed description, including pseudocode, see Appendix~\ref{app:design_details}.
Moreover, we discuss the impact eventual accuracy has on \ACplus{}~(\cref{sec:eventual_accuracy}) and we provide an informal proof of correctness~(\cref{sec:corectness_proof}).
\ifdefined\TECHREP
Initially, we assume no more than $f$ failures during the whole deployment, i.e., $G_R$ cannot become disconnected. 
Later, we discuss how to update $G_R$ in order to maintain reliability despite failures having occurred~(\cref{sec:update_Gf}).
\fi
\ifdefined\PAPER
We assume no more than $f$ failures during the whole deployment, i.e., $G_R$ cannot become 
disconnected\footnote{In order to maintain reliability despite failures having occurred, $G_R$ must be periodically updated.
Due to space limitations, we defer the discussion on how to update $G_R$ to the extended technical report~\cite{poke2017adual_tr}.}.

\fi

\subsection{Algorithm description}
\label{sec:algo_description}


\ifdefined\TECHREP
\begin{table*}[!tp]
\centering
\begin{tabular}{ c | l | l | l }
  \cmidrule[1.5pt](){1-4}
  
  \textbf{\#} & \textbf{State} & \textbf{Event} & \textbf{Actions} \\
  \hline  
  \rowcolor{DarkGray}  
  1 & $\acState{\tilde{e}}{\tilde{r}}$ or $\acFState{\tilde{e}}{\tilde{r}}$ & recv. $m_{j,\acState{\tilde{e}}{\tilde{r}+1}}$ & postpone sending and delivery for $\acState{\tilde{e}}{\tilde{r}+1}$ \\ 
  
  2 &$\acState{\tilde{e}}{\tilde{r}}$ or $\acFState{\tilde{e}}{\tilde{r}}$ & recv. $m_{j,\acState{\tilde{e}}{\tilde{r}}}$ &  
  \begin{tabular}{@{}l@{}} 
    (1) send $m_{j,\acState{\tilde{e}}{\tilde{r}}}$ further (via $G_U$) \\ 
    (2) A-broadcast $m_{i,\acState{\tilde{e}}{\tilde{r}}}$ (if not done already) \\ 
    (3) try to complete $\acState{\tilde{e}}{\tilde{r}}$ 
  \end{tabular} \\

  
  \rowcolor{DarkGray} 
  3 & $\acState{\tilde{e}}{\tilde{r}}$ & recv. (valid) $\mathit{fn}_{j, k}$ & move to $\acRState{\tilde{e}+1}{\tilde{r}-1}$ and re-handle $\mathit{fn}_{j, k}$ (see \#9) \\
  
  4 & $\acFState{\tilde{e}}{\tilde{r}}$ & recv. (valid) $\mathit{fn}_{j, k}$ & move to $\acRState{\tilde{e}+1}{\tilde{r}}$ and re-handle $\mathit{fn}_{j, k}$ (see \#9) \\
  
  \rowcolor{DarkGray} 
  5 & $\acRState{\tilde{e}}{\tilde{r}}$ & recv. $m_{j,\acState{\tilde{e}}{\tilde{r}+1}}$ & postpone sending and delivery for $\acState{\tilde{e}}{\tilde{r}+1}$ \\

  6 & $\acRState{\tilde{e}}{\tilde{r}}$ & recv. $m_{j,\acRState{\tilde{e}+1}{\tilde{r}+1}}$ & 
    \begin{tabular}{@{}l@{}} 
    (1) send $m_{j,\acRState{\tilde{e}+1}{\tilde{r}+1}}$ further (via $G_R$) \\ 
    (2) postpone delivery for $\acRState{\tilde{e}+1}{\tilde{r}+1}$
  \end{tabular} \\ 

  \rowcolor{DarkGray} 
  7 & $\acRState{\tilde{e}}{\tilde{r}}$ & recv. $m_{j,\acRState{\tilde{e}}{\tilde{r}+1}}$ & 
    \begin{tabular}{@{}l@{}} 
    (1) A-deliver $\acState{\tilde{e}-1}{\tilde{r}}$ \\ 
    (2) move to $\acRState{\tilde{e}}{\tilde{r}+1}$ and re-handle $m_{j,\acRState{\tilde{e}}{\tilde{r}+1}}$ (see \#8)
  \end{tabular} \\ 
  
  8 & $\acRState{\tilde{e}}{\tilde{r}}$ & recv. $m_{j,\acRState{\tilde{e}}{\tilde{r}}}$ &  
  \begin{tabular}{@{}l@{}} 
    (1) send $m_{j,\acRState{\tilde{e}}{\tilde{r}}}$ further (via $G_R$) \\ 
    (2) A-broadcast $m_{i,\acRState{\tilde{e}}{\tilde{r}}}$ (if not done already) \\ 
    (3) try to complete $\acRState{\tilde{e}}{\tilde{r}}$ 
  \end{tabular} \\ 
  
  
  \rowcolor{DarkGray}
  9 & $\acRState{\tilde{e}}{\tilde{r}}$ & recv. (valid) $\mathit{fn}_{j, k}$  & 
    \begin{tabular}{@{}l@{}} 
    (1) send $\mathit{fn}_{j, k}$ further (via $G_R$) \\ 
    (2) update tracking digraphs \\ 
    (3) try to complete $\acRState{\tilde{e}}{\tilde{r}}$ 
  \end{tabular} \\
  
  \cmidrule[1.5pt](){1-4} 
\end{tabular}
  \caption{The actions performed by a server $p_i$ when different events occur while in epoch $\tilde{e}$ and round $\tilde{r}$. 
  $m_{j,\acState{e}{r}}$ denotes a message sent by $p_j$ while in $\acState{e}{r}$, 
  while  $\mathit{fn}_{j, k}$ denotes a notification sent by $p_k$ indicating $p_j$'s failure.
  $p_i$ drops all other messages as well as invalid failure notifications received.}
\label{tab:event_action}
\end{table*}
\fi

Table~\ref{tab:event_action} 
\ifdefined\PAPER
(see Appendix~\ref{app:design_details})
\fi
summarizes the actions performed by a server $p_i$ when different events occur. 
We assume $p_i$ is in epoch $\tilde{e}$ and round $\tilde{r}$. We distinguish between the three states described in Section~\ref{sec:state_machine}: 
$\acFState{\tilde{e}}{\tilde{r}}$; $\acState{\tilde{e}}{\tilde{r}}$; and $\acRState{\tilde{e}}{\tilde{r}}$. We consider also the three events that can occur: 
(1) receiving $m_{j,\acState{e}{r}}$, an unreliable message sent by $p_j$ while in $\acState{e}{r}$; 
(2) receiving $m_{j,\acRState{e}{r}}$, a reliable message sent by $p_j$ while in $\acRState{e}{r}$; 
and (3) receiving $\mathit{fn}_{j, k}$, a notification sent by $p_k$ indicating $p_j$'s failure. 
Note that messages are uniquely identified by the tuple (\emph{source id}, \emph{epoch number}, \emph{round number}, \emph{round type}), 
while failure notifications by the tuple (\emph{target id}, \emph{owner id}). 

\subsubsection{Handling unreliable messages}
If $p_i$ receives in any state an unreliable message $m_{j,\acState{e}{r}}$, 
then it cannot have been sent from a subsequent epoch, i.e., $e \leq \tilde{e}$ (cf. Theorem~\ref{prop:unrel_msg}).
Moreover, unreliable messages from either previous epochs or previous rounds can be dropped. 
Thus, $p_i$ must handle unreliable messages sent only from the current epoch, i.e., $e=\tilde{e}$. 
If $m_{j,\acState{\tilde{e}}{r}}$ was sent from a subsequent round, then $r=\tilde{r}+1$ (cf. Theorem~\ref{prop:unrel_msg});
in this case, $p_i$ postpones both the sending and the delivery of $m_j$ for $\acState{\tilde{e}}{\tilde{r}+1}$ (see \#1 and \#5 in Table~\ref{tab:event_action}).
Otherwise, $m_{j,\acState{\tilde{e}}{r}}$ was sent from the current round, i.e., $r=\tilde{r}$, and thus, $p_i$ can only be in an unreliable round (cf. Theorem~\ref{prop:unique_states}).
Handling $m_{j,\acState{\tilde{e}}{\tilde{r}}}$ while in $\acState{\tilde{e}}{\tilde{r}}$ consists of three operations (see \#2 in Table~\ref{tab:event_action}): 
(1) send $m_j$ further via $G_U$; 
(2) A-broadcast own message (if not done already);
and (3) try to complete the unreliable round $\tilde{r}$.
The necessary and sufficient condition for $p_i$ to complete an unreliable round is to receive a message (sent in that round) from every server.
Once $p_i$ completes $\acState{\tilde{e}}{\tilde{r}}$ (and not $\acFState{\tilde{e}}{\tilde{r}}$), it A-delivers $\acState{\tilde{e}}{\tilde{r}-1}$.
The completion of either $\acState{\tilde{e}}{\tilde{r}}$ or $\acFState{\tilde{e}}{\tilde{r}}$ is followed by a $\tuu{}$ transition to $\acState{\tilde{e}}{\tilde{r}+1}$, 
which entails handling any postponed unreliable messages (see \#1 in Table~\ref{tab:event_action}).

\subsubsection{Handling reliable messages}
If $p_i$ receives in any state a reliable message $m_{j,\acRState{e}{r}}$ sent from a subsequent epoch, 
then it is also from a subsequent round, i.e., $e > \tilde{e} \Rightarrow r > \tilde{r}$ (cf. Theorem~\ref{prop:rel_msg}).
Thus, unreliable messages from either previous epochs or previous rounds can be dropped 
(i.e., clearly, messages from preceding epochs are outdated and, since $e > \tilde{e} \Rightarrow r > \tilde{r}$, 
messages from preceding rounds are outdated as well).
As a result, $p_i$ must handle reliable messages sent from either the current or the subsequent epoch, i.e., $e\geq \tilde{e}$; 
in both cases, $p_i$ can only be in an reliable round (cf. Theorems~\ref{prop:rel_msg} and~\ref{prop:rel_msg_same_epoch}).
If $m_{j,\acRState{e}{r}}$ was sent from a subsequent epoch, then $e=\tilde{e}+1$ and $r=\tilde{r}+1$ (cf. Theorem~\ref{prop:rel_msg});
in this case, $p_i$ postpones the delivery of $m_j$ for $\acRState{\tilde{e}+1}{\tilde{r}+1}$; yet, it sends $m_j$ further 
via $G_R$\footnote{\ACplus{} does not postpone the sending of failure notifications and, to not break the message tracking mechanism that relies on message ordering, 
it does not postpone sending reliable messages either (see Appendix~\ref{app:prem_msg}).
}
(see \#6 in Table~\ref{tab:event_action}).
Otherwise, $m_{j,\acRState{\tilde{e}}{r}}$ was sent from either the subsequent or the current round (cf. Theorem~\ref{prop:rel_msg_same_epoch}).
Receiving $m_{j,\acRState{\tilde{e}}{\tilde{r}+1}}$ while in $\acRState{\tilde{e}}{\tilde{r}}$ triggers a $\tsk{}$ transition (see Figure~\ref{fig:skip_trans}).
$\tsk{}$ consists of two operations (see \#7 in Table~\ref{tab:event_action}): (1) A-deliver the last completed state (i.e., $\acState{\tilde{e}-1}{\tilde{r}}$);
and (2) move to $\acRState{\tilde{e}}{\tilde{r}+1}$ and re-handle $m_{j,\acRState{\tilde{e}}{\tilde{r}+1}}$. 
Finally, receiving $m_{j,\acRState{\tilde{e}}{\tilde{r}}}$ while in $\acRState{\tilde{e}}{\tilde{r}}$ consists of three operations (see \#8 in Table~\ref{tab:event_action}):
(1) send $m_j$ further via $G_R$; 
(2) A-broadcast own message (if not done already);
and (3) try to complete the reliable round $\tilde{r}$.

To complete a reliable round, \ACplus{} uses early termination---the necessary and sufficient condition 
for $p_i$ to complete $\acRState{\tilde{e}}{\tilde{r}}$ is to stop tracking all messages~\cite{poke2017allconcur}. 
Thus, once $p_i$ completes $\acRState{\tilde{e}}{\tilde{r}}$, it can safely A-delivered it.
Moreover, servers, for which no message was A-delivered, are removed.
This entails both updating $G_U$ to ensure connectivity (i.e., in every reliable round, the servers agree on the $G_U$ for that epoch)
and removing the invalid failure notifications~(\cref{sec:round_based}). 
Depending on whether all failure notifications are removed, we distinguish between a no-fail transition $\trf{}$ and a fail transition $\trr{}$.
A $\trf{}$ transition to $\acFState{\tilde{e}}{\tilde{r}+1}$ entails handling any postponed unreliable messages (see \#2 in Table~\ref{tab:event_action}).
A $\trr{}$ transition to $\acRState{\tilde{e}+1}{\tilde{r}+1}$ entails delivering any postponed reliable messages (see \#8 in Table~\ref{tab:event_action}).

\subsubsection{Handling failure notifications}
A failure notification $\mathit{fn}_{j, k}$ is valid only if both the owner $p_k$ and the target $p_j$ are not removed~(\cref{sec:round_based}).
Receiving a valid notification while in an unreliable round, triggers a rollback to the latest A-delivered round and the reliable rerun of the subsequent round. 
Thus, if $p_i$ is in $\acState{\tilde{e}}{\tilde{r}}$, then it moves to $\acRState{\tilde{e}+1}{\tilde{r}-1}$ (see \#3 in Table~\ref{tab:event_action}), 
while from $\acFState{\tilde{e}}{\tilde{r}}$ it moves to $\acRState{\tilde{e}+1}{\tilde{r}}$ (see \#4 in Table~\ref{tab:event_action}). 
In both cases, $p_i$ re-handles $\mathit{fn}_{j, k}$ in the new reliable round.
Handling a valid notification while in a reliable round consists of three operations (see \#9 in Table~\ref{tab:event_action}):
(1) send the notification further via $G_R$; 
(2) update the tracking digraphs;
and (3) try to complete the reliable round $\tilde{r}$.

Updating the tracking digraphs after receiving a valid notification $\mathit{fn}_{j, k}$ follows the procedure described in \AC{}~\cite{poke2017allconcur}. 
For any tracking digraph $\mathbf{g}[p_\ast]$ that contains $p_j$, we identify two cases. First, if $p_j$ has no successors in $\mathbf{g}[p_\ast]$, 
then $\mathbf{g}[p_\ast]$ is recursively expanded by (uniquely) adding the successors of any vertex $p_p$ (from $\mathbf{g}[p_\ast]$) 
that is the target of a received failure notification, except for those successors that are the owner of a failure notification targeting $p_p$.
Second, if one of $p_j$ successors in $\mathbf{g}[p_\ast]$ is $p_k$, then the edge $(p_j,p_k)$ is removed from $\mathbf{g}[p_\ast]$. 
In addition, $\mathbf{g}[p_\ast]$ is pruned by first removing the servers with no path from $p_\ast$ to themselves and then, 
if all remaining servers are targeted by received failure notifications, by removing all of them. 
When starting a reliable round, the tracking digraphs are reset and the valid failure notifications are redelivered.
Thus, this procedure needs to be repeated for all valid failure notifications. 

\subsubsection{Initial bootstrap and dynamic membership}
Bootstrapping \ACplus{} requires a reliable centralized service, such as Zookeeper~\cite{Hunt:2010:ZWC:1855840.1855851},
that enables the servers to agree on the initial configuration (i.e., the identity of the participating servers and the two digraphs).
Once \ACplus{} starts, it is completely decentralized---any further reconfigurations (including also membership changes) 
are agreed upon via atomic broadcast.


\subsection{The impact of eventual accuracy}
\label{sec:eventual_accuracy}

Using $\Diamond\mathcal{P}$ entails that failure notifications may not result in their targets being eventually removed.
Let $\mathit{fn}_{j, k}$ be an inaccurate notification, i.e., its target $p_j$ is non-faulty. 
Despite $p_k$ disseminating $\mathit{fn}_{j, k}$ throughout the system, 
\ifdefined\TECHREP
$p_j$'s messages are received by all non-faulty servers. 
Thus, $p_j$ is never removed (since a server is removed if and only if, at the end of a round, no message is A-delivered for that server).
\fi
\ifdefined\PAPER
$p_j$'s messages are received by all non-faulty servers, which means $p_j$ is never removed.
\fi
Not removing $p_j$ implies that, as long as $p_k$ is also non-faulty, the notification $\mathit{fn}_{j, k}$ remains valid, 
which results in \ACplus{} running in the reliable mode, even though no failures have occurred. 

In order to enable the redundancy-free agreement of the unreliable mode, \ACplus{} must eventually invalidate inaccurate failure notifications. 
To this end, every failure notification is tagged with a unique sequence 
\ifdefined\TECHREP
number\footnote{For instance, a straightforward method to ensure 
the uniqueness of the sequence number, is for each server to keep a separate counter for each one of the predecessors it monitors.}~$s$, i.e., $\mathit{fn}_{j, k, s}$.
\fi
\ifdefined\PAPER
number~$s$, 
\fi
Then, once $p_k$ realizes that $\mathit{fn}_{j, k, s}$ is inaccurate (i.e., it no longer suspects $p_j$ to have failed), 
it includes a $\mathit{revoke}_{j,k,s}$ control message in the next message it A-broadcasts. 
When a server A-delivers a $\mathit{revoke}_{j,k,s}$ message, it considers $\mathit{fn}_{j, k, s}$ to be invalid and thus, removes it from its set of received failure notifications.


\ifdefined\PAPER
\subsection{Informal proof of correctness}
\label{sec:corectness_proof}

To prove \ACplus{}'s correctness, we show that the four properties of non-uniform atomic broadcast are guaranteed. 
Due to space limitations, we defer the proof of correctness to Appendix~\ref{app:corectness_proof}.
Moreover, in Appendix~\ref{app:uniformity}, we discuss the modifications required for these properties to apply also to faulty servers (i.e., uniform atomic broadcast).
\fi

\ifdefined\TECHREP
\subsection{Informal proof of correctness}
\ifdefined\TECHREP
\label{sec:corectness_proof}
\fi
\ifdefined\PAPER
\label{app:corectness_proof}
\fi

As described in Section~\ref{sec:algo_description}, \ACplus{} solves non-uniform atomic broadcast.
Thus, to prove \ACplus{}'s correctness, we show that the four properties of non-uniform atomic broadcast 
are guaranteed (see Theorems~\ref{th:integrity},~\ref{th:validity},~\ref{th:agreement}, and~\ref{th:total_order}). 
Throughout the proof, we assume both $\mathcal{P}$ and no more than $f$ failures. 
In Appendix~\ref{app:uniformity}, we discuss the modifications required for these properties to apply also to faulty servers (i.e., uniform atomic broadcast).

%

\begin{theorem}[Integrity]
\label{th:integrity}
For any message $m$, every non-faulty server A-delivers $m$ at most once, and only if 
$m$ was previously A-broadcast by \sender{m}.
\end{theorem}
\begin{IEEEproof}
Integrity is guaranteed by construction; the reason is twofold.  
First, when rolling back, servers rerun rounds not yet A-delivered; thus, every round can be A-delivered at most once. 
Second, when a round is A-delivered, all the messages in the set $M$ (which are A-delivered in a deterministic order) 
were previously A-broadcast.
\end{IEEEproof}

To prove validity, we introduce the following lemma that proves \ACplus{} makes progress:
\begin{lemma} 
\label{lemma:progress}
Let $p_i$ be a non-faulty server that starts epoch $e$ and round $r$. 
Then, eventually, $p_i$ makes progress and moves to another state. 
\end{lemma}
\begin{IEEEproof}
We identify two cases: $\acRState{e}{r}$ and $\acState{e}{r}$.
If $p_i$ starts $\acRState{e}{r}$, it eventually either completes it (due to early termination) 
and moves to the subsequent round (after either $\trf{}$ or $\trr{}$), or it skips it and moves 
to $\acRState{e}{r+1}$ (after $\tsk{}$).
If $p_i$ starts $\acState{e}{r}$, either it eventually completes it, after receiving messages from all servers 
(e.g., if no failures occur), and moves to $\acState{e}{r+1}$ (after $\tuu{}$) 
or it eventually receives a failure notifications and rolls back to the latest A-delivered round (after either $\tur$ or~$\tfr$).
\end{IEEEproof}

\begin{theorem}[Validity]
\label{th:validity}
If a non-faulty server A-broadcasts $m$, then it eventually A-delivers $m$.
\end{theorem}
\begin{IEEEproof}
Let $p_i$ be a non-faulty server that A-broadcast $m$ (for the first time) in round $r$ and epoch $e$. 
If $\acRState{e}{r}$, then due to early termination, $p_i$ eventually A-delivers $r$.
If $p_i$ A-broadcast $m$ in $\acState{e}{r}$, then $\acState{e}{r}$ is either completed or interrupted by a failure notification (cf.~Lemma~\ref{lemma:progress}).
If $p_i$ completes $\acState{e}{r}$, then  $p_i$ either completes $\acState{e}{r+1}$ or, due to a fail transition, 
reruns $r$ reliably in $\acRState{e+1}{r}$ (cf.~Lemma~\ref{lemma:progress}). In both cases, $p_i$ eventually A-delivers $r$~(\cref{sec:adeliv_msg}).

If $p_i$ does not complete $\acState{e}{r}$ (due to a failure notification), we identify two cases 
depending on whether $\acState{e}{r}$ is the first in a sequence of unreliable rounds.
If $\acFState{e}{r}$, then $p_i$ reruns $r$ reliably in $\acRState{e+1}{r}$ (after a $\tfr{}$ transition); 
hence, $p_i$ eventually A-delivers~$r$.
Otherwise, $p_i$ moves to $\acRState{e+1}{r-1}$, which will eventually be followed by one of $\trf{}$, $\trr{}$, or $\tsk{}$.
$\trf{}$ leads to $\acFState{e+1}{r}$ and, by following one of the above cases, eventually to the A-delivery of~$r$.
Both $\trr{}$ and $\tsk{}$ lead to a reliable rerun of $r$ and thus, to its eventual A-delivery. 
\end{IEEEproof}

To prove both agreement and total order, we first prove \emph{set agreement}---all non-faulty servers agree on 
the same set of messages for all A-delivered rounds. To prove set agreement, we introduce the following lemmas: 
\begin{lemma} 
\label{lemma:adeliv_unrel}
Let $p_i$ be a non-faulty server that A-delivers $\acState{e}{r}$ after completing $\acState{e}{r+1}$.
Then, any other non-faulty server $p_j$ eventually A-delivers $\acState{e}{r}$.
\end{lemma}
\begin{IEEEproof}
If $p_i$ A-delivers $\acState{e}{r}$ after completing $\acState{e}{r+1}$, 
$p_j$ must have started $\acState{e}{r+1}$ (cf.~Proposition~\ref{prop:one_end_all_start}),
and hence, completed $\acState{e}{r}$.
As a result, $p_j$ either receives no failure notifications, which means it eventually completes $\acState{e}{r+1}$ and A-delivers $\acState{e}{r}$, 
or receives a failure notification and moves to $\acRState{e+1}{r}$ (after $\tur{}$).
Yet, this failure notification will eventually trigger on $p_i$ a $\tur{}$ transition from $\acState{e}{r+2}$ to $\acRState{e+1}{r+1}$, 
which eventually will trigger on $p_j$ a $\tsk{}$ transition that leads to the A-delivery of $\acState{e}{r}$ (see Figure~\ref{fig:skip_trans}). 
\end{IEEEproof}

\begin{lemma} 
\label{lemma:same_epoch_deliv}
If a non-faulty server A-delivers round $r$ in epoch $e$ (i.e., either $\acRState{e}{r}$ or $\acState{e}{r}$), 
then, any non-faulty server eventually A-delivers round $r$ in epoch~$e$.
\end{lemma}
\begin{IEEEproof}
If $p_i$ A-delivers $\acRState{e}{r}$ (when completing it), then every non-faulty server eventually A-delivers $\acRState{e}{r}$. 
The reason is twofold: (1) since $p_i$ completes $\acRState{e}{r}$, every non-faulty server must start $\acRState{e}{r}$ (cf.~Proposition~\ref{prop:one_end_all_start});
and (2) due to early termination, every non-faulty server eventually also completes and A-delivers $\acRState{e}{r}$~\cite{poke2017allconcur,allconcur_tla}.

Otherwise, $p_i$ A-delivers $\acState{e}{r}$ either once it completes the subsequent unreliable round 
or after a skip transition from $\acRState{e+1}{r}$ to $\acRState{e+1}{r+1}$~(\cref{sec:adeliv_msg}).
On the one hand, if $p_i$ A-delivers $\acState{e}{r}$ after completing $\acState{e}{r+1}$, 
then any other non-faulty server eventually A-delivers $\acState{e}{r}$ (cf.~Lemma~\ref{lemma:adeliv_unrel}).
On the other hand, if $p_i$ A-delivers $\acState{e}{r}$ after a skip transition, then at least one 
non-faulty server A-delivered $\acState{e}{r}$ after completing $\acState{e}{r+1}$ (see Figure~\ref{fig:skip_trans}). 
Thus, any other non-faulty server eventually A-delivers $\acState{e}{r}$ (cf.~Lemma~\ref{lemma:adeliv_unrel}).
\end{IEEEproof}

\begin{theorem}[Set agreement]
\label{th:set_agreement}
If two non-faulty servers A-deliver round $r$, then both A-deliver the same set of messages.
\end{theorem}
\begin{IEEEproof}
Let $p_i$ and $p_j$ be two non-faulty servers that A-deliver round $r$. 
Clearly, both servers A-deliver $r$ in the same epoch~$e$ (cf.~Lemma~\ref{lemma:same_epoch_deliv}). 
Thus, we distinguish between $\acRState{e}{r}$ and $\acState{e}{r}$.
If $\acRState{e}{r}$, both $p_i$ and $p_j$ A-deliver the same set of messages due to the set agreement property of early termination~\cite{poke2017allconcur,allconcur_tla}.
If $\acState{e}{r}$, both $p_i$ and $p_j$ completed $\acState{e}{r}$, i.e., both received messages from all servers; 
thus, both A-deliver the same set of messages.
\end{IEEEproof}

\begin{theorem}[Agreement]
\label{th:agreement}
If a non-faulty server A-delivers $m$, then all non-faulty servers eventually A-deliver $m$.
\end{theorem}
\begin{IEEEproof}
We prove by contradiction. Let $p_i$ be a non-faulty server that A-delivers $m$ in round $r$ and epoch $e$.
We assume there is a non-faulty server $p_j$ that never A-delivers $m$. 
According to Lemma~\ref{lemma:same_epoch_deliv}, $p_j$ eventually A-delivers round $r$ in epoch~$e$.
Yet, this means $p_j$ A-delivers (in round $r$) the same set of messages as $p_i$ (cf.~Theorem~\ref{th:set_agreement}), 
which contradicts the initial assumption.
\end{IEEEproof}

\begin{theorem}[Total order]
\label{th:total_order}
If two non-faulty servers $p_i$ and $p_j$ A-deliver messages $m_1$ and $m_2$, 
then $p_i$ A-delivers $m_1$ before $m_2$, if and only if $p_j$ A-delivers $m_1$ before $m_2$.
\end{theorem}
\begin{IEEEproof}
From construction, in \ACplus{}, every server A-delivers rounds in order (i.e., $r$ before $r+1$). 
Also, the messages of a round are A-delivered in a deterministic order. 
Moreover, according to both Lemma~\ref{lemma:same_epoch_deliv} and Theorem~\ref{th:set_agreement}, 
$p_i$ A-delivers $m_1$ and $m_2$ in the same states as $p_j$. Thus, $p_i$ and $p_j$ A-deliver 
$m_1$ and $m_2$ in the same order. 
\end{IEEEproof}
\fi

\ifdefined\TECHREP
\subsection{Widening the scope: Updating the reliable digraph}
\label{sec:update_Gf}

In general, it is not necessary to update $G_R$, as long as the number of failed servers does not exceed~$f$.
In \ACplus{}, $f$ provides the reliability of the system. Once servers fail, this reliability drops. 
To keep the same level of reliability, $G_R$ must be periodically updated.
\ifdefined\PAPER
Due to space limitations, we defer the discussion on how to update $G_R$ to the extended technical report~\cite{poke2017adual_tr}.
\fi
\ifdefined\TECHREP
Yet, in \ACplus{}, failure notifications are immediately handled, which requires $G_R$ to remain unchanged.
Thus, we introduce the concept of an \emph{eon}---a sequence of epochs in which $G_R$ remains unchanged. 
To connect two subsequent eons $\epsilon_1$ and $\epsilon_2$, we extend $\epsilon_1$ with a \emph{transitional} reliable round.
This transitional round is similar to the transitional configuration in Raft~\cite{Ongaro2014}.
In comparison to a normal reliable round, in a transitional round, a server executes two additional operations 
before A-broadcasting its message: (1) set up $G_R^{\epsilon_2}$, the digraph for eon $\epsilon_2$; and 
(2) start sending heartbeat messages also on $G_R^{\epsilon_2}$. 
Note that all the other operations are executed on $G_R^{\epsilon_1}$.

The transitional round acts as a delimiter between the two digraphs, $G_R^{\epsilon_1}$ and $G_R^{\epsilon_2}$, 
and provides the following guarantee: 
no non-faulty server can start eon $\epsilon_2$, before all non-faulty servers start the transitional round.
Thus, when a server starts using $G_R^{\epsilon_2}$ for both R-broadcasting and detecting failures, all other 
non-faulty servers already set up $G_R^{\epsilon_2}$ and started sending heartbeat messages on it. 
Note that failure notifications are eon specific---failure notifications from $\epsilon_1$ are dropped in $\epsilon_2$, 
while failure notifications from $\epsilon_2$, received in $\epsilon_1$, are postponed.
\fi

\fi



\section{Evaluation}
\label{sec:evaluation}

\ifdefined\PAPER
\begin{figure*}[!tp]
\captionsetup[subfigure]{justification=centering}
\centering
\subcaptionbox{[SDC] Latency (log-log)\label{fig:SDC_allconcurplus_latency}} {
\includegraphics[width=.23\textwidth]{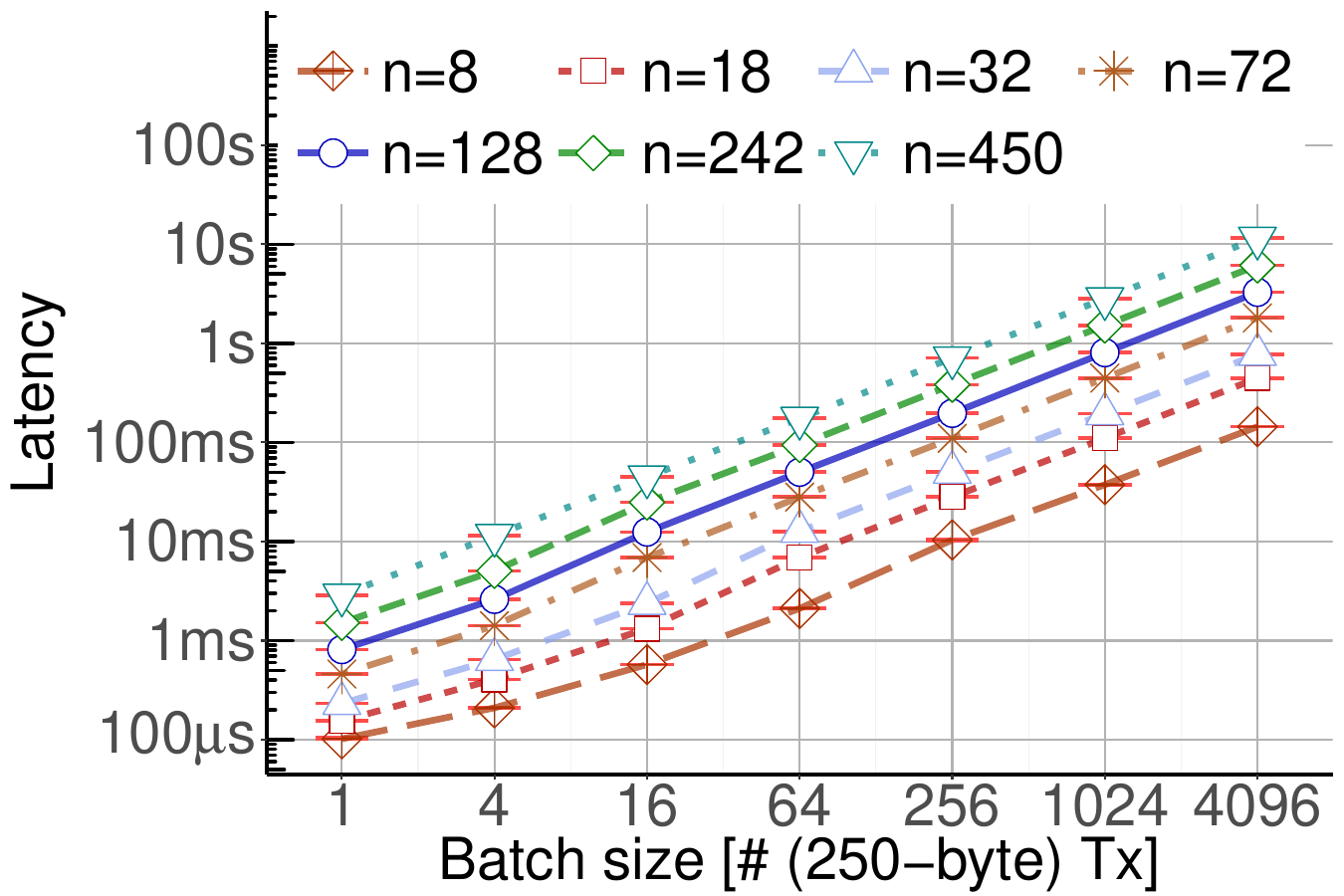}
}\hfill
\subcaptionbox{[SDC] Throughput (log-log)\label{fig:SDC_allconcurplus_throughput}} {
\includegraphics[width=.23\textwidth]{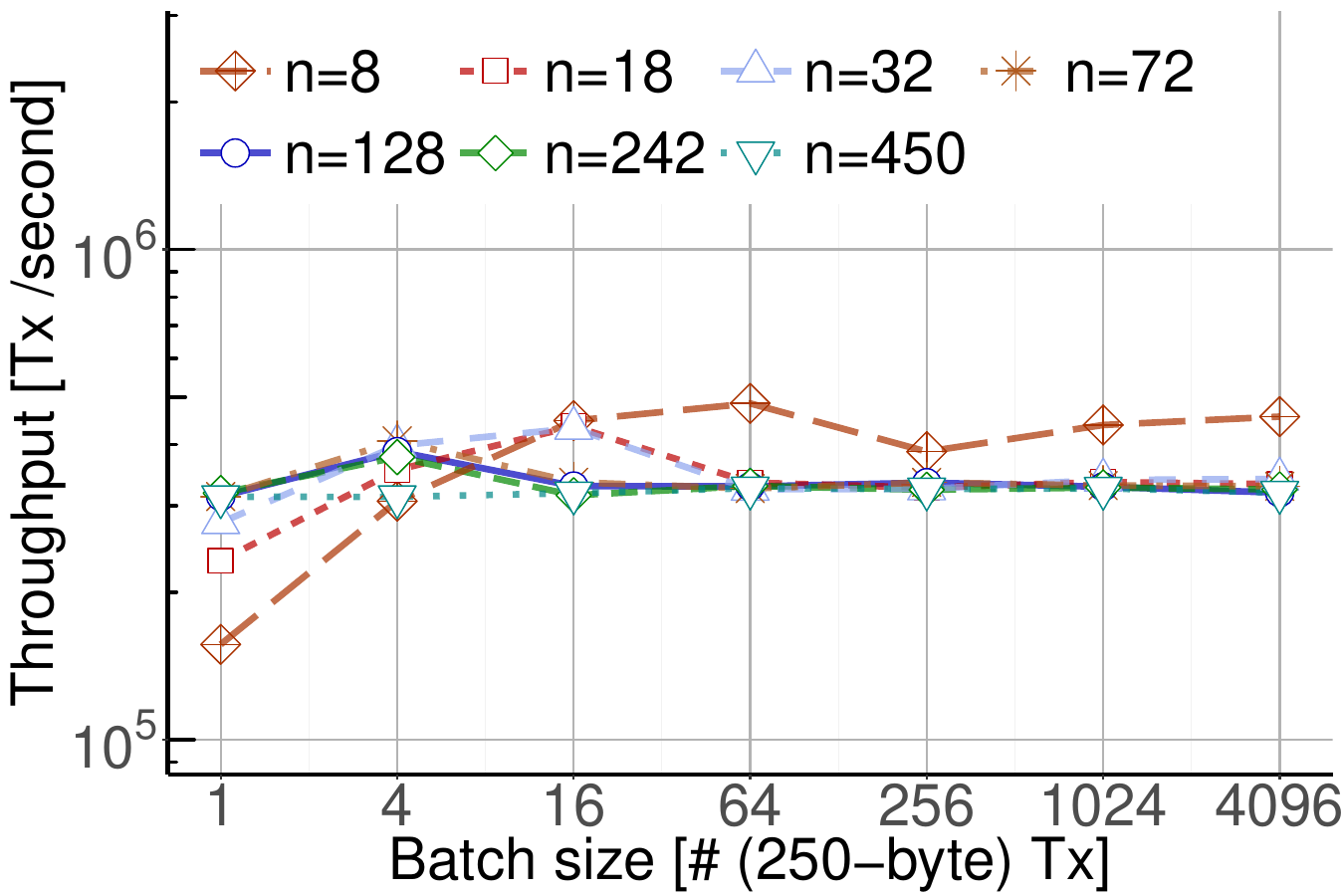}
}\hfill
\subcaptionbox{[MDC] Latency (log-log)\label{fig:MDC_allconcurplus_latency}} {
\includegraphics[width=.23\textwidth]{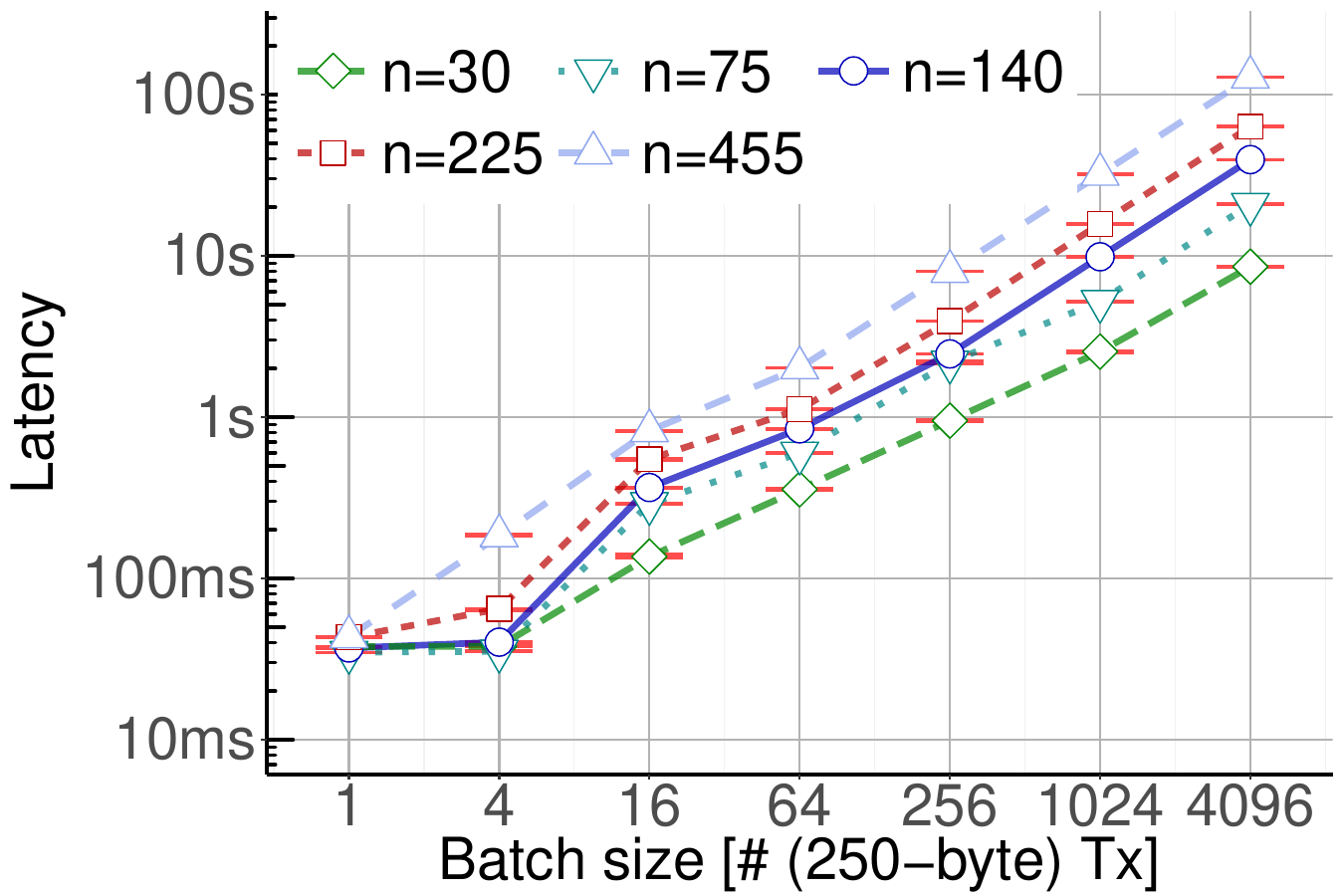}
}\hfill
\subcaptionbox{[MDC] Throughput (log-log)\label{fig:MDC_allconcurplus_throughput}} {
\includegraphics[width=.23\textwidth]{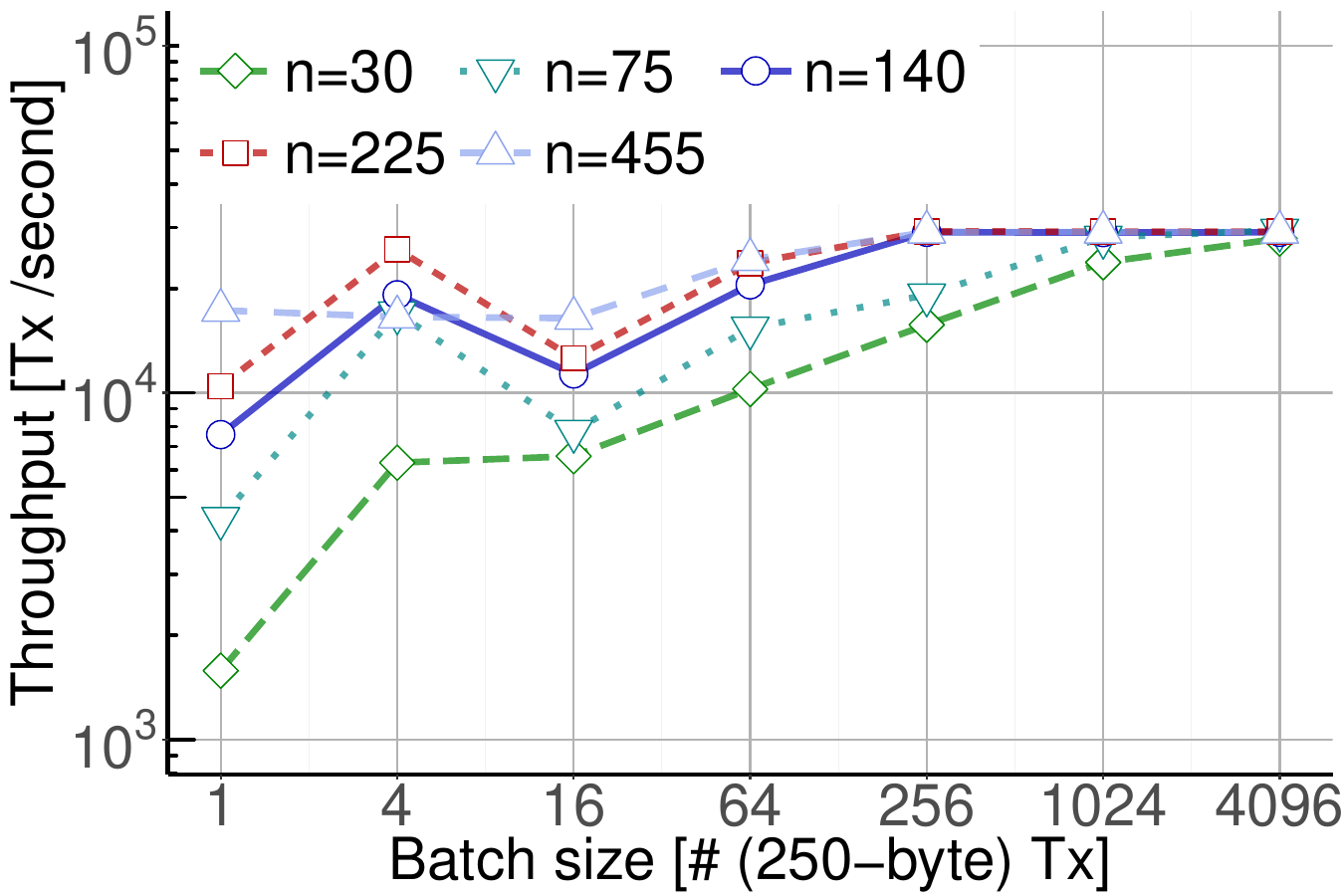}
}
\caption{The effect of batching on \ACplus{}'s performance in a non-failure scenario, for both SDC and MDC deployments. 
The measurements are done using OMNeT++.
}
\label{fig:ev_nonfailure_batching}
\end{figure*}
\fi

\ifdefined\PAPER
\begin{figure*}[!tp]
\captionsetup[subfigure]{justification=centering}
\centering
\subcaptionbox{[SDC] Latency (log-log)\label{fig:SDC_latency}} {
\includegraphics[width=.23\textwidth]{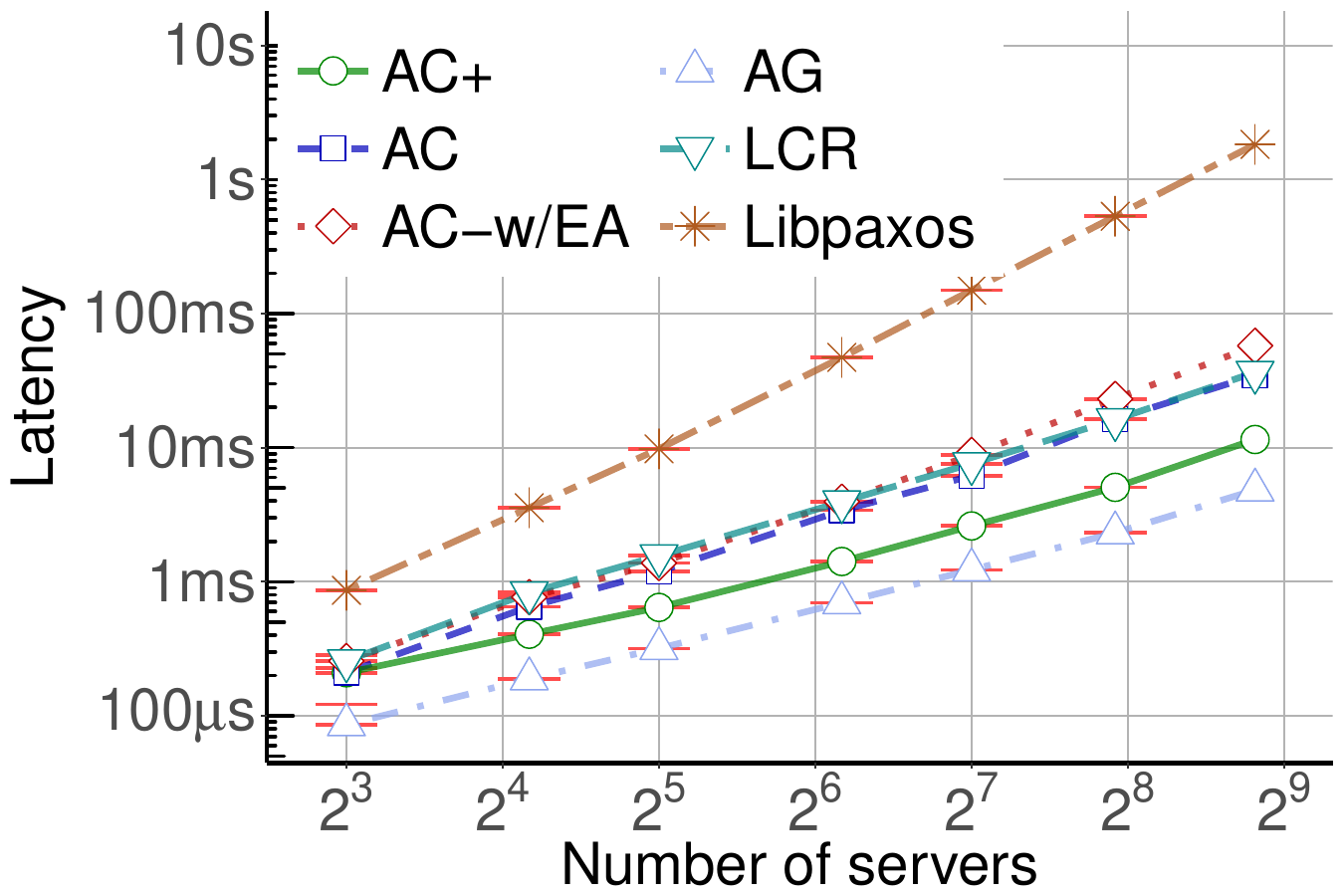}
}\hfill
\subcaptionbox{[SDC] Throughput (log-log)\label{fig:SDC_throughput}} {
\includegraphics[width=.23\textwidth]{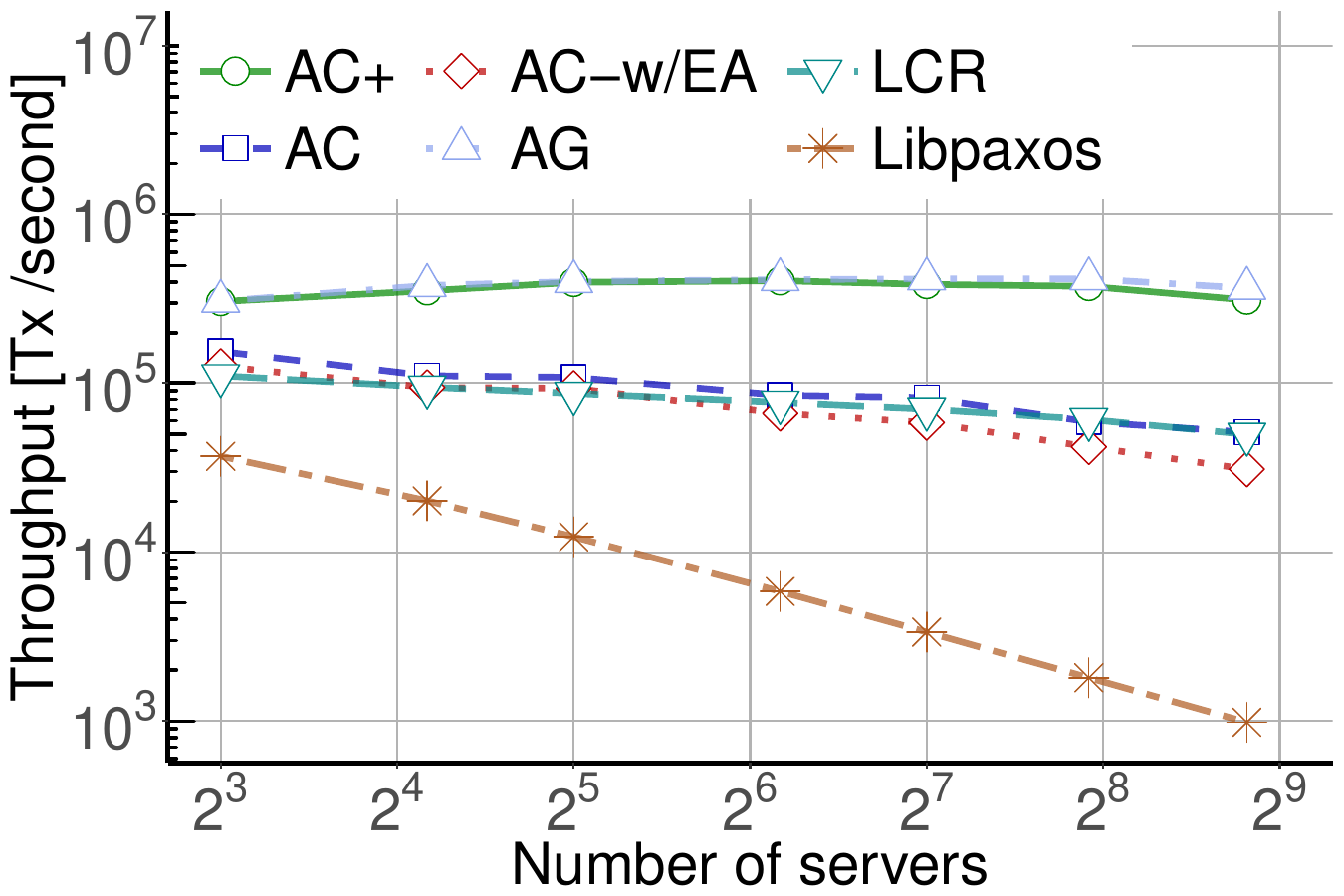}
} \hfill
\subcaptionbox{[MDC] Latency (log-log)\label{fig:MDC_latency}} {
\includegraphics[width=.23\textwidth]{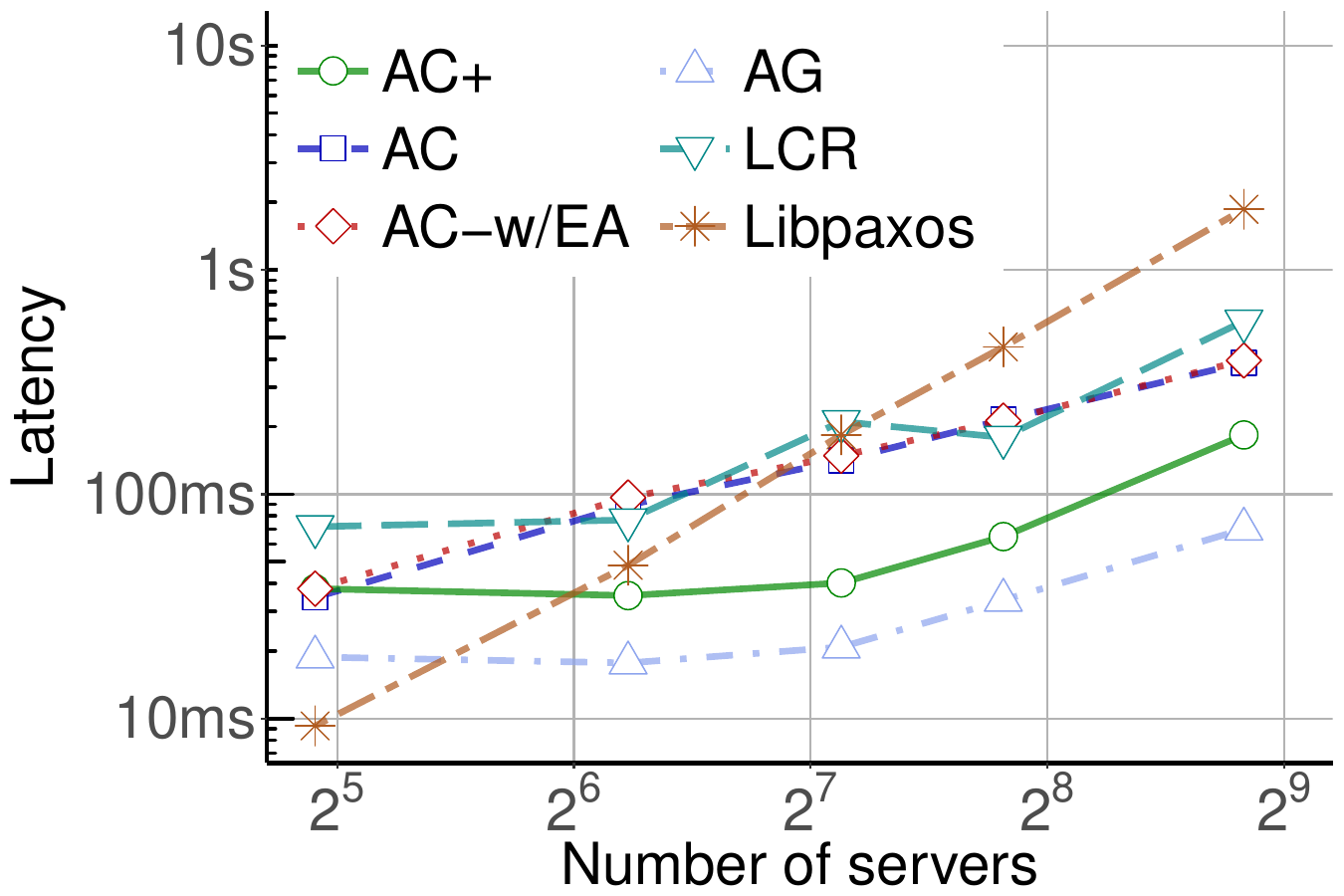}
}\hfill
\subcaptionbox{[MDC] Throughput (log-log)\label{fig:MDC_throughput}} {
\includegraphics[width=.23\textwidth]{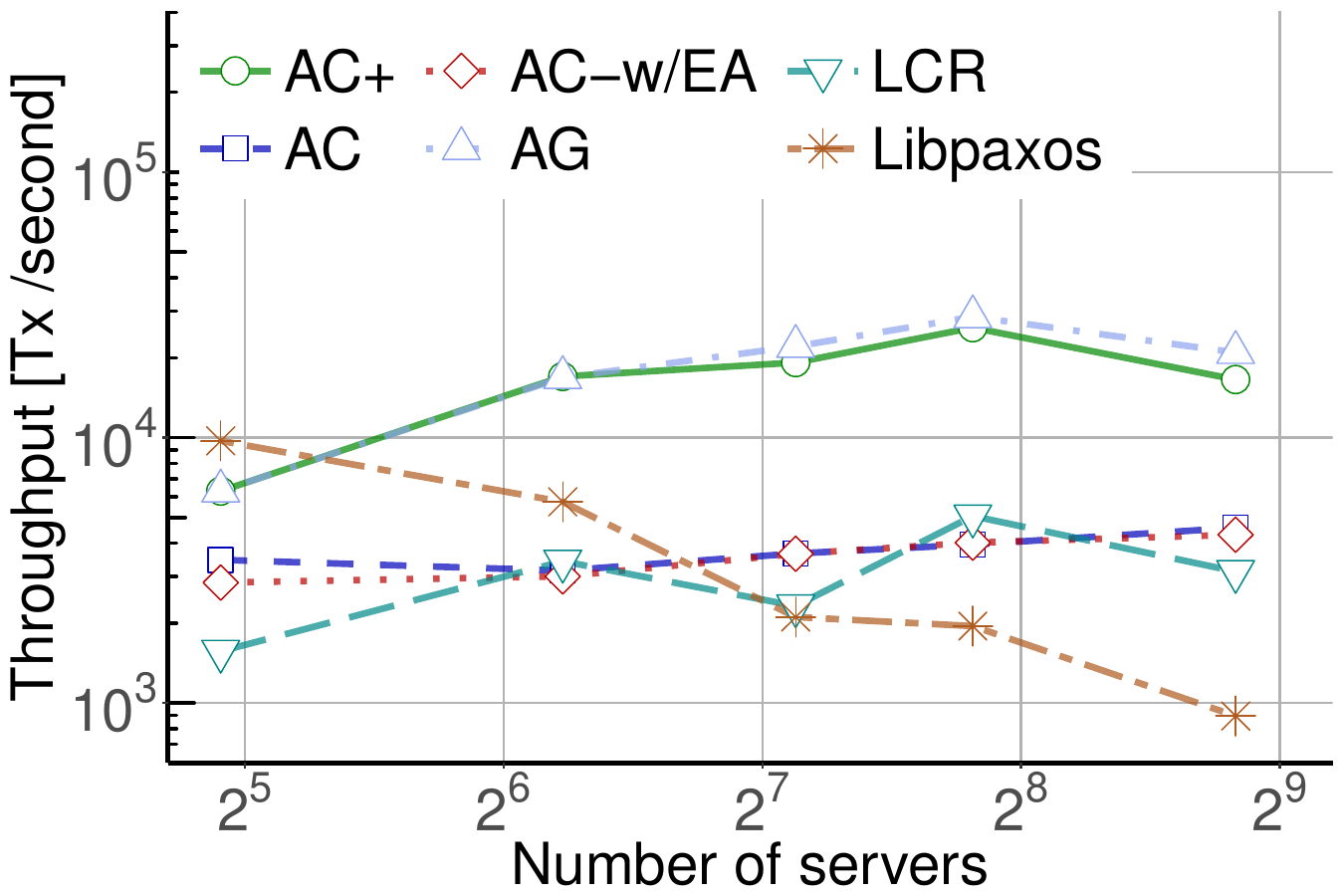}
}
\caption{\ACplus{}'s performance evaluated against \AC{}, \ACwEA{}, AllGather, LCR and Libpaxos, 
in a non-failure scenario, for both SDC and MDC deployments with batch size four.
The measurements are done using OMNeT++.
}
\label{fig:ev_nonfailure_compare}
\end{figure*}
\fi

To evaluate \ACplus{}, we consider a distributed ledger, a representative application for large-scale atomic broadcast.
In a distributed ledger, servers receive transactions as input, and their goal is to agree on a common subset of ordered 
transactions to be added to the ledger. 
%
We consider two deployments, to which, for brevity, we refer as SDC and MDC.
In SDC, the servers are located inside a single datacenter.
For example, a group of airplane tickets retailers, located in the same area, want to offer the same flights, without divulging their clients queries.
In MDC, the servers are distributed throughout five datacenters across Europe (i.e., Dublin, London, Paris, Frankfurt, Stockholm).
For example, a group of airplane tickets retailers with clients in multiple countries, 
distributed the ledger across multiple datacenters in order to provide faster local queries.
%
For both deployments, we base our evaluation on OMNeT++, a discrete-event simulator~\cite{Varga:2008:OOS:1416222.1416290}.
To realistically simulate the communication network, we use the INET framework of OMNeT++. 

For every datacenter, we consider a fat-tree network topology that provides support for multi-path routing~\cite{Al-Fares:2008:SCD:1402946.1402967}.  
The topology consists of three layers of $k$-port switches, resulting in $k$ \emph{pods} connected among each other through $k^2/4$ core switches;
every pod is directly connected to $k/2$ subnets of $k/2$ hosts each. 
In the case of MDC, one port from every core switch is used to stream traffic between datacenters;
thus, the number of pods is reduces to $k-1$.
Hosts are connected to switches via $1$~GigE $10$m cables (i.e., $0.05\mu$s delay), 
while switches are connected to each other via $1$~GigE $100$m cables (i.e., $0.5\mu$s delay).
Datacenters are interconnected via fiber optic (i.e., $5\mu$s delay per km) 
with an available bandwidth of $10$~Gbps (i.e., $10$\% of the typical datacenter interconnect bandwidth).
The length of the fiber optic is estimated as $1.1\times$ the geographical distance between the datacenters, 
resulting in latencies between $2.5$ms and $8.9$ms.
To reduce the likelihood of correlated failures, we deploy one server per subnet 
(i.e., $n=k^2/2$ for SDC and $n=5(k-1)k/2$ for MDC). 
The servers communicate via TCP.

We evaluate \ACplus{} against \AC{}, \ACwEA{}, AllGather~\cite{mpi-3.1}, LCR~\cite{Guerraoui:2010:TOT:1813654.1813656} and Libpaxos~\cite{libpaxos}.
\AC{} assumes $\mathcal{P}$, while \ACwEA{} assumes $\Diamond\mathcal{P}$, i.e., it uses the forward-backward mechanism~(\cref{sec:allconcur_overview}).
We deploy both with a reliability of $6$-nines, estimated conservatively over a period of $24$ hours with a server $\mathit{MTTF}\approx2$ years~\cite{tsubame}.
The servers are connected via a $G_S(n,d)$ digraph~\cite{Soneoka:1996:DDC:227095.227101}; \ACplus{} uses the same digraph in the reliable mode.
\ifdefined\TECHREP
Table~\ref{tab:gs_connectivity} shows the vertex-connectivity of $G_S(n,d)$ for the values of $n$ used throughout the evaluation.
\fi
AllGather is a round-based non-fault-tolerant distributed agreement algorithm where every server uses a binomial tree to A-broadcast its messages;
\ACplus{} uses the same mechanism in the unreliable mode. 
LCR is an atomic broadcast algorithm that is based on a ring topology and uses vector clocks for message ordering. 
Libpaxos is an open-source implementation of Paxos~\cite{Lamport:1998:PP:279227.279229}. 
We deploy it over $n$~servers with one proposer, five acceptors (sufficient for a reliability of $6$-nines) and $n$~learners.

\ifdefined\TECHREP
\begin{table}[!tp]
\centering
\begin{tabular}{ l | c | c | c | c | c | c | c | c | c | c | c | c  }
  \cmidrule[1.5pt](){1-13}

  
  $n$      & 8 & 18 & 30 & 32 & 72 & 75 & 128 & 140 & 225 & 242 & 450 & 455 \\  
  \hline
  \rowcolor{DarkGray}  
  $\kappa$ & 3 & 4  &  4 & 4  & 5  &  5 &  5  &  6  &  6  &   7 &   8 &  8  \\ 
  
  \cmidrule[1.5pt](){1-13} 
\end{tabular}
  \caption{The vertex-connectivity of $G_S(n,d)$.}
\label{tab:gs_connectivity}
\end{table}
\fi

To simplify the evaluation, we omit from where transactions originate 
(i.e., from where they are initially received by the servers) and how they are interpreted by the ledger.
Every server A-broadcasts a message consisting of a batch of transactions; 
once a server A-delivers a message, it adds all enclosed transactions to (its copy of) the ledger. 
Each transaction has a size of $250$ bytes, sufficient to hold a payload and cryptographic signatures 
(e.g., a typical size for Bitcoin~\cite{bitcoin} transactions).
Every server can have one outstanding message at a time: Before A-broadcasting another message, 
it waits either for the round to complete or for the message to be A-delivered.
Using this benchmark, we measure both latency and throughput. 
Latency is defined as the time between a server A-broadcasting and A-delivering a message,
while throughput, as the number of transactions A-delivered per server per second. 

\ifdefined\TECHREP
We first evaluate \ACplus{} in non-faulty scenarios~(\cref{sec:ev_nonfailure}). 
Then, we evaluate the impact of different failure scenarios on \ACplus{}'s performance~(\cref{sec:ev_failure}).
\fi

\subsection{Non-failure scenarios}
\label{sec:ev_nonfailure}

We evaluate \ACplus{}'s performance in scenarios with no failures for both SDC and MDC. 
We measure the performance at each server during a common measuring window between $t_1$ and $t_2$;
we define $t_1$ and $t_2$ as the time when every server A-delivered at least $10\times n$ and $110\times n$ messages, respectively.
If servers A-deliver the same amount of messages from every server, as is the case of \ACplus{}, then every server A-delivers $100$ own messages during a window.
Figures~\ref{fig:ev_nonfailure_batching} and~\ref{fig:ev_nonfailure_compare} report both the median latency with a 95\% nonparametric confidence interval 
and the average throughput.


\subsubsection{Message size}

\ifdefined\TECHREP
\begin{figure}[!tp]
\captionsetup[subfigure]{justification=centering}
\centering

\hspace*{\fill}%
\subcaptionbox{[SDC] Latency (log-log)\label{fig:SDC_allconcurplus_latency}} {
\includegraphics[width=.35\textwidth]{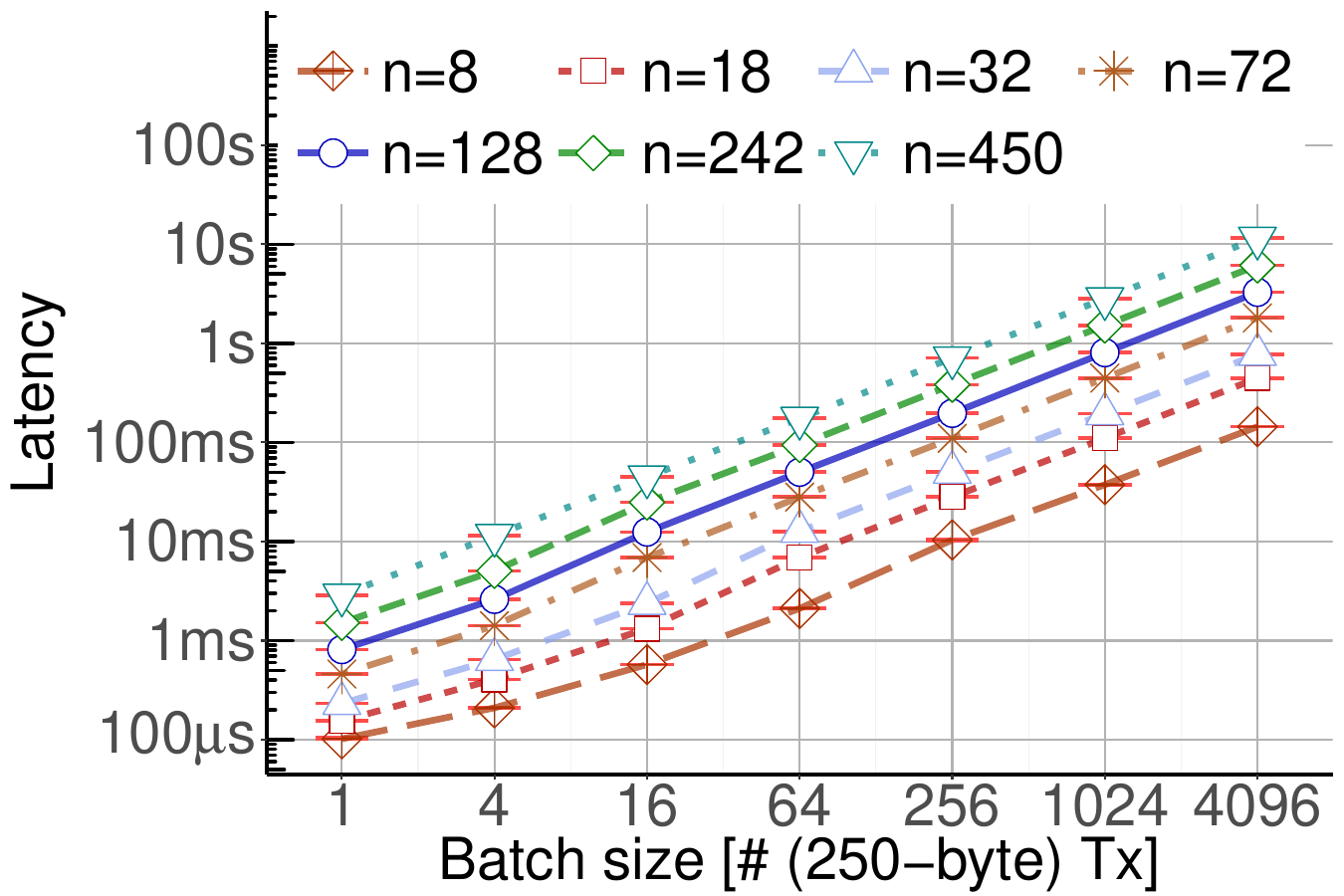}
} \hfill
\subcaptionbox{[SDC] Throughput (lin-log)\label{fig:SDC_allconcurplus_throughput}} {
\includegraphics[width=.35\textwidth]{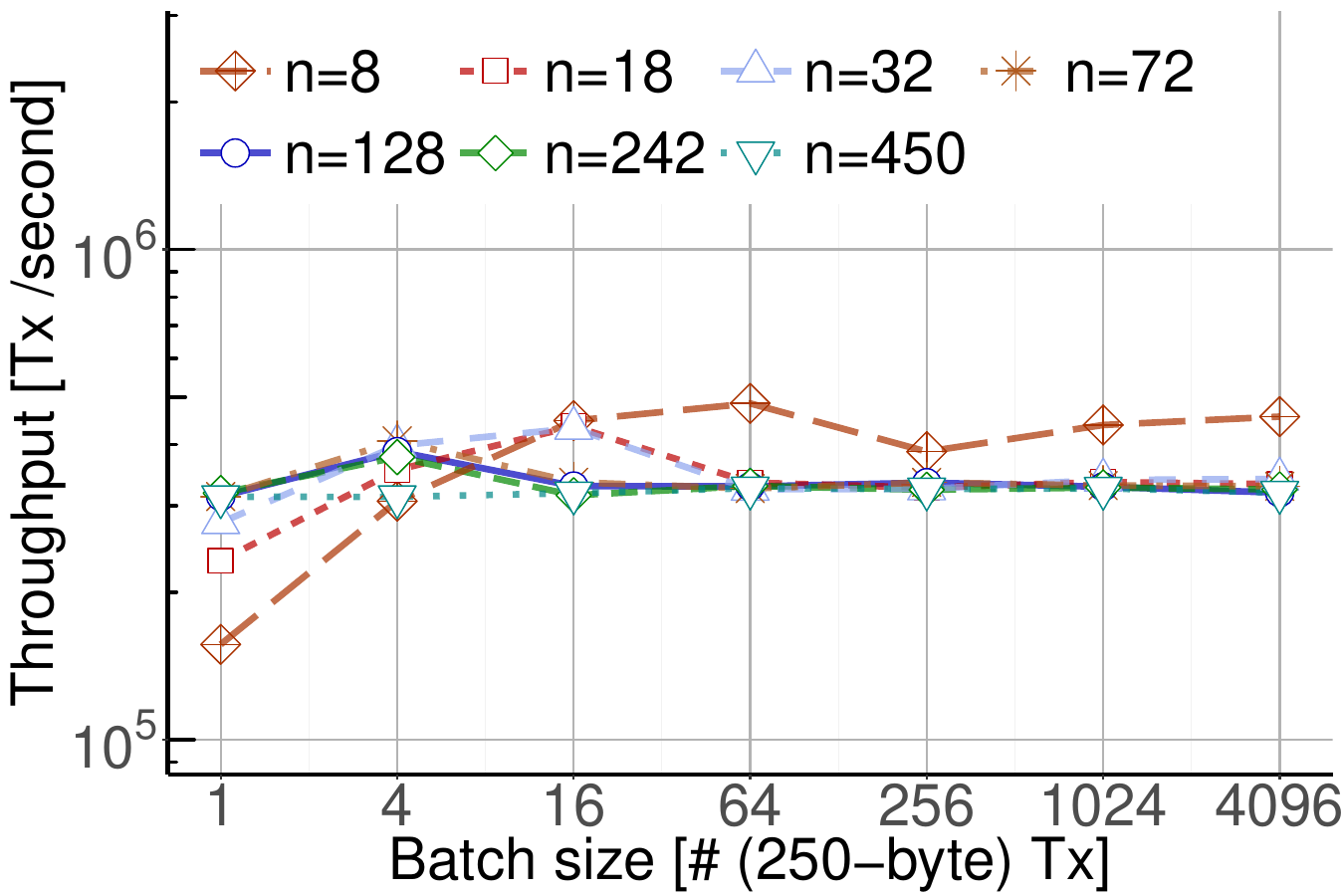}
}
\hspace*{\fill}%

\vspace{\baselineskip}

\hspace*{\fill}%
\subcaptionbox{[MDC] Latency (log-log)\label{fig:MDC_allconcurplus_latency}} {
\includegraphics[width=.35\textwidth]{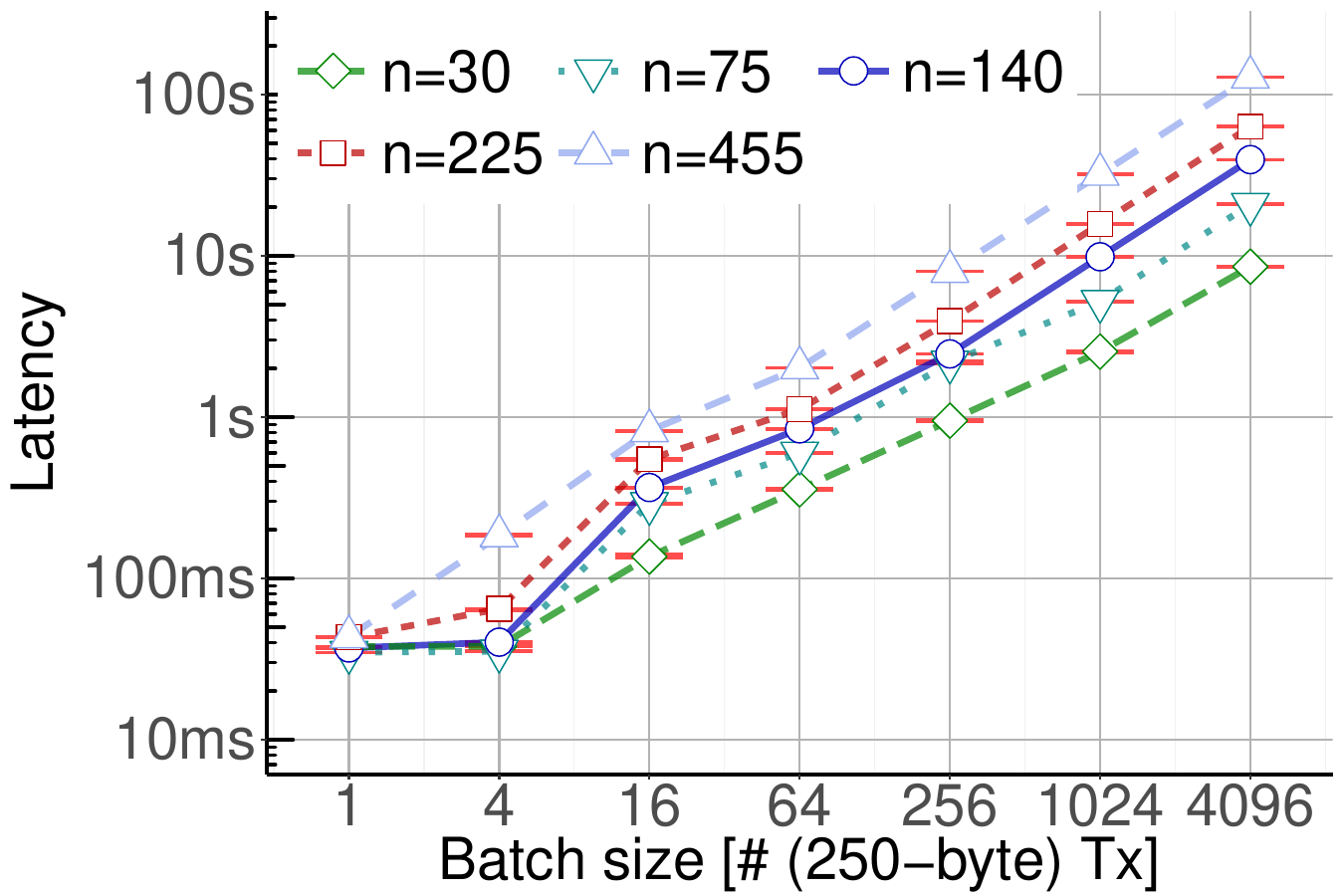}
} \hfill
\subcaptionbox{[MDC] Throughput (log-log)\label{fig:MDC_allconcurplus_throughput}} {
\includegraphics[width=.35\textwidth]{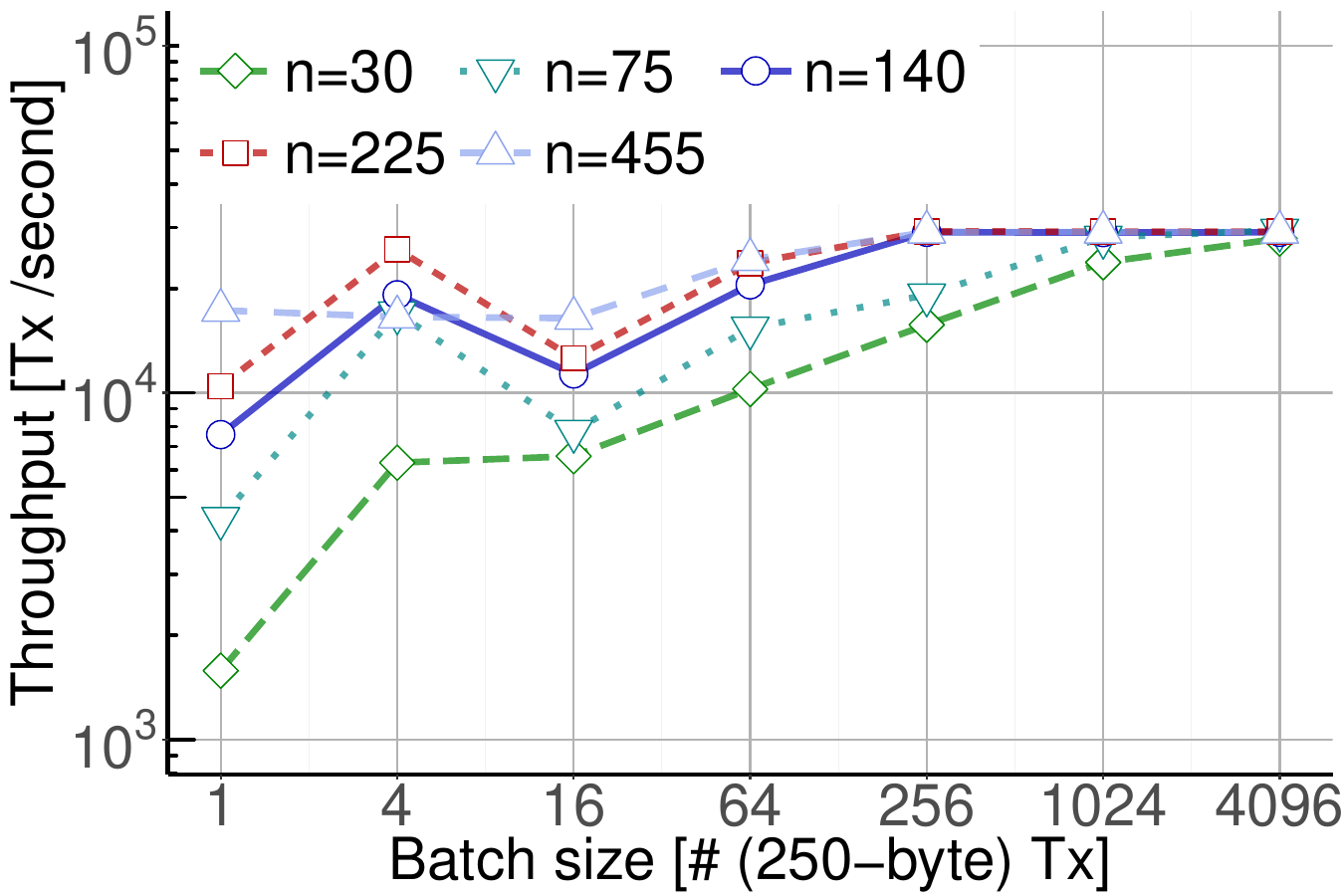}
}
\hspace*{\fill}%

\caption{The effect of batching on \ACplus{}'s performance in a non-failure scenario 
for both a single datacenter [SDC] (a), (b) and multiple datacenters [MDC] (c), (d). 
The measurements are done using OMNeT++.
The latency is reported as the median with a 95\% nonparametric confidence interval; 
the throughput is reported as the average over the measuring window.
}
\label{fig:ev_nonfailure_batching}
\end{figure}
\fi

We evaluate the effect batching transactions has on \ACplus{}'s performance, starting from one transaction per message (i.e., no batching) 
to $4,096$ transactions per message (i.e., $\approx1$MB messages).
Figure~\ref{fig:ev_nonfailure_batching} plots, for different deployments in both SDC and MDC, the latency and the throughput as a function of batching size. 
As expected, the latency is sensitive to increases in message size. 
Without batching multiple transactions into a single message, the latency is minimized. 
Yet, no batching entails usually a low throughput, since the system's available bandwidth is only saturated for large system sizes (e.g., $\approx 450$ servers).
Indeed, increasing the message size leads to higher throughput:
By batching transactions, \ACplus{}'s throughput exceeds $320,000$ transactions per second for all SDC deployments, 
and $27,000$ transactions per second for all MDC deployments (see Figures~\ref{fig:SDC_allconcurplus_throughput} and~\ref{fig:MDC_allconcurplus_throughput}).
This increase in throughput comes though at the cost of higher latency. 

Moreover, increasing the batch size may lead to higher latency due to the TCP protocol. 
For example, in an MDC deployment of $30$ servers with a batch size of $16$, the $65,535$-byte TCP Receive Window
causes servers to wait for TCP packets to be acknowledged before being able to send further all the messages of a round. 
Since acknowledgements across datacenters are slow, this results in a sharp increase in latency and thus, a drop in throughput 
(see Figures~\ref{fig:MDC_allconcurplus_latency} and~\ref{fig:MDC_allconcurplus_throughput}). 
To reduce the impact of TCP, for the remainder of the evaluation, we fix the batch size to four (i.e., $1$kB messages).

%
%


\subsubsection{Comparison to other algorithms}

\ifdefined\TECHREP
\begin{figure}[!tp]
\captionsetup[subfigure]{justification=centering}
\centering

\hspace*{\fill}%
\subcaptionbox{[SDC] Latency (log-log)\label{fig:SDC_latency}} {
\includegraphics[width=.35\textwidth]{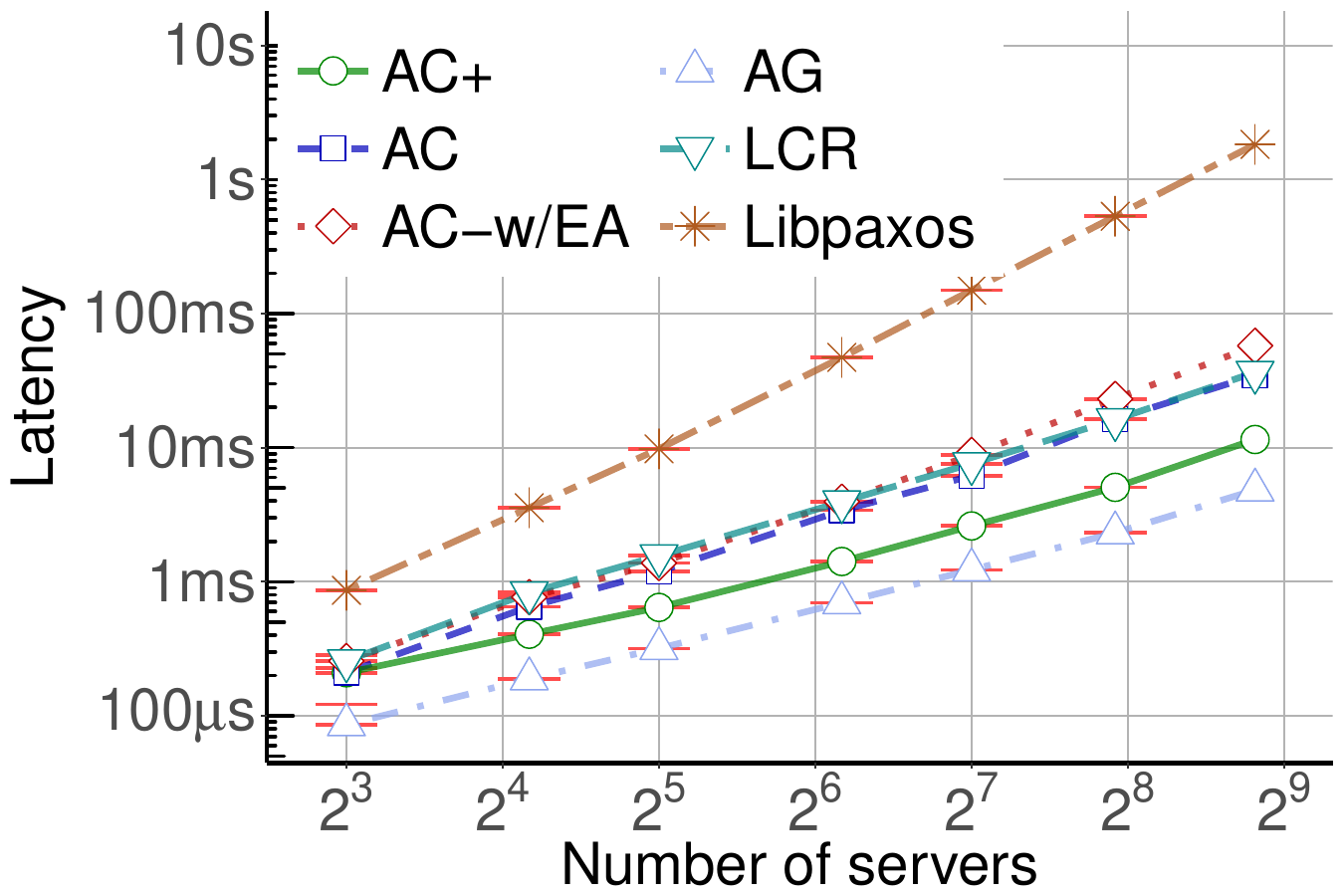}
} \hfill
\subcaptionbox{[SDC] Throughput (log-log)\label{fig:SDC_throughput}} {
\includegraphics[width=.35\textwidth]{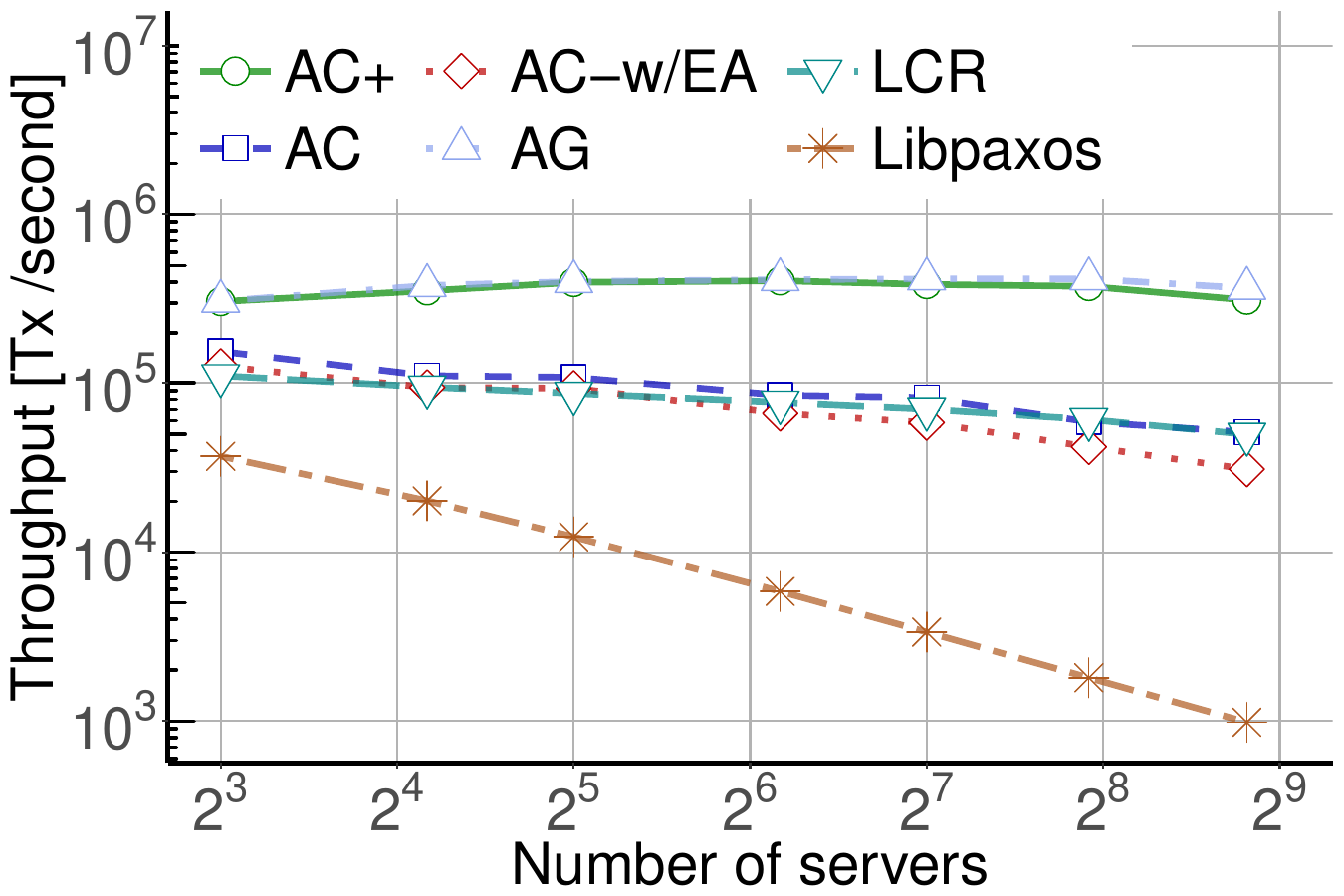}
}
\hspace*{\fill}%

\vspace{\baselineskip}

\hspace*{\fill}%
\subcaptionbox{[MDC] Latency (log-log)\label{fig:MDC_latency}} {
\includegraphics[width=.35\textwidth]{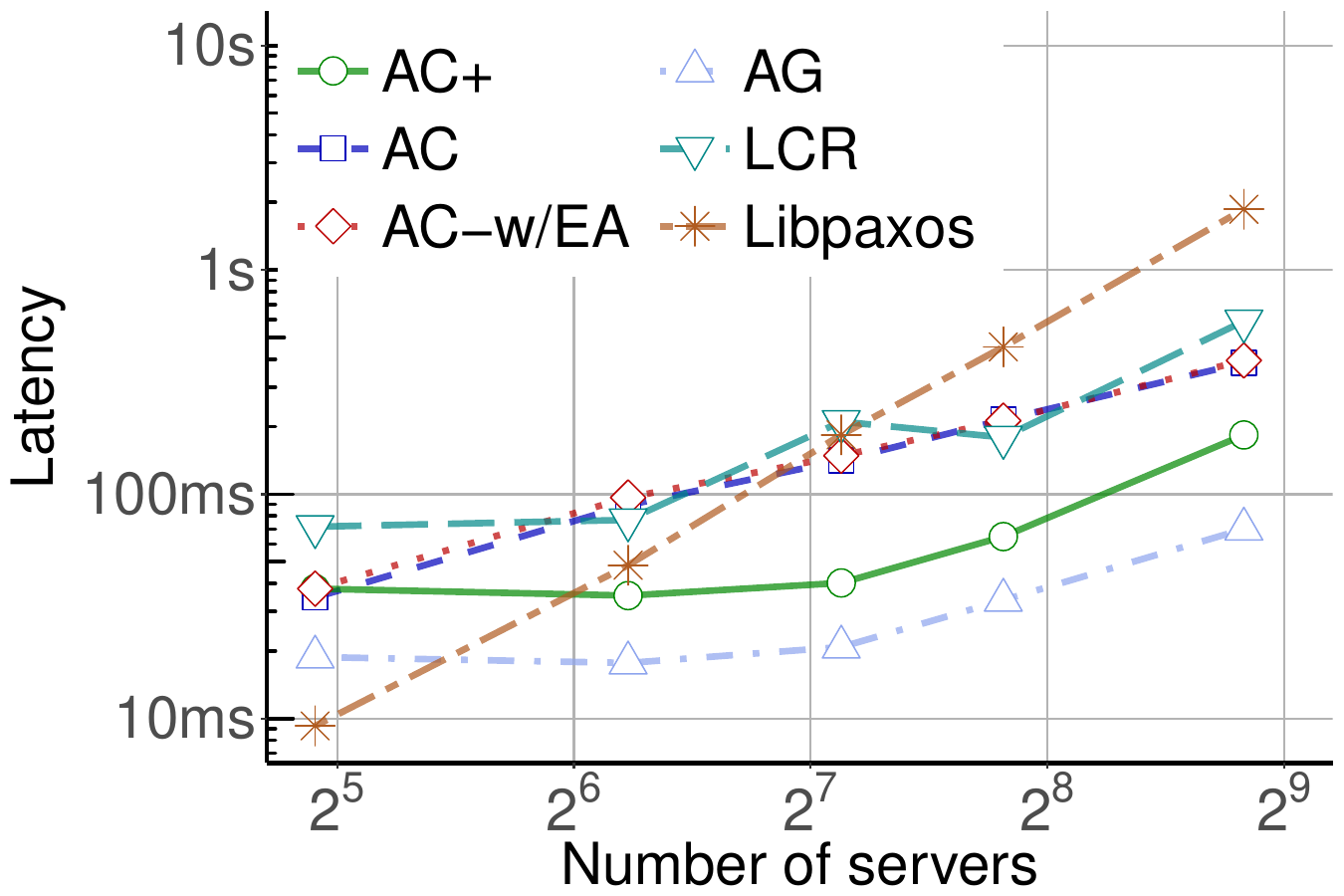}
} \hfill
\subcaptionbox{[MDC] Throughput (log-log)\label{fig:MDC_throughput}} {
\includegraphics[width=.35\textwidth]{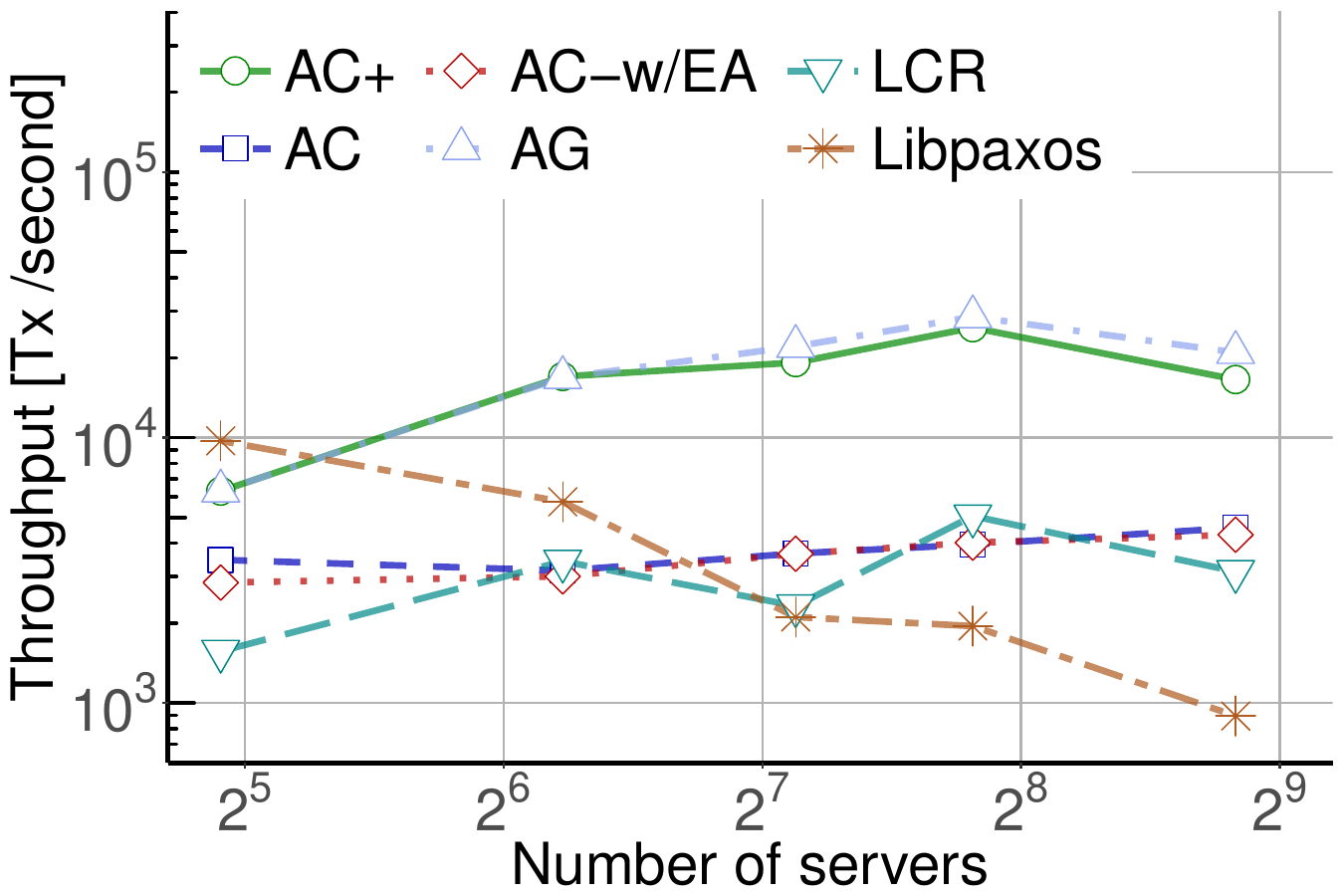}
}
\hspace*{\fill}%
\caption{\ACplus{}'s performance evaluated against \AC{}, \ACwEA{}, AllGather, LCR and Libpaxos, 
in a non-failure scenario with a batch size of four, (i.e., $1$kB messages) 
for both a single datacenter [SDC] (a), (b) and multiple datacenters [MDC] (c), (d).
The measurements are done using OMNeT++.
The latency is reported as the median with a 95\% nonparametric confidence interval; 
the throughput is reported as the average over the measuring window.
}
\label{fig:ev_nonfailure_compare}
\end{figure}
\fi

We evaluate \ACplus{}'s performance against \AC{}, \ACwEA{}, AllGather, LCR and Libpaxos, while scaling up to $455$ servers. 
Figure~\ref{fig:ev_nonfailure_compare} plots, for both SDC and MDC, the latency and the throughput as a function of the number of servers.  

\textbf{\ACplus{} vs. AllGather.} 
For both deployments, \ACplus{}'s latency is around $2$--$2.6\times$ higher than AllGather's.
\ifdefined\TECHREP
The roughly\footnote{The overhead is due to providing fault tolerance (e.g., larger message headers, more TCP connections).} 
\fi
\ifdefined\PAPER
The roughly
\fi
two-fold increase in latency is as expected: When no failures occur, \ACplus{} A-delivers a message after completing two rounds.
At the same time, \ACplus{} achieves between $79\%$ and $100\%$ of the throughput of AllGather, while providing also fault-tolerance.
The reason behind this high throughput is that, similar to AllGather, \ACplus{} A-delivers messages at the end of every round 
(except for the first round).

\textbf{\ACplus{} vs. \AC{}.} 
Due to the redundancy-free overlay network, \ACplus{} performs less work and introduces less messages in the network and, as a result, outperforms both \AC{} and \ACwEA{}. 
When comparing to \AC{}, it achieves up to $3.2\times$ lower latency and up to $6.4\times$ higher throughput for SDC, 
and up to $3.5\times$ lower latency and up to $6.4\times$ higher throughput for MDC. 
When comparing to \ACwEA{}, it achieves up to $5\times$ lower latency and up to $10\times$ higher throughput for SDC, 
and up to $3.7\times$ lower latency and up to $6.5\times$ higher throughput for MDC. 

\textbf{\ACplus{} vs. LCR.}
\ACplus{} outperforms LCR in both latency and throughput. 
The reason behind it is twofold: first, the dissemination latency of the ring topology adopted by LCR; 
and second, the message overhead necessary for using vector clocks.
Thus, for SDC, \ACplus{} is up to $3.2\times$ faster and achieves up to $6.3\times$ higher throughput. 
Although LCR is designed for local area networks~\cite{Guerraoui:2010:TOT:1813654.1813656}, 
for completion, we evaluate its performance also for MDC:
\ACplus{} is up to $5.2\times$ faster and achieves up to $8.3\times$ higher throughput.
When deploying LCR, we order the servers in the ring with minimal communication between different datacenters. 

\textbf{\ACplus{} vs. Libpaxos.} 
Paxos is designed for high-availability and thus, not intended to scale to hundreds of instances. 
Accordingly, the performance of Libpaxos drops sharply with increasing $n$. 
As a result, \ACplus{} is up to $158\times$ faster for SDC and up to $10\times$ faster for MDC; 
also, it achieves up to $318\times$ higher throughput for SDC and up to $18\times$ higher throughput for MDC.

\subsection{Failure scenarios}
\label{sec:ev_failure}

\ifdefined\PAPER
\begin{figure}[!tp]
\centering
\includegraphics[width=.30\textwidth]{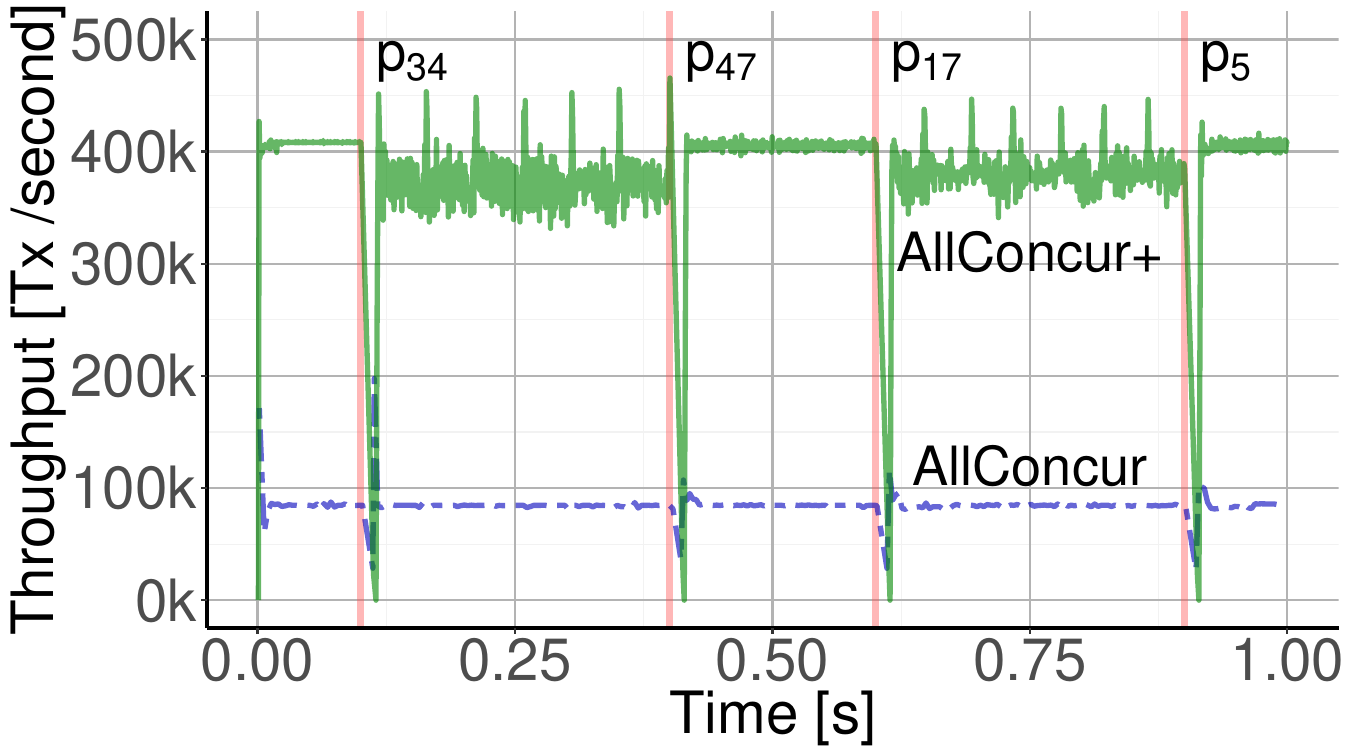}
\caption{\ACplus{}'s performance evaluated against \AC{} during a failure scenario 
for an SDC deployment of $72$ servers, with batch size four. 
The red vertical bars denote the time of failures. The FD has $\Delta_{hb}=1ms$ and $\Delta_{to}=10ms$.
The measurements are done using OMNeT++.
The reported throughput is sampled for $p_0$ at the completion of every round.
}
\label{fig:failure_scenario}
\end{figure}
\fi

\ifdefined\TECHREP
\begin{figure}[!tp]
\centering
\includegraphics[width=.41\textwidth]{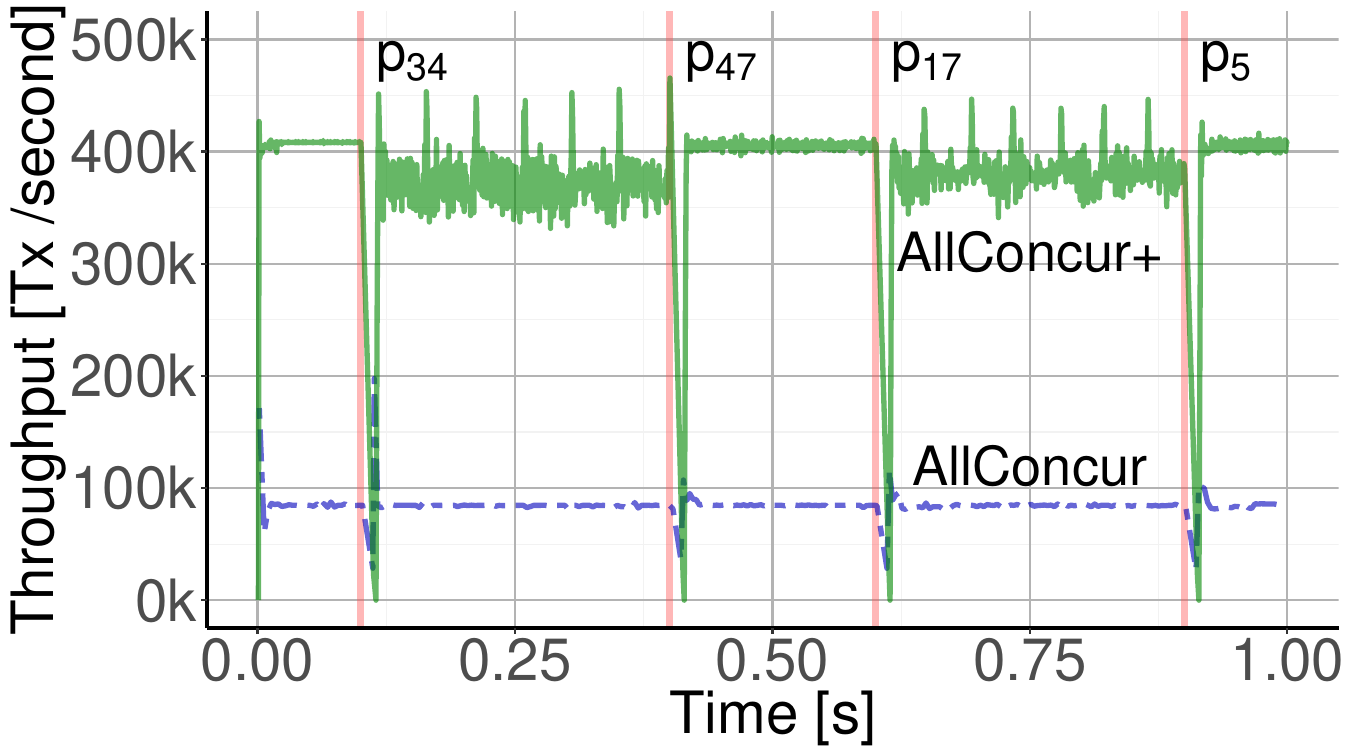}
\caption{\ACplus{}'s performance evaluated against \AC{} during a failure scenario 
for an SDC deployment of $72$ servers, with batch size four. 
The red vertical bars denote the time of failures. The FD has $\Delta_{hb}=1ms$ and $\Delta_{to}=10ms$.
The measurements are done using OMNeT++.
The reported throughput is sampled for $p_0$ at the completion of every round.
}
\label{fig:failure_scenario}
\end{figure}
\fi

To evaluate \ACplus{}'s performance when failures do occur, we first consider the following failure scenario: 
In an SDC deployment of $72$ servers, each A-broadcasting $1$kB messages, four failures occur during an interval of one second.
To detect failures, servers rely on a heartbeat-based FD with a heartbeat period $\Delta_{hb}=1ms$ and a timeout period $\Delta_{to}=10ms$. 
Figure~\ref{fig:failure_scenario} plots the throughput of $p_0$ as a function of time for both \AC{} and \ACplus{}.
The throughput is sampled at the completion of every round, i.e., the number of transactions A-delivered divided by the time needed to complete the round.

The four failures (indicated by red vertical bars) have more impact on \ACplus{}'s throughput than \AC{}'s. 
In \AC{}, a server's failure leads to a longer round, 
i.e., $\approx \Delta_{to}$ instead of $\approx3.3$ms when no failures occur (see Figure~\ref{fig:SDC_latency}).
For example, when $p_{34}$ fails, $p_0$ must track its message, which entails waiting for each of $p_{34}$'s successors 
to detect its failure and send a notification.
Once the round completes, \AC{}'s throughput returns to $\approx85,000$ transactions per second.
In \ACplus{}, every failure triggers a switch to the reliable mode. 
For example, once $p_0$ receives a notification of $p_{34}$'s failure, it rolls back to the latest A-delivered round and reliably reruns the subsequent round.
Once $p_0$ completes the reliable round, it switches back to the unreliable mode, where it requires two unreliable rounds to first A-deliver a round. 
In Figure~\ref{fig:failure_scenario}, we see \ACplus{}'s throughput drop for a short interval of time after every failure (i.e., $\approx16$ms). 
However, since \ACplus{}'s throughput is significantly higher in general, the short intervals of low throughput after failures 
have only a minor impact and on average, \ACplus{} achieves $\approx 4.6\times$ higher throughput than \AC{}.

Furthermore, we analyze the robustness of \ACplus{}'s performance with regard to more frequent failures.
The analysis is based on a model with two parameters: 
$\delta_u$, the expected duration of an unreliable round; and $\delta_r$, the expected duration of a reliable round. 
Clearly, $\delta_u < \delta_r$.
Since every server has one outstanding message at a time, we use data from Figure~\ref{fig:ev_nonfailure_compare} 
to estimate both parameters: $\delta_u$ is half the latency of \ACplus{} and $\delta_r$ is the latency of \AC{}.
We assume an SDC deployment with batch size four. 
\ifdefined\TECHREP
We first focus on both the expected latency and throughput over a sequence of multiple rounds.
Then, for latency sensitive applications, we look into the worst-case latency of a single round.
\fi
\ifdefined\PAPER
We focus on both the expected latency and throughput over a sequence of multiple rounds
(for an analysis of the worst-case latency of a single round, see the extended technical report~\cite{poke2017adual_tr}).
\fi
%
Throughout the analysis, we consider only no-fail and fail transitions, since skip transitions 
entail lower latency and higher throughput, i.e., a skip transition always triggers the A-delivery of messages.

\ifdefined\PAPER
\begin{figure}[!tp]
\captionsetup[subfigure]{justification=centering}
\centering
\subcaptionbox{Expected latency\label{fig:fm_expected_latency}} {
\includegraphics[width=.225\textwidth]{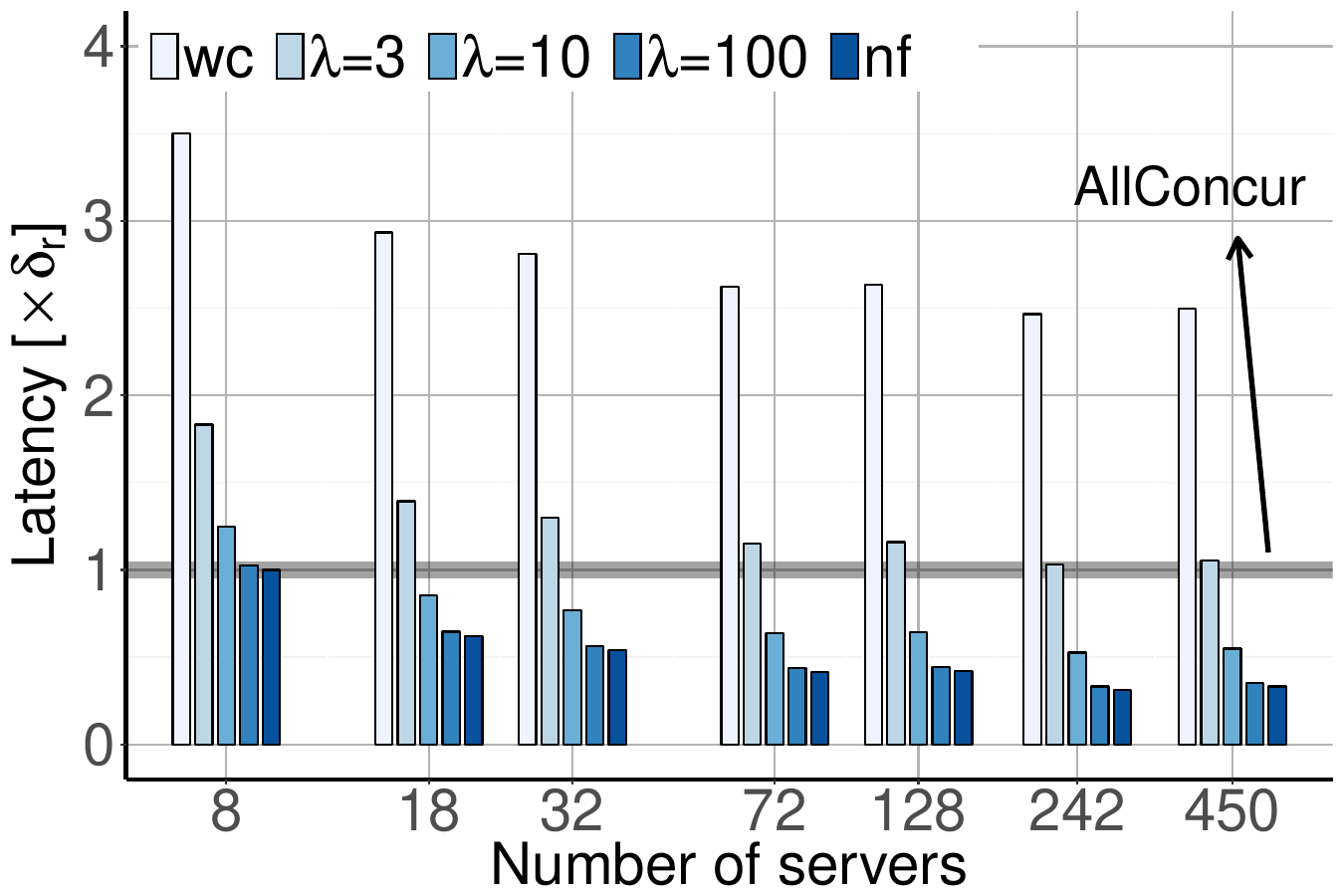}
}
\subcaptionbox{Expected throughput\label{fig:fm_expected_throughput}} {
\includegraphics[width=.225\textwidth]{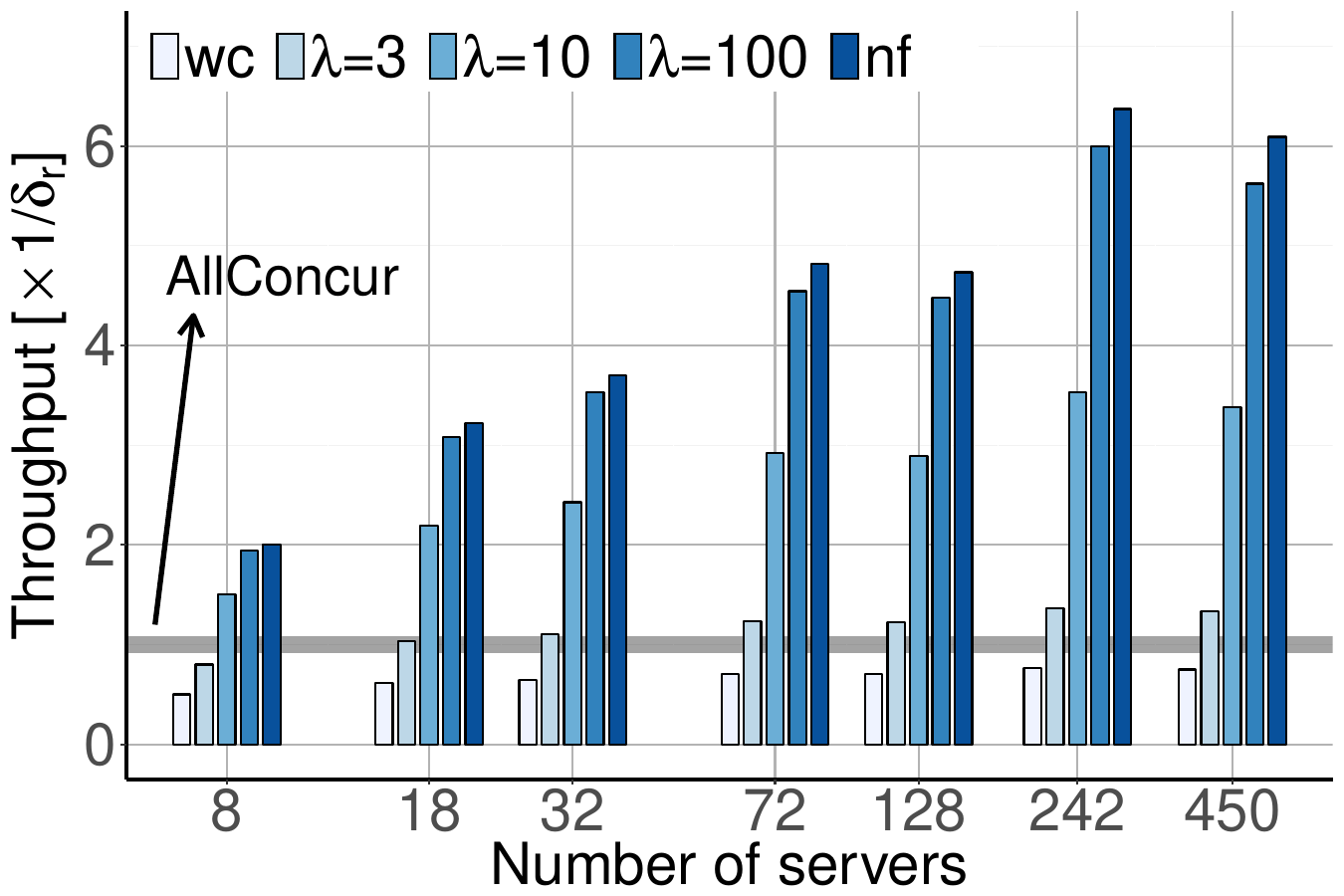}
}
\caption{\ACplus{}'s estimated expected latency and throughput (expressed as ratios to \AC{}'s values) for the non-failure scenario (nf), 
the worst-case scenario (wc), and scenarios with sequences of $\lambda$ unreliable rounds. 
We assume an SDC deployment with batch size four. 
}
\label{fig:ev_failure_analysis}
\end{figure}
\fi

\ifdefined\TECHREP
\begin{figure}[!tp]
\captionsetup[subfigure]{justification=centering}
\centering

\hspace*{\fill}%
\subcaptionbox{Expected latency\label{fig:fm_expected_latency}} {
\includegraphics[width=.35\textwidth]{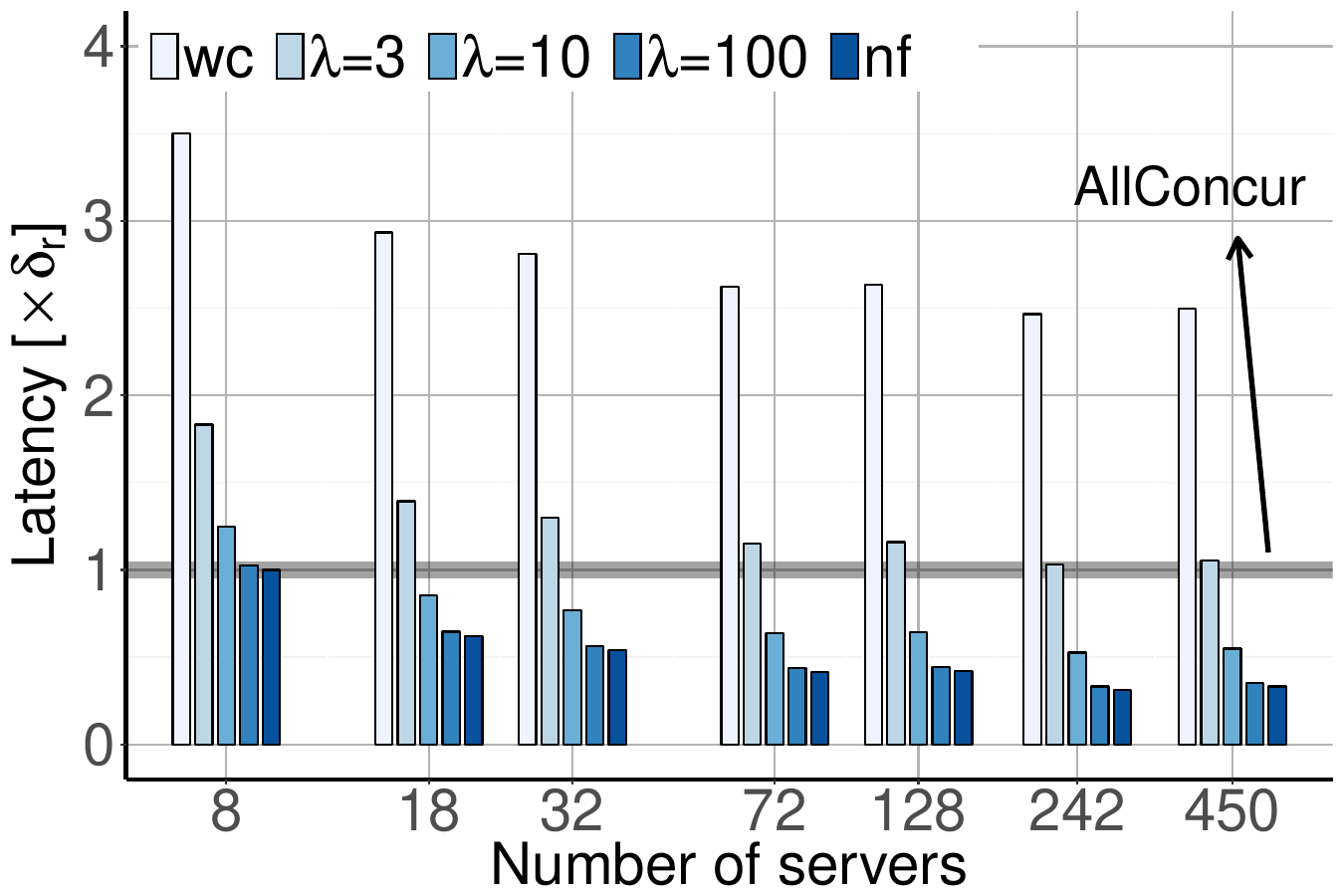}
} \hfill
\subcaptionbox{Expected throughput\label{fig:fm_expected_throughput}} {
\includegraphics[width=.35\textwidth]{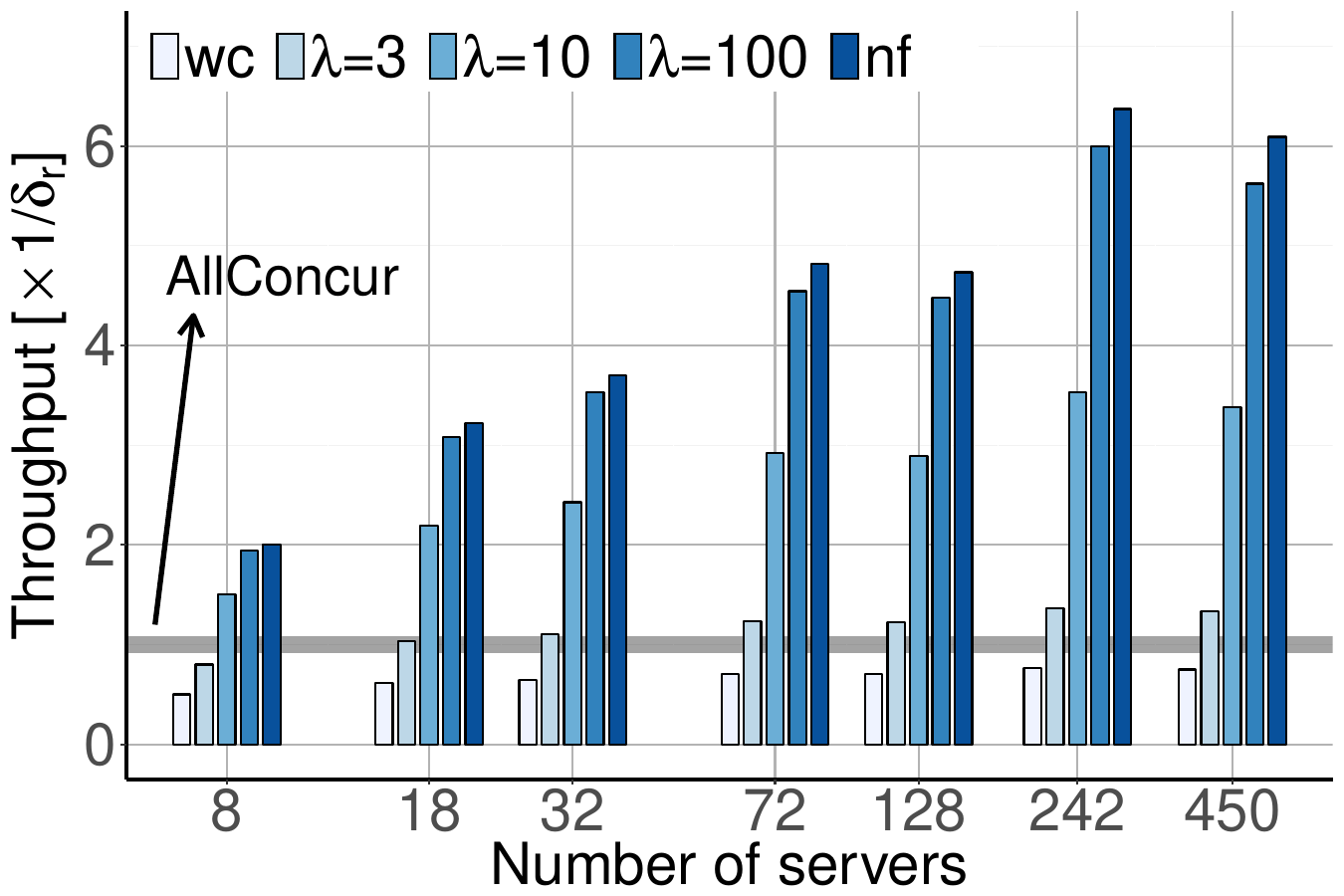}
}
\hspace*{\fill}%

\caption{\ACplus{}'s estimated expected latency and throughput (expressed as ratios to \AC{}'s values) for the non-failure scenario (nf), 
the worst-case scenario (wc), and scenarios with sequences of $\lambda$ unreliable rounds (ordered by decreasing frequency of failures).
We assume an SDC deployment with batch size four. 
}
\label{fig:ev_failure_analysis}
\end{figure}
\fi

\ifdefined\TECHREP
\textbf{Expected performance.}  
\fi
We estimate \ACplus{}'s expected performance for several scenarios with different expected failure frequencies.
If failures occur so frequently that no unreliable rounds can be started (i.e., a single sequence of reliable rounds), 
\ACplus{} is equivalent to \AC{}---a round is A-delivered once it is completed.
Thus, the estimate of the expected latency is $\delta_r$ and of the expected throughput is $1/\delta_r$. 
For comparison, Figure~\ref{fig:ev_failure_analysis} plots the latency as a factor of $\delta_r$ and the throughput as a factor of~$1/\delta_r$.
If no failures occur, the algorithm performs a sequence of unreliable rounds; for this scenario (denoted by nf), 
we use the actual measurements, reported in Figure~\ref{fig:SDC_latency}.

The worst-case scenario (denoted by wc), requires all messages to be A-delivered during reliable rounds and,
in addition, in between any two subsequent reliable rounds there must be exactly two unreliable rounds (the second one not being completed), 
i.e., the longest sequence possible without A-delivering an unreliable round. 
Such a scenario consists of the repetition of the following sequence:
\ifdefined\TECHREP
\begin{align}
\label{eq:worst_case_seq}
  & \ldots \trans \acState{e}{r-1} \trans \acState{e}{r} \fTrans \acRState{e+1}{r-1} \trans 
     \acState{e+1}{r} \trans \acState{e+1}{r+1} \fTrans \acRState{e+2}{r} \trans \ldots
\end{align}
\fi
\ifdefined\PAPER
\begin{align*}
\label{eq:worst_case_seq}
  & \ldots \trans \acState{e}{r-1} \trans \acState{e}{r} \fTrans \acRState{e+1}{r-1} \trans \nonumber \\ 
  & \quad \acState{e+1}{r} \trans \acState{e+1}{r+1} \fTrans \acRState{e+2}{r} \trans \ldots
\end{align*}
\fi
The expected latency is $3\delta_u+2\delta_r$, 
e.g., $r$'s messages are A-broadcast in $\acState{e}{r}$ but A-delivered once $\acRState{e+2}{r}$ completes.
This scenario provides also the worst-case expected throughput---only reliable rounds are A-delivered, hence $1/(2\delta_u + \delta_r)$.
In the worst case, \ACplus{} has up to $3.5\times$ higher expected latency and up to $2\times$ lower expected throughput than \AC{}.

The low performance of the worst-case scenario is due to the lack of sufficiently long sequences of unreliable rounds 
(i.e., at least three, two completed and one begun), thus never enabling the minimal-work distributed agreement of \ACplus{}.
\ifdefined\TECHREP
Let $\lambda \geq 3$ be the length of each sequence of unreliable rounds, 
i.e., the expected frequency of failures is between $1/(\delta_r+\lambda \delta_u)$ and $1/(\delta_r+(\lambda-1) \delta_u)$.
\fi
\ifdefined\PAPER
Let $\lambda \geq 3$ be the length of each sequence of unreliable rounds. 
\fi
Then, the expected latency is $2\delta_u + \frac{\delta_u+2\delta_r}{\lambda}$,
i.e., the first $\lambda-2$ unreliable rounds are A-delivered after $2\delta_u$, 
the $(\lambda-1)$-th round after $2\delta_u +\delta_r$ and the final round after $3\delta_u +\delta_r$.
Also, the expected throughput is $\frac{1-1/\lambda}{\delta_u + \delta_r/\lambda}$, 
i.e., during a period of $\lambda\delta_u + \delta_r$,  $\lambda-2$ unreliable rounds and one reliable round are A-delivered.
For $\lambda=10$, \ACplus{} has up to $1.9\times$ lower expected latency and up to $3.5\times$ higher expected throughput than \AC{}.

\ifdefined\TECHREP
\textbf{Worst-case latency for a single round.}
Let $r$ be a round and $e$ be the epoch in which $r$'s messages were A-broadcast for the first time. 
If $\acRState{e}{r}$, then $r$ is A-delivered once $\acRState{e}{r}$ completes (as in \AC{}); 
thus, the estimated latency is $\delta_r$. 
Yet, if $\acState{e}{r}$, the A-delivery of $r$ must be delayed~(\cref{sec:adeliv_msg}).
If no failures occur, $\acState{e}{r}$ will be followed by two unreliable rounds; thus, the estimated latency is $2\delta_u$. 
In case of a single failure, the worst-case estimated latency is $3\delta_u+\delta_r$
and it corresponds to the sequence in~\eqref{eq:worst_case_seq}, with the exception that $\acState{e+1}{r+1}$ completes 
i.e., $\acState{e}{r}$ is succeeded by a fail transition to $\acRState{e+1}{r-1}$, a reliable rerun
of the preceding not A-delivered unreliable round $r-1$; then, due to no other failures, round $r$ is rerun 
unreliably in epoch $e+1$ and A-delivered once $\acState{e+1}{r+1}$ completes. 
If multiple failures can occur, $\acState{e+1}{r+1}$ may be interrupted by a fail transition to $\acRState{e+2}{r}$; 
thus, in this case, the worst-case estimated latency is $3\delta_u+2\delta_r$.

For latency sensitive applications, we can reduce the worst-case latency by always rerunning rounds reliably (after a rollback). 
For the sequence in~\eqref{eq:worst_case_seq}, $\acRState{e+1}{r-1}$ would then be followed by $\acRState{e+2}{r}$.
In this case, the worst-case estimated latency is $\delta_u+2\delta_r$.
Note that if $\delta_r > 2\delta_u$, then in the case of a single failure, 
rerunning reliably is more expensive, i.e., $\delta_u+2\delta_r$ as compared with $3\delta_u+\delta_r$.
Another optimization is to rerun all (not A-delivered) unreliable rounds in one single reliable 
round. 
Let $\overline{\delta_r}$ denote the expected time of such a round; then, the worst-case estimated latency is $2\delta_u+\overline{\delta_r}$.
Note that if a skip transition is triggered while two rounds $r-1$ and $r$ are rerun in one reliable round, only round $r-1$ must be A-delivered. 
\fi

\section{Related work}


D{\'e}fago, Schiper, and Urb\'{a}n provide a general overview of atomic broadcast algorithms~\cite{Defago:2004:TOB:1041680.1041682}. 
Based on how total order is established, they consider five classes of atomic broadcast algorithms: 
fixed sequencer; moving sequencer; privilege-based; communication history; and destinations agreement. 

The first three classes, i.e., fixed sequencer~\cite{Birman:1991:LCA:128738.128742,Junqueira:2011:ZHB:2056308.2056409}, 
moving sequencer~\cite{Chang:1984:RBP:989.357400,Mao:2008:MBE:1855741.1855767}, and privilege-based~\cite{Amir:1995:TSO:210223.210224,Ekwall:2004:TAB:1032662.1034348} 
rely on a distinguished server to provide total order. 
In fixed sequencer algorithms, the order is established by an elected sequencer 
that holds the responsibility until it is suspected of having failed.
Moving sequencer algorithms, are similar, but the role of sequencer is transfered between servers, in order to distribute the load among them.
In privilege-based algorithms, senders broadcast when they are given the privilege to do so; the privilege is given to one sender at a time. 
For these algorithm classes, the work is unbalanced---the distinguished server is on the critical path for all communication, 
leading to linear work per A-broadcast message (for straightforward implementations).
\ACplus{} is leaderless; it balances the work evenly among all servers. Thus, it achieves sublinear work per A-broadcast message.

The last two classes are leaderless---total order is determined without a leader, either by the senders (communication history) or by the destinations (destinations agreement). 
Usually, in communication history algorithms, A-broadcast messages carry logical timestamps (we do not survey algorithms that rely on physical timestamps). 
The servers use these timestamps to decide when to safely deliver messages. 
D{\'e}fago, Schiper, and Urb\'{a}n  distinguish between causal history and deterministic merge algorithms~\cite{Defago:2004:TOB:1041680.1041682}.
\ifdefined\TECHREP
Causal history algorithms~\cite{Peterson:1989:PUC:65000.65001,Ezhilchelvan:1995:NFG:876885.880005,ng91ordered,dolev1993early,Moser:1993:AFT:161952.161958,Keidar2000,Guerraoui:2010:TOT:1813654.1813656} 
\fi
\ifdefined\PAPER
Causal history algorithms~\cite{Peterson:1989:PUC:65000.65001,Keidar2000,Guerraoui:2010:TOT:1813654.1813656} 
\fi
transform the partial order provided by the timestamps into total order (i.e., the causal order~\cite{Lamport:1978:TCO:359545.359563} is extended by ordering concurrent messages~\cite{Birman:1991:LCA:128738.128742}).
Yet, in such algorithms, the timestamps must provide information on every participating server; thus, the size of every message is linear in $n$.
In \ACplus{}, the size of messages is constant.

In deterministic merge algorithms~\cite{Bar-Joseph:2002:EDA:645959.676132,Chockler:1998:ATO:277697.277741}, 
the messages are timestamped independently (i.e., no causal order) and delivered according to a deterministic policy of merging 
the streams of messages coming from each server. \ACplus{} can be classified as a deterministic merge algorithm---every message is 
timestamped with the round number and the merging policy is round-robin (except for lost messages). 
The Atom algorithm~\cite{Bar-Joseph:2002:EDA:645959.676132} uses the same merging policy.
Yet, it uses an early-deciding mechanism that entails waiting for the worst-case given the actual number of failures~\cite{Dolev:1990:ESB:96559.96565}. 
\ACplus{}'s early termination mechanism does not require waiting for the worst case.
Also, the overlay network in Atom is described by a complete digraph.
Finally, methods to decrease latency by adaptively changing the merging policy (e.g., based on sending rates~\cite{Chockler:1998:ATO:277697.277741}) 
can also be applied to \ACplus{} 
(e.g., servers with slow sending rates can skip rounds). 

Conceptually, destinations agreement algorithms do not require a leader for establishing total order---the destination servers reach an agreement on the delivery order. 
In many existing destinations agreement algorithms, the servers reach agreement either through centralized mechanisms~\cite{Birman:1987:RCP:7351.7478,Luan:1990:FPA:628893.628998} 
\ifdefined\TECHREP
or by solving consensus~\cite{Chandra:1996:UFD:226643.226647,Mostefaoui:2000:LCC:826038.826941,Rodrigues:2000:ABA:850927.851790,Anceaume:1997:LSU:795670.796861}.
\fi
\ifdefined\PAPER
or by solving consensus~\cite{Chandra:1996:UFD:226643.226647,Rodrigues:2000:ABA:850927.851790}.
\fi
Most practical consensus algorithms rely on 
leader-based approaches~\cite{Lamport:1998:PP:279227.279229,Liskov2012,Junqueira:2011:ZHB:2056308.2056409,Ongaro2014,Poke:2015:DHS:2749246.2749267}, 
resulting in centralized destinations agreement algorithms and thus, unbalanced work.
\ACplus{} can be also classified as a destinations agreement algorithm---in every round, the servers agree on a set of messages to A-deliver.
Yet, in \ACplus{}, the servers reach agreement through a completely decentralized mechanism. 

\ifdefined\TECHREP
Moreover, some destinations agreement algorithms rely on the spontaneous total-order property 
(i.e., with high probability, messages broadcast in local-area networks are received in total order)
as a condition for message delivery~\cite{rodrigues1992xamp}, as an optimization~\cite{Pedone:2003:OAB:795635.795644}, 
or as an alternative to overcome the FLP~\cite{Fischer:1985:IDC:3149.214121} impossibility result~\cite{Pedone:2002:SAP:645333.649852}.
Yet, breaking this property leads to potential livelocks~\cite{rodrigues1992xamp}, the need of solving consensus~\cite{Pedone:2003:OAB:795635.795644}, 
or unbounded runs~\cite{Pedone:2002:SAP:645333.649852}. 
\ACplus{} makes no assumption of the order in which messages are received. 
\fi

\ifdefined\TECHREP
Finally, some algorithms can fit multiple classes (so called hybrid algorithms).
\fi
\ifdefined\PAPER
Finally, some algorithms can fit multiple classes.
\fi
For example, Ring-Paxos~\cite{ringpaxos} 
relies on a coordinator, but the communication is done using a logical ring, similarly to the majority of privilege-based algorithms. 
By placing the servers in a ring, Ring-Paxos achieves high-throughput~\cite{Guerraoui:2010:TOT:1813654.1813656}.
Yet, the ring topology is not suitable for large scales---the latency of message dissemination is linear in~$n$.
In \ACplus{}, servers communicate via an overlay network described by 
\ifdefined\TECHREP
any resilient digraph\footnote{The requirement for the digraph to be sparse 
is only needed for performance, i.e., it ensures that the work per broadcast message is sublinear. This requirement is not necessary for guaranteeing safety.};
\fi
\ifdefined\PAPER
any resilient digraph;
\fi
moreover, during intervals with no failures, the digraph must only be connected.
Thus, \ACplus{} enable the trade-off between high-throughput and low-latency topologies. 
\ifdefined\TECHREP
Clearly, an alternative is to use Ring-Paxos as a reliable sequencer that enable a large group of $n$ servers to reach agreement on the order of messages.
Yet, in such a deployment, the high-throughput of Ring-Paxos results only in a constant increase in performance (given by the number of servers running Ring-Paxos).
The work per A-broadcast message is still linear in~$n$.
\fi

\section{Conclusion}

In this paper we describe \ACplus{}, a leaderless concurrent atomic broadcast algorithm that 
provides high-performance fault-tolerant distributed agreement. 
%
During intervals with no failures, \ACplus{} runs in an unreliable mode
that enables minimal-work distributed agreement; when failures do occur, it automatically and safely switches to a reliable mode 
that uses the early termination mechanism of \AC{}~\cite{poke2017allconcur}.
We provide a complete design of \ACplus{} that shows how to transition from one mode to another and 
an informal proof of \ACplus{}'s correctness.

\ifdefined\TECHREP
Our performance evaluation\footnote{Our evaluation is based on the discrete event simulator OMNeT++~\cite{Varga:2008:OOS:1416222.1416290}.} 
\fi
\ifdefined\PAPER
Our performance evaluation 
\fi
of non-failure scenarios demonstrates that \ACplus{} 
achieves high throughput and low latency while scaling out to hundreds of servers deployed both inside a single datacenter and 
across multiple (geographically distributed) datacenters.
It is therefore especially well-suited for large-scale distributed agreement applications, such as distributed ledgers. 
For instance, running on $225$ servers distributed throughout five datacenters across Europe, 
with each server maintaining a copy of the ledger, 
\ACplus{} reaches a throughput of $\approx 26,000$ ($250$-byte) transactions per second, with a latency of less than $65$ms.

Overall, \ACplus{} achieves comparable throughput to AllGather~\cite{mpi-3.1}, a non-fault-tolerant distributed agreement algorithm, 
and, especially at scale, significantly higher throughput and lower latency than 
\ifdefined\TECHREP
other fault-tolerant algorithms, such as \AC{}~\cite{poke2017allconcur}, LCR~\cite{Guerraoui:2010:TOT:1813654.1813656} and Libpaxos~\cite{libpaxos}. 
\fi
\ifdefined\PAPER
other fault-tolerant algorithms. 
\fi
Moreover, our evaluation of various failure scenarios shows that \ACplus{}'s expected performance is robust with regard to occasional failures.


{\small
\textbf{Acknowledgements.}
This work was supported by the German Research Foundation (DFG) as part of the Cluster 
of Excellence in Simulation Technology (EXC 310/2) at the University of Stuttgart.
We thank Michael Resch for support; and both Jos\'{e} Gracia and George Cristian Tudor for helpful discussions. 
}

{
\ifdefined\TECHREP
\bibliographystyle{abbrv}
\bibliography{references}
\fi
\ifdefined\PAPER
\bibliographystyle{IEEEtran}
\bibliography{IEEEabrv,references}
\fi
}

\newpage

\ifdefined\APPENDIX
\appendix


\subsection{Concurrent states: description and properties}

\subsubsection{Description of concurrent states}
\label{app:conc_states}

\begin{figure*}[!tp]
\captionsetup[subfigure]{justification=centering}
\centering
\subcaptionbox{$p_j$ in $\acState{e}{r}$\label{fig:conc_states_unrel_t1}} {
\scalebox{0.80}{\input{figures/conc_states_unrel_t1.tex}}
} \hfill
\subcaptionbox{$p_j$'s $\acRState{e}{r}$ state is preceded by $\tur{}$\label{fig:conc_states_rel_t3}} {
\scalebox{0.80}{\input{figures/conc_states_rel_t3.tex}}
}\\[+1ex]

\subcaptionbox{$p_j$'s $\acRState{e}{r}$ state is preceded by $\tfr{}$\label{fig:conc_states_rel_t4}} {
\scalebox{0.80}{\input{figures/conc_states_rel_t4.tex}}
} \hfill
\subcaptionbox{$p_j$ in $\acFState{e}{r}$\label{fig:conc_states_unrel_t2}} {
\scalebox{0.80}{
\ifx\du\undefined
  \newlength{\du}
\fi
\setlength{\du}{15\unitlength}
\begin{tikzpicture}
\pgftransformxscale{1.000000}
\pgftransformyscale{-1.000000}
\definecolor{dialinecolor}{rgb}{0.000000, 0.000000, 0.000000}
\pgfsetstrokecolor{dialinecolor}
\definecolor{dialinecolor}{rgb}{1.000000, 1.000000, 1.000000}
\pgfsetfillcolor{dialinecolor}
\pgfsetlinewidth{0.050000\du}
\pgfsetdash{}{0pt}
\pgfsetdash{}{0pt}
\pgfsetroundjoin
{\pgfsetcornersarced{\pgfpoint{0.200000\du}{0.200000\du}}\definecolor{dialinecolor}{rgb}{0.847059, 0.847059, 0.847059}
\pgfsetfillcolor{dialinecolor}
\fill (17.100000\du,12.500000\du)--(17.100000\du,13.500000\du)--(18.700000\du,13.500000\du)--(18.700000\du,12.500000\du)--cycle;
}{\pgfsetcornersarced{\pgfpoint{0.200000\du}{0.200000\du}}\definecolor{dialinecolor}{rgb}{0.301961, 0.301961, 0.301961}
\pgfsetstrokecolor{dialinecolor}
\draw (17.100000\du,12.500000\du)--(17.100000\du,13.500000\du)--(18.700000\du,13.500000\du)--(18.700000\du,12.500000\du)--cycle;
}\pgfsetlinewidth{0.050000\du}
\pgfsetdash{}{0pt}
\pgfsetdash{}{0pt}
\pgfsetbuttcap
{
\definecolor{dialinecolor}{rgb}{0.301961, 0.301961, 0.301961}
\pgfsetfillcolor{dialinecolor}
\pgfsetarrowsstart{stealth}
\definecolor{dialinecolor}{rgb}{0.301961, 0.301961, 0.301961}
\pgfsetstrokecolor{dialinecolor}
\draw (17.100000\du,13.000000\du)--(16.100000\du,13.000000\du);
}
\pgfsetlinewidth{0.050000\du}
\pgfsetdash{}{0pt}
\pgfsetdash{}{0pt}
\pgfsetroundjoin
{\pgfsetcornersarced{\pgfpoint{0.200000\du}{0.200000\du}}\definecolor{dialinecolor}{rgb}{0.847059, 0.847059, 0.847059}
\pgfsetfillcolor{dialinecolor}
\fill (19.700000\du,12.500000\du)--(19.700000\du,13.500000\du)--(22.300000\du,13.500000\du)--(22.300000\du,12.500000\du)--cycle;
}{\pgfsetcornersarced{\pgfpoint{0.200000\du}{0.200000\du}}\definecolor{dialinecolor}{rgb}{0.301961, 0.301961, 0.301961}
\pgfsetstrokecolor{dialinecolor}
\draw (19.700000\du,12.500000\du)--(19.700000\du,13.500000\du)--(22.300000\du,13.500000\du)--(22.300000\du,12.500000\du)--cycle;
}\pgfsetlinewidth{0.050000\du}
\pgfsetdash{}{0pt}
\pgfsetdash{}{0pt}
\pgfsetbuttcap
{
\definecolor{dialinecolor}{rgb}{0.301961, 0.301961, 0.301961}
\pgfsetfillcolor{dialinecolor}
\pgfsetarrowsstart{stealth}
\definecolor{dialinecolor}{rgb}{0.301961, 0.301961, 0.301961}
\pgfsetstrokecolor{dialinecolor}
\draw (19.700000\du,13.000000\du)--(18.700000\du,13.000000\du);
}
\pgfsetlinewidth{0.050000\du}
\pgfsetdash{}{0pt}
\pgfsetdash{}{0pt}
\pgfsetroundjoin
{\pgfsetcornersarced{\pgfpoint{0.200000\du}{0.200000\du}}\definecolor{dialinecolor}{rgb}{0.749020, 0.749020, 0.749020}
\pgfsetfillcolor{dialinecolor}
\fill (23.500000\du,12.500000\du)--(23.500000\du,13.500000\du)--(26.100000\du,13.500000\du)--(26.100000\du,12.500000\du)--cycle;
}{\pgfsetcornersarced{\pgfpoint{0.200000\du}{0.200000\du}}\definecolor{dialinecolor}{rgb}{0.301961, 0.301961, 0.301961}
\pgfsetstrokecolor{dialinecolor}
\draw (23.500000\du,12.500000\du)--(23.500000\du,13.500000\du)--(26.100000\du,13.500000\du)--(26.100000\du,12.500000\du)--cycle;
}\pgfsetlinewidth{0.050000\du}
\pgfsetdash{}{0pt}
\pgfsetdash{}{0pt}
\pgfsetroundjoin
{\pgfsetcornersarced{\pgfpoint{0.100000\du}{0.100000\du}}\definecolor{dialinecolor}{rgb}{0.847059, 0.847059, 0.847059}
\pgfsetfillcolor{dialinecolor}
\fill (23.600000\du,12.600000\du)--(23.600000\du,13.400000\du)--(26.000000\du,13.400000\du)--(26.000000\du,12.600000\du)--cycle;
}{\pgfsetcornersarced{\pgfpoint{0.100000\du}{0.100000\du}}\definecolor{dialinecolor}{rgb}{0.301961, 0.301961, 0.301961}
\pgfsetstrokecolor{dialinecolor}
\draw (23.600000\du,12.600000\du)--(23.600000\du,13.400000\du)--(26.000000\du,13.400000\du)--(26.000000\du,12.600000\du)--cycle;
}\pgfsetlinewidth{0.050000\du}
\pgfsetdash{}{0pt}
\pgfsetdash{}{0pt}
\pgfsetmiterjoin
\pgfsetbuttcap
{
\definecolor{dialinecolor}{rgb}{0.301961, 0.301961, 0.301961}
\pgfsetfillcolor{dialinecolor}
\pgfsetarrowsend{stealth}
{\pgfsetcornersarced{\pgfpoint{0.000000\du}{0.000000\du}}\definecolor{dialinecolor}{rgb}{0.301961, 0.301961, 0.301961}
\pgfsetstrokecolor{dialinecolor}
\draw (14.800000\du,13.500000\du)--(14.800000\du,14.300000\du)--(24.800000\du,14.300000\du)--(24.800000\du,13.500000\du);
}}
\pgfsetlinewidth{0.050000\du}
\pgfsetdash{}{0pt}
\pgfsetdash{}{0pt}
\pgfsetbuttcap
{
\definecolor{dialinecolor}{rgb}{0.301961, 0.301961, 0.301961}
\pgfsetfillcolor{dialinecolor}
\pgfsetarrowsstart{stealth}
\definecolor{dialinecolor}{rgb}{0.301961, 0.301961, 0.301961}
\pgfsetstrokecolor{dialinecolor}
\draw (23.500000\du,13.000000\du)--(22.300000\du,13.000000\du);
}
\pgfsetlinewidth{0.050000\du}
\pgfsetdash{}{0pt}
\pgfsetdash{}{0pt}
\pgfsetmiterjoin
\pgfsetbuttcap
{
\definecolor{dialinecolor}{rgb}{0.301961, 0.301961, 0.301961}
\pgfsetfillcolor{dialinecolor}
\pgfsetarrowsend{stealth}
{\pgfsetcornersarced{\pgfpoint{0.000000\du}{0.000000\du}}\definecolor{dialinecolor}{rgb}{0.301961, 0.301961, 0.301961}
\pgfsetstrokecolor{dialinecolor}
\draw (17.900000\du,12.500000\du)--(17.900000\du,11.700000\du)--(24.800000\du,11.700000\du)--(24.800000\du,12.500000\du);
}}
\pgfsetlinewidth{0.050000\du}
\pgfsetdash{}{0pt}
\pgfsetdash{}{0pt}
\pgfsetroundjoin
{\pgfsetcornersarced{\pgfpoint{0.200000\du}{0.200000\du}}\definecolor{dialinecolor}{rgb}{0.749020, 0.749020, 0.749020}
\pgfsetfillcolor{dialinecolor}
\fill (13.500000\du,12.500000\du)--(13.500000\du,13.500000\du)--(16.100000\du,13.500000\du)--(16.100000\du,12.500000\du)--cycle;
}{\pgfsetcornersarced{\pgfpoint{0.200000\du}{0.200000\du}}\definecolor{dialinecolor}{rgb}{0.301961, 0.301961, 0.301961}
\pgfsetstrokecolor{dialinecolor}
\draw (13.500000\du,12.500000\du)--(13.500000\du,13.500000\du)--(16.100000\du,13.500000\du)--(16.100000\du,12.500000\du)--cycle;
}\pgfsetlinewidth{0.050000\du}
\pgfsetdash{}{0pt}
\pgfsetdash{}{0pt}
\pgfsetroundjoin
{\pgfsetcornersarced{\pgfpoint{0.100000\du}{0.100000\du}}\definecolor{dialinecolor}{rgb}{0.847059, 0.847059, 0.847059}
\pgfsetfillcolor{dialinecolor}
\fill (13.600000\du,12.600000\du)--(13.600000\du,13.400000\du)--(16.000000\du,13.400000\du)--(16.000000\du,12.600000\du)--cycle;
}{\pgfsetcornersarced{\pgfpoint{0.100000\du}{0.100000\du}}\definecolor{dialinecolor}{rgb}{0.301961, 0.301961, 0.301961}
\pgfsetstrokecolor{dialinecolor}
\draw (13.600000\du,12.600000\du)--(13.600000\du,13.400000\du)--(16.000000\du,13.400000\du)--(16.000000\du,12.600000\du)--cycle;
}\pgfsetlinewidth{0.070000\du}
\pgfsetdash{{1.000000\du}{1.000000\du}}{0\du}
\pgfsetdash{{0.200000\du}{0.200000\du}}{0\du}
\pgfsetmiterjoin
\pgfsetbuttcap
{
\definecolor{dialinecolor}{rgb}{0.533333, 0.356863, 0.035294}
\pgfsetfillcolor{dialinecolor}
\pgfsetarrowsend{stealth}
\definecolor{dialinecolor}{rgb}{0.533333, 0.356863, 0.035294}
\pgfsetstrokecolor{dialinecolor}
\pgfpathmoveto{\pgfpoint{15.700000\du}{12.300000\du}}
\pgfpathcurveto{\pgfpoint{15.900000\du}{11.900000\du}}{\pgfpoint{17.000000\du}{11.900000\du}}{\pgfpoint{17.500000\du}{12.400000\du}}
\pgfusepath{stroke}
}
\definecolor{dialinecolor}{rgb}{0.000000, 0.000000, 0.000000}
\pgfsetstrokecolor{dialinecolor}
\node at (14.800000\du,12.995000\du){$e : r\!-\!1$};
\definecolor{dialinecolor}{rgb}{0.000000, 0.000000, 0.000000}
\pgfsetstrokecolor{dialinecolor}
\node at (17.900000\du,12.995000\du){$e : r$};
\definecolor{dialinecolor}{rgb}{0.000000, 0.000000, 0.000000}
\pgfsetstrokecolor{dialinecolor}
\node at (21.000000\du,12.995000\du){$e : r\!+\!1$};
\definecolor{dialinecolor}{rgb}{0.000000, 0.000000, 0.000000}
\pgfsetstrokecolor{dialinecolor}
\node at (24.800000\du,12.995000\du){$e\!+\!1 : r$};
\definecolor{dialinecolor}{rgb}{0.000000, 0.000000, 0.000000}
\pgfsetstrokecolor{dialinecolor}
\node at (15.300000\du,13.922500\du){\small{fail}};
\definecolor{dialinecolor}{rgb}{0.000000, 0.000000, 0.000000}
\pgfsetstrokecolor{dialinecolor}
\node at (22.800000\du,12.622500\du){\small{fail}};
\definecolor{dialinecolor}{rgb}{0.000000, 0.000000, 0.000000}
\pgfsetstrokecolor{dialinecolor}
\node at (18.400000\du,12.022500\du){\small{fail}};
\pgfsetlinewidth{0.000000\du}
\pgfsetdash{}{0pt}
\pgfsetdash{}{0pt}
\pgfsetmiterjoin
\pgfsetbuttcap
\definecolor{dialinecolor}{rgb}{0.301961, 0.301961, 0.301961}
\pgfsetfillcolor{dialinecolor}
\fill (18.600000\du,13.300000\du)--(18.400000\du,13.200000\du)--(18.400000\du,13.400000\du)--cycle;
\definecolor{dialinecolor}{rgb}{0.301961, 0.301961, 0.301961}
\pgfsetstrokecolor{dialinecolor}
\draw (18.600000\du,13.300000\du)--(18.400000\du,13.200000\du)--(18.400000\du,13.400000\du)--cycle;
\end{tikzpicture}}
}\\[+1ex] 

\subcaptionbox{$p_j$'s $\acRState{e}{r}$ state is preceded by $\trr{}$\label{fig:conc_states_rel_t5}} {
\scalebox{0.80}{\input{figures/conc_states_rel_t5.tex}}
} \hfill
\subcaptionbox{$p_j$'s $\acRState{e}{r}$ state is preceded by $\tsk{}$\label{fig:conc_states_rel_t6}} {
\scalebox{0.80}{\input{figures/conc_states_rel_t6.tex}}
}
\caption{The possible states of a non-faulty server $p_i$, while a non-faulty server $p_j$ is in epoch $e$ and round $r$. 
Boxes indicate states---single-edged for unreliable rounds and double-edged for reliable rounds. 
A triangle in the lower-right corner indicates the first state in a sequence of unreliable rounds.
Straight arrows indicate $p_i$'s possible transitions; curved dashed arrows indicate $p_j$'s latest transition.}
\label{fig:conc_states}
\end{figure*}

We are interested in what states a non-faulty server $p_i$ can be, given that a non-faulty server $p_j$ is either in $\acState{e}{r}$ or in $\acRState{e}{r}$.
First, let $p_j$ be in $\acState{e}{r}$. Then, $\acState{e}{r}$ is preceded by either $\tuu{}$ or $\trf{}$
(see the curved dashed arrows in Figures~\ref{fig:conc_states_unrel_t1} and~\ref{fig:conc_states_unrel_t2}).
$\tuu{}$ entails $p_i$ has started $\acState{e}{r-1}$, while $\trf{}$ entails $p_i$ has started $\acRState{e}{r-1}$.
Using both this information and the fact that $p_j$ already started $\acState{e}{r}$, 
we can deduce all the possible states of $p_i$ (see the boxes in Figures~\ref{fig:conc_states_unrel_t1} and~\ref{fig:conc_states_unrel_t2}).
Since $p_j$ started $\acState{e}{r}$, $p_i$ can be in one of the following states with 
unreliable rounds: $\acState{e}{r-1}$ (if $p_j$ is not in $\acFState{e}{r}$); $\acState{e}{r}$; or $\acState{e}{r+1}$.
Moreover, if $p_i$ receives a failure notification, it moves (from any of these states) 
to a state with a reliable round (indicated by double-edged boxes). 
Note that a fail transition from $\acState{e}{r-1}$ depends on whether round $r-1$ is the first in a sequence of unreliable rounds
(see the lower-right corner triangle in Figure~\ref{fig:conc_states_unrel_t1}).
Finally, $p_i$ can skip a reliable round, i.e., either $r-2$ or $r-1$ (see Figure~\ref{fig:conc_states_unrel_t1}); 
the precondition is for at least one other server to have A-delivered the corresponding unreliable round. 
Notice though that $p_i$ cannot skip both rounds---this would imply both $\acState{e}{r-2}$ and $\acState{e}{r-1}$ 
to be A-delivered, which is not possible since $p_i$ does not start $\acState{e}{r}$.

Second, let $p_j$ be in $\acRState{e}{r}$. Clearly, $p_i$ can be in $\acRState{e}{r}$; 
also, since $p_j$ started $\acRState{e}{r}$, $p_i$ can be ahead in either $\acFState{e}{r+1}$ 
or $\acRState{e+1}{r+1}$.
To infer the other possible states of $p_i$, we consider the four transitions that can precede $p_j$'s $\acRState{e}{r}$ state 
(see the curved dashed arrows in Figures~\ref{fig:conc_states_rel_t3}, \ref{fig:conc_states_rel_t4}, \ref{fig:conc_states_rel_t5}, and~\ref{fig:conc_states_rel_t6}).
Note that $\acRState{e}{r}$, $\acFState{e}{r+1}$, and $\acRState{e+1}{r+1}$ are present in all four figures.

$\tur{}$ entails $p_j$ has started $\acState{e-1}{r+1}$ and hence, $p_i$ started either $\acState{e-1}{r}$ or $\acFState{e-1}{r}$ (cf. Proposition~\ref{prop:one_end_all_start}). 
Receiving a failure notification while in either of 
these states triggers (eventually) a transition to $\acRState{e}{r}$ 
(either directly from $\acFState{e-1}{r}$ or through a skip transition from $\acState{e-1}{r}$ via $\acRState{e}{r-1}$).
Moreover, if the failure notification is delayed, $p_i$ can move to $\acState{e-1}{r+1}$ or even to $\acState{e-1}{r+2}$. 
Finally, receiving a failure notification while in $\acState{e-1}{r+2}$ triggers a transition to $\acRState{e}{r+1}$;
$\acRState{e}{r+1}$ can also be reached by a skip transition from $\acRState{e}{r}$ (see Figure~\ref{fig:skip_trans}).
Both $\tfr{}$ and $\trr{}$ entail $p_i$ has started $\acRState{e-1}{r-1}$, since $p_j$ completed it (cf. Proposition~\ref{prop:one_end_all_start}).
In addition, $\tfr{}$ entails $p_j$ started $\acState{e-1}{r}$ before moving to $\acRState{e}{r}$;
thus, in this case, $p_i$ can also start $\acState{e-1}{r+1}$ (after completing $\acState{e-1}{r}$).
$\tsk{}$ entails at least one server completed $\acState{e-1}{r}$ (see Figure~\ref{fig:skip_trans})
and thus, $p_i$ has started $\acState{e-1}{r}$. Note that the state transitions illustrated in Figure~\ref{fig:conc_states_rel_t6}
are also included in Figure~\ref{fig:conc_states_rel_t3}.

In summary, while $p_j$ is in $\acState{e}{r}$, $p_i$ can be in four states with unreliable rounds, 
i.e., $\acState{e}{r-1}$, $\acFState{e}{r-1}$, $\acState{e}{r}$, and $\acState{e}{r+1}$,
and in four states with reliable rounds, i.e., $\acRState{e}{r-1}$, $\acRState{e+1}{r-2}$, $\acRState{e+1}{r-1}$, and $\acRState{e+1}{r}$.
Also, while $p_j$ is in $\acRState{e}{r}$, $p_i$ can be in five states with unreliable rounds, 
i.e., $\acState{e-1}{r}$, $\acFState{e-1}{r}$, $\acState{e-1}{r+1}$, $\acState{e-1}{r+2}$ and $\acFState{e}{r+1}$,
and in five states with reliable rounds, i.e., $\acRState{e-1}{r-1}$, $\acRState{e}{r-1}$, $\acRState{e}{r}$, $\acRState{e}{r+1}$ and $\acRState{e+1}{r+1}$.

\subsubsection{Informal proofs of the concurrency properties}
\label{app:conc_proofs}

\propUniqueState*
\begin{IEEEproof}
W.l.o.g., we assume $p_i$ is in state $\acState{e}{r}$ and $p_j$ is in state $\acRState{e}{r}$.
Clearly, $p_i$'s state $\acState{e}{r}$ is preceded by either a $\tuu{}$ transition (and thus by $\acState{e}{r-1}$) 
or a $\trf{}$ transition (and thus by $\acRState{e}{r-1}$).
This entails, $p_j$ at least started either $\acState{e}{r-1}$ or $\acRState{e}{r-1}$ (cf. Proposition~\ref{prop:one_end_all_start}). 
Moreover, $p_j$'s state $\acRState{e}{r}$ can only be preceded by either a fail or a skip transition. 
Yet, a fail transition from either $\acState{e}{r-1}$ or $\acRState{e}{r-1}$ is not possible since it would entail an increase in epoch;
also, a skip transition cannot be preceded by an unreliable round. 
Thus, $p_j$ went through $\acRState{e}{r-1} \sTrans \acRState{e}{r}$ and, as a result, 
$p_i$ went through $\acRState{e}{r-1} \trans \acState{e}{r}$.

The skip transition requires at least one non-faulty server to have the following fail transition 
$\acState{e-1}{r+1} \fTrans \acRState{e}{r}$ (see Figure~\ref{fig:skip_trans}).
Let $p_k$ be such a server; clearly, $p_k$ has not started $\acRState{e}{r-1}$.
Yet, this contradicts $p_i$'s transition $\acRState{e}{r-1} \trans \acState{e}{r}$ (cf. Proposition~\ref{prop:one_end_all_start}).
\end{IEEEproof}

\propUnrelMsg*
\begin{IEEEproof}
First, we assume $\tilde{e} < e$. Clearly, $e \neq 1$. Thus, $p_j$'s state $\acState{e}{r}$ 
is preceded by a state $\acRState{e}{r^\prime < r}$ that increased the epoch to $e$. 
This entails $p_i$ started state $\acRState{e}{r^\prime}$ (cf. Proposition~\ref{prop:one_end_all_start}). 
Moreover, $p_i$ started this state before $p_j$ sent $m_j^{(e,r)}$. 
Yet, this contradicts the assumption that $p_i$ received the message in epoch $\tilde{e} < e$.

Then, we assume  $\tilde{e} = e \wedge \tilde{r} < r \nRightarrow \tilde{r}=r-1$; 
thus, $\tilde{e} = e$ and $\tilde{r} < r-1$. 
Yet, $p_j$'s state $\acState{e}{r}$ is preceded by either $\acState{e}{r-1}$ or $\acRState{e}{r-1}$.
In both cases, $p_i$ was in epoch $e$ and round $r-1$ before receiving $m_j^{(e,r)}$ in epoch $e$ 
and round $\tilde{r}$, which contradicts $\tilde{r} < r-1$.
\end{IEEEproof}

\begin{figure}[!tp]
\centering
\scalebox{0.80}{\input{figures/fail_notif_epoch_unique.tex}}
\caption{A non-faulty server $p_i$ is in state $\acState{e-1}{r}$; $p_i$ received a failure notification $\phi$ in a preceding state.
While in the completed reliable round of epoch $e^\prime$, $\phi \in F_i$; 
after the reliable round of epoch $e^{\prime\prime}$ is completed, $\phi$ is removed from $F_i$;
after the reliable round of epoch $e-1$ is completed, $F_i = \emptyset$.
If $e^\prime = e^{\prime\prime} = e-1$, the three reliable rounds refer to the same actual round.
Empty rectangles indicate unreliable rounds; filled (gray) rectangles indicate reliable rounds.
Rectangles with solid edges indicate completed round, while with dashed edges indicate interrupted rounds.}
\label{fig:fail_notif_epoch_unique}
\end{figure}

To prove Theorem~\ref{prop:rel_msg} we introduce the following lemma:

\begin{lemma}
\label{lemma:fail_trigger}
Let $p_i$ be a non-faulty server in state $\acState{e-1}{r}$;
let $p_j$ be another non-faulty server in epoch $e$.
Let $\phi(a,b)$ be a notification sent by $b$ indicating $a$'s failure; 
$\phi(a,b)$ triggered $p_j$'s transition from epoch $e-1$ to $e$. 
Then, $p_i$ cannot receive $\phi(a,b)$ in any state that precedes $\acState{e-1}{r}$.
\end{lemma}
\begin{IEEEproof}
We assume $p_i$ receives $\phi(a,b)$ in a state preceding $\acState{e-1}{r}$ (see Figure~\ref{fig:fail_notif_epoch_unique}). 
Let $\acRState{e^\prime}{\ast}$, with $e^\prime \leq e-1$ and $\ast$ a placeholder for any round, 
be either the state in which $p_i$ receives $\phi(a,b)$ or the state in which $p_i$ moves 
after receiving $\phi(a,b)$ (i.e., a fail transition).
Clearly, while $p_i$ is in $\acRState{e^\prime}{\ast}$, $\phi(a,b) \in F_i$. 
Moreover, $\acState{e-1}{r}$ must be preceded by a reliable round followed by a $\trf{}$ transition, 
which requires $F_i=\emptyset$~(\cref{sec:state_machine}).
Thus, in the sequence of states from $\acRState{e^\prime}{\ast}$ to $\acState{e-1}{r}$ there is 
a state $\acRState{e^{\prime\prime}}{\ast}$, with $e^\prime \leq e^{\prime\prime} \leq  e-1$, 
whose completion results in the removal of $\phi(a,b)$ from~$F_i$.
This means that when $p_i$ A-delivers state $\acRState{e^{\prime\prime}}{\ast}$, 
it removes from the system at least one of the servers $a$ and $b$~(\cref{sec:round_based}).
Since $p_i$ completes $\acRState{e^{\prime\prime}}{\ast}$ and $e^{\prime\prime} < e$, 
$p_j$ also completes (and A-delivers) $\acRState{e^{\prime\prime}}{\ast}$.  
Thus, $p_j$ removes the same servers since set agreement is ensured by early termination~\cite{poke2017allconcur,allconcur_tla}.
Yet, this leads to a contradiction---$\phi(a,b)$ was invalidated during epoch $e-1$ and therefore, cannot trigger $p_j$'s transition to epoch~$e$.
\end{IEEEproof}

\propRelMsg*
\begin{IEEEproof}
We assume $\tilde{e} < e \nRightarrow (\tilde{e} = e-1 \wedge \acRState{\tilde{e}}{\tilde{r}} \wedge  \tilde{r} = r-1)$;
thus, $\tilde{e} < e-1 \vee \acState{\tilde{e}}{\tilde{r}} \vee \tilde{r} \neq r-1$. 
First, we assume $\tilde{e} < e-1$. Yet, $p_j$'s state $\acRState{e}{r}$ is preceded by a state in epoch $e-1$; 
note that $\acRState{e}{r}$ cannot be the initial state since $m_j^{(e,r)}$ is A-broadcast in $\acRState{e}{r}$.
As a result, $p_j$ completes at least one state in epoch $e-1$ before sending $m_j^{(e,r)}$. 
This means, $p_i$ starts at least one state in epoch $e-1$ before receiving $m_j^{(e,r)}$ in epoch $\tilde{e}$, 
which contradicts $\tilde{e} < e-1$. Thus, $\tilde{e} = e-1$.

Second, we assume $p_i$ receives $m_j^{(e,r)}$ in $\acState{e-1}{\tilde{r}}$.
Let~$\phi$ be the failure notification that triggers $p_j$'s transition from epoch $e-1$ to epoch $e$.
Then, $p_j$ sends $\phi$ to $p_i$ before it sends $m_j^{(e,r)}$; 
thus, $p_i$ receives $\phi$ either in $\acState{e-1}{\tilde{r}}$ or in a preceding state.
On the one hand, a failure notification received in $\acState{e-1}{\tilde{r}}$ increases the epoch, 
which contradicts the assumption that $p_i$ receives $m_j^{(e,r)}$ in epoch $e-1$.
On the other hand, receiving $\phi$ in a state preceding $\acState{e-1}{\tilde{r}}$ 
contradicts the result of Lemma~\ref{lemma:fail_trigger}. Thus, $p_i$ receives $m_j^{(e,r)}$ in $\acRState{e-1}{\tilde{r}}$.

Finally, we assume $\tilde{r} \neq r-1$. According to the transitions leading to reliable rounds, 
$p_j$'s state $\acRState{e}{r}$ is preceded by a completed state with epoch $e-1$ 
and round either $r$ or $r-1$~(\cref{sec:state_machine}).
Thus, $p_i$ starts a state with epoch $e-1$ and round either $r$ or $r-1$ 
before receiving $m_j^{(e,r)}$ in $\acRState{e-1}{\tilde{r}}$.
For $\acRState{e-1}{\tilde{r}}$ to be preceded by a state with epoch $e-1$, 
$\acRState{e-1}{\tilde{r}}$ must be preceded by a skip transition;
yet, in such a case the preceding state cannot be completed by any non-faulty server~(\cref{sec:state_machine}).
Thus, $p_j$'s state $\acRState{e}{r}$ is preceded by $\acRState{e-1}{r}$ (since $\tilde{r} \neq r-1$).
Yet, there is no sequence of transitions from $\acRState{e-1}{r}$ to $\acRState{e}{r}$; 
in other words, when $p_j$ completes $\acRState{e-1}{r}$ it A-delivers round $r$ and thus it 
cannot rerun round $r$ in a subsequent epoch. Thus, $\tilde{r} = r-1$.
\end{IEEEproof}

\propRelMsgSameEpoch*
\begin{IEEEproof}
First, we assume $p_i$ is in state $\acState{e}{\tilde{r}}$. 
Then, $\acState{e}{\tilde{r}}$ is preceded by a completed state $\acRState{e}{r^{\prime} < \tilde{r}}$. 
Yet, this implies that $p_j$ started $\acRState{e}{r^{\prime}}$ (cf. Proposition~\ref{prop:one_end_all_start}). 
Moreover, $p_j$ could not have skipped $\acRState{e}{r^{\prime}}$, since $p_i$ completed it. 
Thus, $p_j$'s state $\acRState{e}{r}$ is preceded by a completed state $\acRState{e}{r^{\prime}}$ (since $r^{\prime} < r$), 
which contradicts the definition of an epoch. As a result, $p_i$ is in state $\acRState{e}{\tilde{r}}$.

Second, we assume $\tilde{r} < r-1$. Thus, $p_i$ and $p_j$ are in the same epoch, both in reliable rounds, 
but with at least another round between them. We show that this is not possible. 
According to the transitions leading to reliable rounds,  
$p_j$'s $\acRState{e}{r}$ is preceded by either a completed state, 
i.e., $\acState{e-1}{r+1}$, $\acState{e-1}{r}$ or $\acRState{e-1}{r-1}$,
or a skipped state, i.e., $\acRState{e}{r-1}$~(\cref{sec:state_machine}).
On the one hand, a completed preceding state $s$ entails that $p_i$ started $s$ (cf. Proposition~\ref{prop:one_end_all_start}); 
moreover, since $p_i$ is in epoch $e$ and the state $s$ has epoch $e-1$, $p_i$ also completed $s$. 
Yet, if $p_i$ completed any of the $\acState{e-1}{r+1}$, $\acState{e-1}{r}$ and $\acRState{e-1}{r-1}$ states, 
there is no reason to rerun (in epoch $e$) round $\tilde{r} < r-1$.
On the other hand, if $p_j$'s $\acRState{e}{r}$ is preceded by a skip transition, then
at least one non-faulty server completed $\acState{e-1}{r}$ and thus, 
$p_i$ completed $\acState{e-1}{r-1}$ (see Figure~\ref{fig:skip_trans}). 
As a result, again there is no reason for $p_i$ to rerun (in epoch $e$) round $\tilde{r} < r-1$.
\end{IEEEproof}


\subsection{AllConcur+: design details}
\label{app:design_details}

In this section, we provide the details of \ACplus{}'s design as a non-uniform atomic broadcast algorithm. 
For ease of presentation, the following description omits both the forward-backward mechanism required by $\Diamond\mathcal{P}$~(\cref{sec:allconcur_overview})
and the mechanism for updating 
\ifdefined\TECHREP
$G_R$~(\cref{sec:update_Gf}). 
\fi
\ifdefined\PAPER
$G_R$~(\cref{sec:acp_design}). 
\fi
Therefore, we assume both $\mathcal{P}$ and no more than $f$ failures.
First, we describe the variables used and their initial values; 
also, we outline the main loop and the main communication primitives, i.e., \broadcast{}, \rbroadcast{} and \tobroadcast{}.
Then, we split the description of the design into the following points:
(1)~handling of the three possible events---receiving an unreliable message, receiving a reliable message, and receiving a failure notification;
(2)~conditions for completing and A-delivering a round;
(3)~handling of premature messages;
and (4)~updating the tracking digraphs.
Table~\ref{tab:notations} summarizes the digraph notations used throughout this section. 

\ifdefined\PAPER
%
\begin{table}[!tp]
\centering
\begin{tabular}{ l l | l l }

  \cmidrule[1.5pt](){1-4}

  \textbf{Notation} &
  \textbf{Description} &
  \textbf{Notation} &
  \textbf{Description} \\
  \hline

  \rowcolor{DarkGray}  
  $G_U$ & 
  unreliable digraph &  
  $G_R$ & 
  reliable digraph \\
  
  $\mathcal{V}(G)$ & 
  vertices  &
  $\mathcal{E}(G)$ & 
  edges \\
  
  \rowcolor{DarkGray}  
  $d(G)$ &
  degree &
  $\kappa(G)$ &
  vertex-connectivity \\
  
  $v^+(G)$ &
  $v$'s successors &
  $\pi_{u,v}$ &
  path from $u$ to $v$ \\

  \cmidrule[1.5pt](){1-4} 
\end{tabular}
  \caption{Digraph notations.}
\label{tab:notations}
\end{table}
\fi

\ifdefined\TECHREP
\begin{table}[!tp]
\centering
\begin{tabular}{ l l | l l | l l | l l }

  \cmidrule[1.5pt](){1-8}

  \textbf{Notation} &
  \textbf{Description} &
  \textbf{Notation} &
  \textbf{Description} &
  \textbf{Notation} &
  \textbf{Description} &
  \textbf{Notation} &
  \textbf{Description} \\
  \hline

  \rowcolor{DarkGray}  
  $G_U$ & 
  unreliable digraph &  
  $G_R$ & 
  reliable digraph &  
  $\mathcal{V}(G)$ & 
  vertices  &
  $\mathcal{E}(G)$ & 
  edges \\
  
  $d(G)$ &
  degree &
  $\kappa(G)$ &
  vertex-connectivity &  
  $v^+(G)$ &
  $v$'s successors &
  $\pi_{u,v}$ &
  path from $u$ to $v$ \\

  \cmidrule[1.5pt](){1-8} 
\end{tabular}
  \caption{Digraph notations.}
\label{tab:notations}
\end{table}
\fi

\ifdefined\PAPER

\fi


\subsubsection{Overview} 

\ACplus{}'s design is outlined in Algorithm~\ref{alg:allconcur}; the code is executed by a server $p_i$. 
The variables used are described in the Input section. 
Apart from $m_j^{(e,r)}$, which denotes a message A-broadcast by $p_j$ in epoch $e$ and round~$r$, all variables are local to $p_i$:
$\tilde{e}$ denotes $p_i$'s  epoch; $\tilde{r}$ denotes $p_i$'s  round; $M_i$ contains messages received (by $p_i$) during the current state;
$M_i^\mathit{prev}$ contains messages received (by $p_i$) in the previous completed, but not yet A-delivered round (if any); 
$M_i^\mathit{next}$ contains messages received prematurely (by $p_i$), i.e., $m_j^{(e,r)} \in M_i^\mathit{next} \Rightarrow e > \tilde{e} \vee (e = \tilde{e} \wedge r > \tilde{r})$;
$F_i$ contains $p_i$'s valid failure notifications; 
and $\mathbf{g_i}$ are $p_i$'s tracking digraphs.

Initially, $p_i$ is in state $\acRState{1}{0}$; before starting the main loop, it moves to $\acState{1}{1}$ (line~\ref{line:first_trans}).
In the main loop (line~\ref{line:main_loop}), $p_i$ A-broadcasts its own message $m_i^{(\tilde{e},\tilde{r})}$, if it is neither empty, nor was already A-broadcast.
Depending on \ACplus{}'s mode, the A-broadcast primitive is either a broadcast that uses $G_U$ or 
an R-broadcast that uses $G_R$ (lines~\ref{line:abcast_start}--\ref{line:abcast_end}).
Also, in the main loop, $p_i$ handles received messages (unreliable and reliable) and failure notifications. 
Note that messages are uniquely identified by the tuple (\emph{source id}, \emph{epoch number}, \emph{round number}, \emph{round type}), 
while failure notifications by the tuple (\emph{target id}, \emph{owner id}). 

\begin{algorithm}[!tb]
\scriptsize
\DontPrintSemicolon

\SetKwData{Input}{$M_i$}
\SetKwData{Fails}{$F_i$}
\SetKwData{aServer}{$p_\ast$}
\SetKwData{aMessage}{$m_\ast$}
\SetKwData{Suspected}{$\mathit{Q}$}
\SetKwData{InputGraph}{$\mathbf{g_i}$}
\SetKwData{FastDigraph}{$G_U$}
\SetKwData{FTDigraph}{$G_R$}
\SetKwData{Epoch}{$\tilde{e}$}
\SetKwData{Round}{$\tilde{r}$}
\SetKwData{Vertices}{$\mathcal{V}$}

\SetKwComment{Comment}{\{}{\}}

\SetKwData{Terminate}{$\mathit{try\_to\_Adeliver}$}
\SetKwData{Done}{$\mathit{done}$}
\SetKwData{Continue}{$\mathit{continue}$}

\SetKwFunction{HandleBCAST}{HandleBCAST}
\SetKwFunction{HandleRBCAST}{HandleRBCAST}
\SetKwFunction{HandleFAIL}{HandleFAIL}

\SetKwFor{Loop}{loop}{}{}


\KwIn{\FastDigraph ; \FTDigraph  \Comment*[f]{$p_i$'s unreliable and reliable digraphs}
  \newline $\acRState{\Epoch\leftarrow1}{\Round\leftarrow0}$ \Comment*[f]{$p_i$'s initial state}
  \newline $\Input\leftarrow\emptyset$ \Comment*[f]{messages received by $p_i$ in current state}
  \newline $\Input^{\mathit{prev}}\leftarrow\emptyset$ \Comment*[f]{messages received by $p_i$ in previous not A-delivered (unreliable) round}
  \newline $\Input^{\mathit{next}}\leftarrow\emptyset$ \Comment*[f]{messages received prematurely by $p_i$}
  \newline $\Fails\leftarrow\emptyset$ \Comment*[f]{$p_i$'s known failure notifications}
  \newline $\Vertices(\InputGraph[p_\ast])\leftarrow\{p_\ast\},\,\forall p_\ast \in \Vertices(\FTDigraph)$
  \Comment*[f]{tracking digraphs}
  \newline $m_j^{(e,r)}$ \Comment*[f]{$p_j$'s message while in $\acState{e}{r}$ or $\acRState{e}{r}$}
  } 
\BlankLine

$\acRState{\Epoch}{\Round} \trans \acState{\Epoch}{\Round+1}$ \Comment*[r]{transition from initial state $\acRState{1}{0}$}
\label{line:first_trans}
\Loop{}{\label{line:main_loop}
  \lIf
  {$m_i^{(\Epoch,\Round)}$ not empty \normalfont{ and } $m_i^{(\Epoch,\Round)} \notin \Input$}
    {\tobroadcast{m_i^{(\Epoch,\Round)}}}
  \lIf 
  {receive $\langle \mathit{BCAST},\, m_j^{(e,r)}\rangle$}
  {\HandleBCAST{$m_j^{(e,r)}$}}
  \lIf 
  {receive $\langle \mathit{RBCAST},\, m_j^{(e,r)}\rangle$}
  {\HandleRBCAST{$m_j^{(e,r)}$}}
  \lIf 
  {receive $\langle \mathit{FAIL},\, p_j,\, p_k \in p_j^+(\FTDigraph) \rangle$}
  {\HandleFAIL{$p_j$,$p_k$}}
}

\BlankLine


\SetKwBlock{tobcast}{def \tobroadcast{m_j^{(e,r)}}:}{end}
\tobcast
{\label{line:abcast_start}
  \lIf{$\acState{\Epoch}{\Round}$} {\broadcast{m_j^{(e,r)}}}
  \lElseIf{$\acRState{\Epoch}{\Round}$} {\rbroadcast{m_j^{(e,r)}}}
}

\SetKwBlock{bcast}{def \broadcast{m_j^{(e,r)}}:}{end}
\bcast
{
  \lIf{$m_j^{(e,r)} \notin \Input$}{\textbf{send} $\langle \mathit{BCAST},\, m_j^{(e,r)} \rangle$ \textbf{to} $p_i^+(\FastDigraph)$}
  $\Input \leftarrow \Input \cup \{m_j^{(e,r)}\}$\;
}

\SetKwBlock{rbcast}{def \rbroadcast{m_j^{(e,r)}}:}{end}
\rbcast
{
  \lIf{$m_j^{(e,r)} \notin \Input$}{\textbf{send} $\langle \mathit{RBCAST},\, m_j^{(e,r)} \rangle$ \textbf{to} $p_i^+(\FTDigraph)$}
  $\Input \leftarrow \Input \cup \{m_j^{(e,r)}\}$\;
  $\Vertices(\InputGraph[p_j]) \leftarrow \emptyset$\;
  \label{line:abcast_end}
}

\caption{The \ACplus{} algorithm that tolerates up to $f$ failures; code executed by server $p_i$;
see Table~\ref{tab:notations} for digraph notations.}
\label{alg:allconcur}
\end{algorithm}


\subsubsection{Handling unreliable messages} 

Algorithm~\ref{alg:handleBCAST} shows the code executed by $p_i$ when it receives an unreliable message $m_j^{(e,r)}$, 
i.e., $m_j^{(e,r)}$ was sent by $p_j$ while in $\acState{e}{r}$. 
Unreliable messages cannot originate from subsequent epochs (cf.~Theorem~\ref{prop:unrel_msg}).
Also, unreliable messages from preceding states are  dropped (line~\ref{line:unrel_drop}), because messages from preceding epochs are outdated and, 
since $e \leq \tilde{e}$, messages from preceding rounds are outdated as well.
As a result, $m_j^{(e,r)}$ was sent from either $\acState{\tilde{e}}{\tilde{r}+1}$ (cf. Theorem~\ref{prop:unrel_msg}) or $\acState{\tilde{e}}{\tilde{r}}$ (cf. Theorem~\ref{prop:unique_states}). 
We defer the handling of premature messages, i.e., sent from subsequent states, to 
Appendix~\ref{app:prem_msg}.
Handling an unreliable message sent from $\acState{\tilde{e}}{\tilde{r}}$ consists of 
three operations (lines~\ref{line:unrel_handle_start}--\ref{line:unrel_handle_end}): 
(1) send $m_j^{(e,r)}$ further (using $G_U$); 
(2) A-broadcast own message (if not done so already);
and (3) try to complete round $\tilde{r}$ (see Appendix~\ref{app:try_to_adeliver}).

\begin{algorithm}[!tb]
\scriptsize
\DontPrintSemicolon

\SetKwData{Input}{$M_i$}
\SetKwData{Fails}{$F_i$}
\SetKwData{aServer}{$p_\ast$}
\SetKwData{aMessage}{$m_\ast$}
\SetKwData{Suspected}{$\mathit{Q}$}
\SetKwData{InputGraph}{$\mathbf{g_i}$}
\SetKwData{FastDigraph}{$G_U$}
\SetKwData{FTDigraph}{$G_R$}
\SetKwData{Epoch}{$\tilde{e}$}
\SetKwData{Round}{$\tilde{r}$}

\SetKwComment{Comment}{\{}{\}}

\SetKwFunction{HandleBCAST}{HandleBCAST}
\SetKwFunction{TryToAdeliver}{TryToComplete}

\SetKwBlock{handlebcast}{def \HandleBCAST{$m_j^{(e,r)}$}:}{end}
\handlebcast
{
  \tcc{$e \leq \Epoch$ (cf. Theorem~\ref{prop:unrel_msg})}
  \lIf{$e < \Epoch$ \normalfont{ or } $r<\Round$} {
    drop $m_j^{(e,r)}$ \label{line:unrel_drop}
  }
  \ElseIf(\Comment*[f]{$r = \Round+1$ (cf. Theorem~\ref{prop:unrel_msg})}){$r > \Round$} { \label{line:unrel_prem_start}
      \tcc{postpone handling $m_j^{(e,r)}$ for $\acState{\Epoch}{\Round+1}$}
      \lIf{$\forall m_\ast^{(e^\prime,r^\prime)} \in \Input^{\mathit{next}} : e^\prime = \Epoch$} { \label{line:unrel_prem_drop}
        $\Input^{\mathit{next}}\leftarrow \Input^{\mathit{next}} \cup \{m_j^{(e,r)}\}$
       }
  } \label{line:unrel_prem_end}
  \Else(\Comment*[f]{$e = \Epoch \wedge r = \Round \Rightarrow \acState{\Epoch}{\Round}$ (cf. Theorem~\ref{prop:unique_states})}) { 
    \broadcast{m_j^{(e,r)}} \Comment*[r]{send $m_j^{(e,r)}$ further via $\FastDigraph$} \label{line:unrel_handle_start}
    \tobroadcast{m_i^{(\Epoch,\Round)}} \Comment*[r]{A-broadcast own message} 
    \TryToAdeliver{}\; \label{line:unrel_handle_end}
   }
}
\caption{Handling an unreliable message---while in epoch $\tilde{e}$ and round $\tilde{r}$, $p_i$ receives $m_j^{(e,r)}$ sent by $p_j$ while in $\acState{e}{r}$.}
\label{alg:handleBCAST}
\end{algorithm}


\subsubsection{Handling reliable messages} 

Algorithm~\ref{alg:handleRBCAST} shows the code executed by $p_i$ when it receives a reliable message $m_j^{(e,r)}$, 
i.e., $m_j^{(e,r)}$ was sent by $p_j$ while in $\acRState{e}{r}$.
Reliable messages from preceding states are dropped (line~\ref{line:rel_drop}), because  messages from preceding epochs are outdated and, 
since $e > \tilde{e} \Rightarrow r > \tilde{r}$ (cf.~Theorem~\ref{prop:rel_msg}), messages from preceding rounds are outdated as well.
As a result, we identify three scenarios (cf.~Theorems~\ref{prop:rel_msg} and~\ref{prop:rel_msg_same_epoch}):
(1) $m_j^{(e,r)}$ was sent from the subsequent epoch (i.e., from $\acRState{\tilde{e}+1}{\tilde{r}+1}$);
(2) $m_j^{(e,r)}$ was sent from the current state $\acRState{\tilde{e}}{\tilde{r}}$;
and (3) $m_j^{(e,r)}$ was sent from the current epoch, but from the subsequent reliable round (i.e., from $\acRState{\tilde{e}}{\tilde{r}+1}$).
In all three scenario, $p_i$ is in a reliable round (cf.~Theorems~\ref{prop:rel_msg} and~\ref{prop:rel_msg_same_epoch}).

\begin{algorithm}[!tb]
\scriptsize
\DontPrintSemicolon

\SetKwData{Input}{$M_i$}
\SetKwData{Fails}{$F_i$}
\SetKwData{aServer}{$p_\ast$}
\SetKwData{aMessage}{$m_\ast$}
\SetKwData{Suspected}{$\mathit{Q}$}
\SetKwData{InputGraph}{$\mathbf{g_i}$}
\SetKwData{FastDigraph}{$G_U$}
\SetKwData{FTDigraph}{$G_R$}
\SetKwData{Epoch}{$\tilde{e}$}
\SetKwData{Round}{$\tilde{r}$}
\SetKwData{Vertices}{$\mathcal{V}$}

\SetKwComment{Comment}{\{}{\}}

\SetKwFunction{HandleRBCAST}{HandleRBCAST}
\SetKwFunction{Sort}{sort}
\SetKwFunction{UpdateTrackingDigraph}{UpdateTrackingDigraph}
\SetKwFunction{TryToAdeliver}{TryToComplete}

\SetKwBlock{handlerbcast}{def \HandleRBCAST{$m_j^{(e,r)}$}:}{end}
\handlerbcast
{
    \tcc{$e > \Epoch \Rightarrow r > \Round$ (cf. Theorem~\ref{prop:rel_msg})}
    \lIf{$e < \Epoch$ \normalfont{ or } $r<\Round$} {
      drop $m_j^{(e,r)}$ \label{line:rel_drop}
    }
    \ElseIf (\Comment*[f]{$e=\tilde{e}+1 \wedge r=\tilde{r}+1$ (cf. Theorem~\ref{prop:rel_msg})}) {$e > \Epoch$} 
    { \label{line:rel_prem_start}
      \tcc{send $m_j^{(e,r)}$; deliver later in $\acRState{\Epoch+1}{\Round+1}$}
      \textbf{send} $\langle \mathit{RBCAST},\, m_j^{(e,r)} \rangle$ \textbf{to} $p_i^+(\FTDigraph)$\;\label{line:rel_prem_send}
      \lIf{$\exists m_\ast^{(e^\prime,r^\prime)} \in \Input^{\mathit{next}} : e^\prime = \Epoch$} 
        {$\Input^{\mathit{next}}\leftarrow \emptyset$} \label{line:rel_prem_drop}
      $\Input^{\mathit{next}}\leftarrow \Input^{\mathit{next}} \cup \{m_j^{(e,r)}\}$ \label{line:rel_prem_end}
    }
    \Else (\Comment*[f]{$r=\Round \vee r=\Round+1$ (cf. Theorem~\ref{prop:rel_msg_same_epoch})}){
      \If(\Comment*[f]{skip transition}){$r=\Round+1$}{ \label{line:rel_skip_start}
        \ForEach (\Comment*[f]{A-deliver $\acState{\Epoch-1}{\Round}$}) {$m \in \Input^{\mathit{prev}}$} {\todeliver{m}} 
        $\Input^{\mathit{prev}} \leftarrow \emptyset$; $\Input\leftarrow\emptyset$; $\Input^{\mathit{next}} \leftarrow \emptyset$\;  
        \ForEach{$\aServer \in \Vertices(\FTDigraph)$} {
          $\Vertices(\InputGraph[\aServer])\leftarrow\{\aServer\}$\;
          \UpdateTrackingDigraph{$\InputGraph[\aServer]$, $\emptyset$, \Fails}
         }
        $\acRState{\Epoch}{\Round} \sTrans \acRState{\Epoch}{\Round+1}$  \Comment*[r]{$\tsk{}$} \label{line:rel_skip_end}
%
      }
      \rbroadcast{m_j^{(e,r)}} \Comment*[r]{send $m_j^{(e,r)}$ further via $\FTDigraph$} \label{line:rel_handle_start}
      \tobroadcast{m_i^{(\Epoch,\Round)}} \Comment*[r]{A-broadcast own message} 
      \TryToAdeliver{}\; \label{line:rel_handle_end}
    }
}
\caption{Handling a reliable message---while in epoch $\tilde{e}$ and round $\tilde{r}$, $p_i$ receives $m_j^{(e,r)}$ sent by $p_j$ while in $\acRState{e}{r}$;
see Table~\ref{tab:notations} for digraph notations.}
\label{alg:handleRBCAST}
\end{algorithm}

First, we defer the handling of premature messages (i.e., sent from $\acRState{\tilde{e}+1}{\tilde{r}+1}$) to Appendix~\ref{app:prem_msg}.
Second, handling a reliable message sent from the current state $\acRState{\tilde{e}}{\tilde{r}}$ 
consists of three operations (lines~\ref{line:rel_handle_start}--\ref{line:rel_handle_end}): 
(1) send $m_j^{(e,r)}$ further (using $G_R$);
(2) A-broadcast own message (if not done so already);
and (3) try to complete round $\tilde{r}$ (see Appendix~\ref{app:try_to_adeliver}).
Third, receiving a reliable message from the same epoch and a subsequent round while in a reliable round 
triggers a $\tsk{}$ transition (see Figure~\ref{fig:skip_trans}). 
$\tsk{}$ consists of three operations (lines~\ref{line:rel_skip_start}--\ref{line:rel_skip_end}):
(1) A-deliver last completed state (i.e., $\acState{\tilde{e}-1}{\tilde{r}}$);
(2) initialize the next round, including 
updating the tracking digraphs (see Appendix~\ref{app:td_update});
and (3) move to  $\acRState{\tilde{e}}{\tilde{r}+1}$.
Then, $m_j^{(e,r)}$ is handled as if it was received in the same state (lines~\ref{line:rel_handle_start}--\ref{line:rel_handle_end}).


\subsubsection{Handling failure notifications} 

Algorithm~\ref{alg:handleFAIL} shows the code executed by $p_i$ when it receives 
a failure notification sent by $p_k$ targeting one of its predecessors $p_j$.
Note that if $k = i$, then the notification is from the local FD. 
The notification is valid only if both the owner (i.e., $p_k$) and the target (i.e., $p_j$) 
are part of the reliable digraph~(\cref{sec:round_based}).
Handling a valid notification consists of five operations:
(1) send the notification further using $G_R$ (line~\ref{line:fd_send});
(2) if in an unreliable round, rollback to the latest A-delivered round and start a reliable round (lines~\ref{line:rollback_start}--\ref{line:rollback_end});
(3) update the tracking digraphs (see Appendix~\ref{app:td_update});
(4) add the failure notification to $F_i$ (line~\ref{line:fd_add});
and (5) try to complete current reliable round $\tilde{r}$ (see Appendix~\ref{app:try_to_adeliver}).

\begin{algorithm}[!tb]
\scriptsize
\DontPrintSemicolon

\SetKwData{Input}{$M_i$}
\SetKwData{Fails}{$F_i$}
\SetKwData{aServer}{$p_\ast$}
\SetKwData{aMessage}{$m_\ast$}
\SetKwData{Suspected}{$\mathit{Q}$}
\SetKwData{InputGraph}{$\mathbf{g_i}$}
\SetKwData{FastDigraph}{$G_U$}
\SetKwData{FTDigraph}{$G_R$}
\SetKwData{Epoch}{$\tilde{e}$}
\SetKwData{Round}{$\tilde{r}$}
\SetKwData{Done}{$\mathit{return}$}
\SetKwData{Vertices}{$\mathcal{V}$}

\SetKwComment{Comment}{\{}{\}}

\SetKwFunction{HandleFAIL}{HandleFAIL}
\SetKwFunction{TryToAdeliver}{TryToComplete}
\SetKwFunction{UpdateTrackingDigraph}{UpdateTrackingDigraph}

\SetKwBlock{handlefail}{def \HandleFAIL{$p_j$,$p_k$}:}{end}
\handlefail
{
  \tcc{if $k=i$ then notification from local FD}
  \lIf {$p_j \notin \Vertices(\FTDigraph)$ \normalfont{ or } $p_k \notin \Vertices(\FTDigraph)$ }{\Done}
  \textbf{send} $\langle \mathit{FAIL},\, p_j,\, p_k \rangle$ \textbf{to} $p_i^+(\FTDigraph)$ \Comment*[r]{send further via $\FTDigraph$} \label{line:fd_send}
%
  \If(\Comment*[f]{rollback}){$\acState{\Epoch}{\Round}$} { \label{line:rollback_start}
    $\Input\leftarrow\emptyset$; $\Input^{\mathit{next}} \leftarrow \emptyset$\;
    \lIf{$\Input^{\mathit{prev}} \neq \emptyset$} {
      $\acState{\Epoch}{\Round} \fTrans \acRState{\Epoch+1}{\Round-1}$
       \Comment*[f]{$\tur{}$}
    }
    \lElse(\Comment*[f]{$\tfr{}$}) {
      $\acFState{\Epoch}{\Round} \fTrans \acRState{\Epoch+1}{\Round}$
    } \label{line:rollback_end}
  }
  \ForEach{$\aServer \in \Vertices(\FTDigraph)$}
  {
    \UpdateTrackingDigraph{$\InputGraph[\aServer]$, \Fails, $\{(p_j,p_k)\}$}
  }  
  $\Fails \leftarrow \Fails \cup \{(p_j,p_k)\}$\; \label{line:fd_add}
  \TryToAdeliver{}\;
}
\caption{Handling a failure notification---while in epoch $\tilde{e}$ and round $\tilde{r}$, $p_i$ receives from $p_k$ a notification of $p_j$'s failure;
see Table~\ref{tab:notations} for digraph notations.}
\label{alg:handleFAIL}
\end{algorithm}


\subsubsection{Round completion} 
\label{app:try_to_adeliver}

After handling either a message or a failure notification, $p_i$ tries to complete the current round.
If successful, it tries to safely A-deliver messages and finally, it moves to the next state.
We distinguish between completing an unreliable round (lines~\ref{line:unrel_complete_start}--\ref{line:unrel_complete_end})
and completing a reliable round (lines~\ref{line:rel_complete_start}--\ref{line:rel_complete_end}).
The necessary and sufficient condition for $p_i$ to complete $\acState{\tilde{e}}{\tilde{r}}$ is to receive a message (sent in $\acState{\tilde{e}}{\tilde{r}}$)
from every server (line~\ref{line:unrel_complete_cond}).
The completion is followed by a $\tuu{}$ transition to $\acState{\tilde{e}}{\tilde{r}+1}$ (line~\ref{line:unrel_complete_next}).
Starting $\acState{\tilde{e}}{\tilde{r}+1}$ entails handling the messages (sent in $\acState{\tilde{e}}{\tilde{r}+1}$) received (by $p_i$) prematurely 
(lines~\ref{line:unrel_complete_prem_handle_start}--\ref{line:unrel_complete_prem_handle_end}); we defer this discussion to Appendix~\ref{app:prem_msg}.
Also, if any messages were already received for $\acState{\tilde{e}}{\tilde{r}+1}$, then $p_i$ A-broadcast its own message (line~\ref{line:unrel_complete_abcast}).  
In addition, if $\acState{\tilde{e}}{\tilde{r}}$ is not the first in a sequence of unreliable rounds, 
then $p_i$ A-delivers $\acState{\tilde{e}}{\tilde{r}-1}$ (line~\ref{line:unrel_complete_commit}).

\begin{algorithm}[!tb]
\scriptsize
\DontPrintSemicolon

\SetKwData{Input}{$M_i$}
\SetKwData{Fails}{$F_i$}
\SetKwData{aServer}{$p_\ast$}
\SetKwData{aMessage}{$m_\ast$}
\SetKwData{Suspected}{$\mathit{Q}$}
\SetKwData{InputGraph}{$\mathbf{g_i}$}
\SetKwData{FastDigraph}{$G_U$}
\SetKwData{FTDigraph}{$G_R$}
\SetKwData{Epoch}{$\tilde{e}$}
\SetKwData{Round}{$\tilde{r}$}
\SetKwData{Vertices}{$\mathcal{V}$}
\SetKwData{Edges}{$\mathcal{E}$}

\SetKwComment{Comment}{\{}{\}}

\SetKwFunction{TryToAdeliver}{TryToComplete}
\SetKwFunction{UpdateTrackingDigraph}{UpdateTrackingDigraph}
\SetKwFunction{UpdateUnreliableDigraph}{UpdateUnreliableDigraph}
\SetKwFunction{Sort}{sort}

\SetKwBlock{adeliver}{def \TryToAdeliver{}:}{end}
\adeliver
{
\If{$\acState{\Epoch}{\Round}$} { \label{line:unrel_complete_start}
    \If {$|\Input|=|\Vertices(\FastDigraph)|$} 
    { \label{line:unrel_complete_cond}
        \ForEach (\Comment*[f]{A-deliver $\acState{\Epoch}{\Round-1}$ (if it exists)}) {$m \in \Input^{\mathit{prev}}$} {\todeliver{m}} \label{line:unrel_complete_commit}
      $\acState{\Epoch}{\Round} \trans \acState{\Epoch}{\Round+1}$ \Comment*[r]{$\tuu{}$} \label{line:unrel_complete_next}
      \tcc{handle postp. unreliable messages}
      \ForEach {$m_\ast^{(e,r)} \in \Input^{\mathit{next}}$} {\label{line:unrel_complete_prem_handle_start}
        \textbf{send} $\langle \mathit{BCAST},\, m_\ast^{(e,r)} \rangle$ \textbf{to} $p_i^+(\FastDigraph)$ 
      }
      $\Input^{\mathit{prev}} \leftarrow \Input$; $\Input \leftarrow \Input^{\mathit{next}}$; $\Input^{\mathit{next}} \leftarrow \emptyset$\;\label{line:unrel_complete_prem_handle_end}
      \lIf{$\Input \neq \emptyset$}{
        \tobroadcast{m_i^{(\Epoch,\Round)}} 
      } \label{line:unrel_complete_abcast}         
    }
  }\label{line:unrel_complete_end}
  \ElseIf {$\Vertices(\InputGraph[\aServer])=\emptyset,\,\forall \aServer \in \Vertices(\FTDigraph)$} 
  {\label{line:rel_complete_cond} \label{line:rel_complete_start}
    \lForEach(\Comment*[f]{A-deliver $\acRState{\Epoch}{\Round}$}){$m \in \Input$} {
      \todeliver{m} \label{line:rel_complete_commit}
    }
     \If(\Comment*[f]{remove servers}) {$\{\aServer : \aMessage \notin \Input\} \neq \emptyset$}{
      \UpdateUnreliableDigraph{$\FastDigraph$, $\{\aServer : \aMessage \notin \Input\}$}\; \label{line:rel_complete_G1} 
      \ForEach{$p \in \{\aServer : \aMessage \notin \Input\}$}
      {
        $\Vertices(\FTDigraph) \leftarrow \Vertices(\FTDigraph) \setminus \{p\}$ \Comment*[r]{adjust $\FTDigraph$} \label{line:rel_complete_Gf} 
        $\Edges(\FTDigraph) \leftarrow \Edges(\FTDigraph) \setminus \{(x,y) : x = p \vee y = p\}$\;
        $\Fails \leftarrow \Fails  \setminus \{(x,y) : x = p \vee y = p\}$ \Comment*[r]{adjust $F_i$} \label{line:rel_complete_upF}
      }
     }
    $\Vertices(\InputGraph[p_\ast])\leftarrow\{p_\ast\},\,\forall p_\ast \in \Vertices(\FTDigraph)$\; \label{line:rel_complete_tdreset}
    $\Input^{\mathit{prev}} \leftarrow \emptyset$\;
    \If {$\Fails = \emptyset$} { \label{line:rel_complete_t2start}
      $\acRState{\Epoch}{\Round} \trans \acFState{\Epoch}{\Round+1}$ \Comment*[r]{$\trf{}$} 
      \tcc{handle postp. unreliable messages}
      \ForEach {$m_\ast^{(e,r)} \in \Input^{\mathit{next}}$} { \label{line:rel_complete_t2prem_handl_start}
      \textbf{send} $\langle \mathit{BCAST},\, m_j^{(e,r)} \rangle$ \textbf{to} $p_i^+(\FastDigraph)$
      }
      $\Input \leftarrow \Input^{\mathit{next}}$; $\Input^{\mathit{next}} \leftarrow \emptyset$\; \label{line:rel_complete_t2prem_handl_end}
    } \label{line:rel_complete_t2end}
    \Else { \label{line:rel_complete_t5start}
      $\acRState{\Epoch}{\Round} \fTrans \acRState{\Epoch+1}{\Round+1}$ \Comment*[r]{$\trr{}$}
      \If (\Comment*[f]{discard}) {$\exists m_\ast^{(e,r)} \in \Input^{\mathit{next}} : e = \Epoch-1$} {  \label{line:rel_complete_t5prem_handl_start}
        $\Input^{\mathit{next}} \leftarrow \emptyset$; $\Input \leftarrow \emptyset$
      }
      \Else(\Comment*[f]{deliver postp. reliable messages}){ 
        $\Input \leftarrow \Input^{\mathit{next}}$; $\Input^{\mathit{next}} \leftarrow \emptyset$\;
        \lForEach{$\aMessage \in \Input$}
        {
          $\Vertices(\InputGraph[p_\ast])\leftarrow \emptyset$
        }
      } \label{line:rel_complete_t5prem_handl_end}
      \lForEach{$\aServer \in \Vertices(\FTDigraph)$}
      {
        \UpdateTrackingDigraph{$\InputGraph[\aServer]$, $\emptyset$, \Fails}
      }
     } \label{line:rel_complete_t5end}
    \lIf{$\Input \neq \emptyset$}{
        \tobroadcast{m_i^{(\Epoch,\Round)}}          
    } \label{line:rel_complete_abcast}
  } \label{line:rel_complete_end} 
}
\caption{Round completion---while in epoch $\tilde{e}$ and round $\tilde{r}$, $p_i$ tries to complete the current round and safely A-deliver messages;
see Table~\ref{tab:notations} for digraph notations.}
\label{alg:trytoadeliver}
\end{algorithm}

To complete a reliable round, \ACplus{} uses early termination---the necessary and sufficient condition 
for $p_i$ to complete $\acRState{\tilde{e}}{\tilde{r}}$ is to stop tracking all messages~\cite{poke2017allconcur}, 
i.e, all its tracking digraphs are empty (line~\ref{line:rel_complete_cond}).
Consequently, a completed reliable round can be safely A-delivered (line~\ref{line:rel_complete_commit}).
Moreover, servers, for which no message was A-delivered, are removed~\cite{poke2017allconcur}.
This entails updating both $G_U$ to ensure connectivity~(\cref{sec:algo_description}) and $F_i$ (line~\ref{line:rel_complete_upF}).
Depending on whether $F_i$ is empty after the update, 
we distinguish between a no-fail transition $\trf{}$ (lines~\ref{line:rel_complete_t2start}--\ref{line:rel_complete_t2end})
and a fail transition $\trr{}$ (lines~\ref{line:rel_complete_t5start}--\ref{line:rel_complete_t5end}).
Starting $\acFState{\tilde{e}}{\tilde{r}+1}$ after $\trf{}$ entails handling the messages (sent in $\acFState{\tilde{e}}{\tilde{r}+1}$) 
received (by $p_i$) prematurely (lines~\ref{line:rel_complete_t2prem_handl_start}--\ref{line:rel_complete_t2prem_handl_end}).
Similarly, starting $\acRState{\tilde{e}+1}{\tilde{r}+1}$ after $\trr{}$ entails handling the messages (sent in $\acRState{\tilde{e}+1}{\tilde{r}+1}$) 
received (by $p_i$) prematurely (lines~\ref{line:rel_complete_t5prem_handl_start}--\ref{line:rel_complete_t5prem_handl_end}). 
We defer the handling of premature messages to Appendix~\ref{app:prem_msg}.
In addition, $\trr{}$ requires 
the tracking digraphs to be updated (see Appendix~\ref{app:td_update});
note that before starting either transition,  the tracking digraphs are reset (line~\ref{line:rel_complete_tdreset}). 
Finally, if any messages were already received for the new state (i.e., $\acFState{\tilde{e}}{\tilde{r}+1}$ or $\acRState{\tilde{e}+1}{\tilde{r}+1}$), 
then $p_i$ A-broadcast its own message (line~\ref{line:rel_complete_abcast}).


\subsubsection{Handling premature messages} 
\label{app:prem_msg}

In \ACplus{}, servers can be in different states~(\cref{sec:conc_states}). 
This entails that a server can receive both A-broadcast messages and failure notifications sent by another server 
from a future state; we refer to these as premature messages. 
When handling a premature message, the following two conditions must be satisfied.
First, the digraph on which the message has arrived needs to be consistent with the digraph on which the message is sent further.
In other words, if the server sending the message has previously updated the digraph, 
then the server receiving the message must postpone sending it further until it also updates the digraph.
Second, changing the message order (with respect to other messages) must not affect early termination; 
note that message order is not relevant for unreliable messages.

The way premature messages are handled depends on the scope of the failure notifications. 
In general, to avoid inconsistencies, failure notifications need to be specific to $G_R$: 
Once $G_R$ is updated, all failure notifications are discarded. 
Thus, we avoid scenarios where the failure notification's owner is not a successor of the target. 
%
To trivially enforce this, failure notifications can be made specific to an epoch (as in \AC{}~\cite{poke2017allconcur}). 
%
In this case, premature failure notifications (i.e., specific to a subsequent epoch) are postponed: Both 
their sending and delivering is delayed until the subsequent epoch. 
%
Since reliable messages are specific to an epoch\footnote{Reliable messages that trigger skip transitions are not premature messages.}, 
they can also be postponed; thereby, message order is preserved.
Once a new epoch starts, the failure notifications specific to the previous epoch are outdated and hence, discarded. 
This entails detecting again the failures of the faulty servers that were not removed in the previous epoch~(\cref{sec:round_based}).
%

As long as $G_R$ remains unchanged, re-detecting failures can be avoided, by only discarding outdated failure notifications 
that are invalid and resending the valid ones (i.e., with both owner and target not removed in the previous epoch).
%
Even more, resending failure notifications can be avoided by not postponing them, i.e. they are sent further and delivered immediately.
This requires that the valid failure notifications are redelivered at the beginning of each epoch, i.e., updating the tracking digraphs. 
Note that for message order to be preserved, premature reliable messages are also sent further immediately and 
only their delivery is postponed to their specific epoch.

In \ACplus{}, failure notifications are handled immediately. 
While $p_i$ is in epoch $\tilde{e}$ and round $\tilde{r}$, it can receive the following premature messages:
(1) unreliable messages sent from $\acState{\tilde{e}}{\tilde{r}+1}$ (cf.~Theorem~\ref{prop:unrel_msg}),
which are handled in $\acState{\tilde{e}}{\tilde{r}+1}$ (lines~\ref{line:unrel_complete_prem_handle_start}--\ref{line:unrel_complete_prem_handle_end} and~\ref{line:rel_complete_t2prem_handl_start}--\ref{line:rel_complete_t2prem_handl_end});
and (2) reliable messages sent from $\acRState{\tilde{e}+1}{\tilde{r}+1}$ (cf.~Theorem~\ref{prop:rel_msg}), 
which are sent immediately (line~\ref{line:rel_prem_send}), 
but delivered in $\acRState{\tilde{e}+1}{\tilde{r}+1}$ 
(lines~\ref{line:rel_complete_t5prem_handl_start}--\ref{line:rel_complete_t5prem_handl_end}). 
Postponed messages are stored in $M_i^\mathit{next}$; 
unreliable messages sent from $\acState{\tilde{e}}{\tilde{r}+1}$ are dropped in favor of reliable messages sent from $\acRState{\tilde{e}+1}{\tilde{r}+1}$ 
(line~\ref{line:unrel_prem_drop} and~\ref{line:rel_prem_drop}).


\subsubsection{Updating the tracking digraphs} 
\label{app:td_update}

The tracking digraphs are needed only for reliable rounds (i.e., for early termination).
Thus, the tracking digraphs are updated either when a reliable round starts 
or during a reliable round when receiving a valid failure notification (see Algorithms~\ref{alg:handleRBCAST}--\ref{alg:trytoadeliver}). 
Updating the tracking digraphs after receiving a failure notification follows the procedure described in \AC{}~\cite{poke2017allconcur}
(see Section~\ref{sec:algo_description} for details). 
Algorithm~\ref{alg:update_td} shows the \texttt{UpdateTrackingDigraph} procedure used by $p_i$ to update a tracking digraph $\mathbf{g_i}[p_\ast]$ 
after adding a set of new failure notifications to the set of already known failure notifications. 

\begin{algorithm}[!tb]
\scriptsize
\DontPrintSemicolon

\SetKwData{Fails}{$F$}
\SetKwData{oldFails}{$F_\mathit{old}$}
\SetKwData{newFails}{$F_\mathit{new}$}
\SetKwData{aServer}{$p_\ast$}
\SetKwData{aMessage}{$m_\ast$}
\SetKwData{Suspected}{$\mathit{Q}$}
\SetKwData{InputGraph}{$\mathbf{g_i}$}
\SetKwData{FTDigraph}{$G_R$}
\SetKwData{Vertices}{$\mathcal{V}$}
\SetKwData{Edges}{$\mathcal{E}$}

\SetKwComment{Comment}{\{}{\}}

\SetKwData{Done}{$\mathit{done}$}
\SetKwData{Continue}{$\mathit{continue}$}

\SetKwFunction{UpdateTrackingDigraph}{UpdateTrackingDigraph}

\SetKwBlock{tobcast}{def \UpdateTrackingDigraph{$\InputGraph[\aServer]$, \oldFails, \newFails}:}{end}
\tobcast
{
  $\Fails \leftarrow \oldFails$\;
  \ForEach{$(p_j,p_k) \in \newFails$}
  {
    $\Fails \leftarrow \Fails \cup \{(p_j,p_k)\}$\;
    \lIf{$p_j \notin \Vertices(\InputGraph[\aServer])$}{\Continue}
    \If (\Comment*[f]{maybe $p_j$ sent $m_\ast$ further before failing}) {$p_j^+(\InputGraph[\aServer]) = \emptyset$}{ \label{line:uptd_expand_start}
      $\Suspected \leftarrow \{(p_j,p):p \in p_j^+(\FTDigraph) \setminus \{p_k\}\}$
      \Comment*[r]{FIFO queue}
      \ForEach {$(p_p,p) \in \Suspected$} {
          \tcc{recursively expand $\InputGraph[\aServer]$}
          $\Suspected \leftarrow \Suspected \setminus \{(p_p,p)\}$\;
          \If{$p \notin \Vertices(\InputGraph[\aServer])$}
          {
            $\Vertices(\InputGraph[\aServer]) \leftarrow \Vertices(\InputGraph[\aServer]) \cup \{p\}$\;
            \lIf{$\exists (p,*) \in \Fails$} 
            {
              $\Suspected \leftarrow \Suspected \cup \{(p,p_s):p_s \in p^+(\FTDigraph)\}\setminus\Fails$
            }
          }
          $\Edges(\InputGraph[\aServer]) \leftarrow \Edges(\InputGraph[\aServer]) \cup \{(p_p,p)\}$\;
      } \label{line:uptd_expand_end}
    }
    \ElseIf {$p_k \in p_j^+(\InputGraph[\aServer])$}
    {
    \tcc{$p_k$ has not received $m_\ast$ from $p_j$}
      $\Edges(\InputGraph[\aServer]) \leftarrow \Edges(\InputGraph[\aServer]) \setminus \{(p_j,p_k)\}$\; \label{line:uptd_rm_edge}
    
    \lForEach  (\Comment*[f]{prune: no input}) {$p \in \Vertices(\InputGraph[\aServer])$ s.t. $\nexists \pi_{p_\ast,p}$ in $\InputGraph[\aServer]$} {\label{line:updtd_prune_start}
      $\Vertices(\InputGraph[\aServer]) \leftarrow \Vertices(\InputGraph[\aServer]) \setminus \{p\}$
    }
    }
    
    \If {$\forall p\in \Vertices(\InputGraph[\aServer]),\, (p,*)\in\Fails$}
    {\label{alg:lsc_no_dissemination}
      $\Vertices(\InputGraph[\aServer]) \leftarrow \emptyset$ \Comment*[r]{prune: no dissemination}
    }\label{line:updtd_prune_end}
    
  }
}

\caption{Updating a tracking digraph---$p_i$ updates $\mathbf{g_i}[p_\ast]$ after adding 
$F_\mathit{new}$ to the set $F_\mathit{old}$ of already known failure notifications;
see Table~\ref{tab:notations} for digraph notations.}
\label{alg:update_td}
\end{algorithm}


\ifdefined\PAPER

\fi


\subsection{Uniform atomic broadcast} 
\label{app:uniformity}

In non-uniform atomic broadcast, neither agreement nor total order holds for non-faulty servers.
Uniform  properties  are  stronger  guarantees, i.e., they apply to all server, including faulty ones~\cite{Defago:2004:TOB:1041680.1041682}:
\begin{itemize}
\setlength{\itemsep}{0pt}
  \item (Uniform agreement) If a server A-delivers $m$, then all non-faulty servers eventually A-deliver~$m$.
  \item (Uniform total order) If two servers $p_i$ and $p_j$ A-deliver messages $m_1$ and $m_2$, 
  then $p_i$ A-delivers $m_1$ before $m_2$, if and only if $p_j$ A-delivers $m_1$ before~$m_2$. 
\end{itemize}
We fist show that \AC{} solves \emph{uniform} atomic broadcast and, then, we adapt \ACplus{} to uniformity. 

\subsubsection{Uniformity in AllConcur}

\AC{} solves \emph{uniform} atomic broadcast---its correctness proof~\cite{poke2017allconcur,allconcur_tla} (that shows it solves non-uniform atomic broadcast) 
can be easily extended to prove uniformity. This is due to \emph{message stability}, i.e., before a server A-delivers a message, it sends it to all its $d(G) > f$ successors. 
The correctness proof entails proving the set agreement property~\cite{poke2017allconcur,allconcur_tla}:
\begin{itemize}
\setlength{\itemsep}{0pt}
  \item (Set agreement) Let $p_i$ and $p_j$ be two non-faulty servers that complete round $r$. Then, both servers agree on the same set of messages, i.e., $M_i^r = M_j^r$.
\end{itemize}

The proof of set agreement is by contradiction: It assumes that $m_\ast \in M_i^r$, but $m_\ast \notin M_j^r$.
Then, to reach a contradiction, it shows that in the digraph used by $p_j$ to track $m_\ast$ there is a server $q$ that 
$p_j$ cannot remove without receiving $m_\ast$. The existence of $q$ is given by both the existence of a path $\pi_{p_\ast,p_i}$ on which 
$m_\ast$ arrives at $p_i$ and the fact that $p_i$ is non-faulty. Yet, in the case of uniformity, $p_i$ can be faulty. 
To ensure the existence of $q$, we extend $\pi_{p_\ast,p_i}$ by appending to it one of $p_i$'s non-faulty successors 
(i.e., at least one of its $d(G) > f$ successors are non-faulty). Since $p_i$ sends $m_\ast$ to all its successors before completing round $r$, 
all servers on the extended path have $m_\ast$ and at least one is non-faulty. Thus, \AC{} solves \emph{uniform} atomic broadcast. 

\subsubsection{Uniformity in AllConcur+}

Because of message stability, \AC{} solves \emph{uniform} atomic broadcast.
Yet, as described in Appendix~\ref{app:design_details}, \ACplus{} solves non-uniform atomic broadcast---agreement and total order apply only to non-faulty servers. 
For \ACplus{} to guarantee uniform properties, the concept of message stability can be extended to rounds, i.e., \emph{round stability}---before a server A-delivers a round $r$ in epoch $e$, 
it must make sure that all non-faulty servers will eventually deliver $r$ in epoch $e$. 
This is clearly the case if $\acRState{e}{r}$ (due to \AC{}'s early termination mechanism~\cite{poke2017allconcur}).
Yet, round stability is not guaranteed when $\acState{e}{r}$. 
Let $p_i$ be a server that A-delivers $\acState{e}{r}$ after completing $\acState{e}{r+1}$ and then it fails. 
All the other non-faulty servers receive a failure notification while in $\acState{e}{r+1}$ and thus, rollback to $\acRState{e+1}{r}$.
Thus, all non-faulty servers eventually A-deliver round $r$ in epoch $e+1$, breaking round stability.  

To ensure round stability in \ACplus{}, $p_i$ A-delivers $\acState{e}{r}$ once it receives messages from at least $f$ servers 
in $\acState{e}{r+2}$\footnote{Clearly, $n > 2f$ in order to avoid livelocks while waiting for messages from at least $f$ servers.}
Thus, at least one non-faulty server A-delivers $\acState{e}{r}$ after completing $\acState{e}{r+1}$, which guarantees all other non-faulty servers 
A-deliver $\acState{e}{r}$ (cf.~Lemma~\ref{lemma:adeliv_unrel}). 
Note that delaying the A-delivery of unreliable rounds is the cost of providing uniform properties.

To prove both uniform agreement and uniform total order, we adapt \ACplus{}'s correctness proof~(\cref{sec:corectness_proof})
to the modification that guarantees round stability.

\begin{lemma} 
\label{lemma:uniform_adeliv_unrel}
Let $p_i$ be a server that A-delivers $\acState{e}{r}$ after completing $\acState{e}{r+1}$.
Then, any other non-faulty server $p_j$ eventually A-delivers $\acState{e}{r}$.
\end{lemma}
\begin{IEEEproof}
Due to round stability, at least one server that A-delivered $\acState{e}{r}$ after completing $\acState{e}{r+1}$ is non-faulty; let $p$ be such a server. 
Moreover, since $p_i$ A-delivers $\acState{e}{r}$ after completing $\acState{e}{r+1}$, 
$p_j$ must have started $\acState{e}{r+1}$ (cf.~Proposition~\ref{prop:one_end_all_start}),
and hence, completed $\acState{e}{r}$.
As a result, $p_j$ either receives no failure notifications, which means it eventually completes $\acState{e}{r+1}$ and A-delivers $\acState{e}{r}$ (since $n > 2f$), 
or receives a failure notification and moves to $\acRState{e+1}{r}$ (after $\tur{}$).
Yet, this failure notification will eventually trigger on $p_i$ a $\tur{}$ transition from $\acState{e}{r+2}$ to $\acRState{e+1}{r+1}$, 
which eventually will trigger on $p_j$ a $\tsk{}$ transition that leads to the A-delivery of $\acState{e}{r}$ (see Figure~\ref{fig:skip_trans}). 
\end{IEEEproof}

\begin{lemma} 
\label{lemma:uniform_same_epoch_deliv}
If a server A-delivers round $r$ in epoch $e$ (i.e., either $\acRState{e}{r}$ or $\acState{e}{r}$), 
then, any non-faulty server eventually A-delivers round $r$ in epoch~$e$.
\end{lemma}
\begin{IEEEproof}
If $p_i$ A-delivers $\acRState{e}{r}$ (when completing it), then every non-faulty server eventually A-delivers $\acRState{e}{r}$. 
The reason is twofold: (1) since $p_i$ completes $\acRState{e}{r}$, every non-faulty server must start $\acRState{e}{r}$ (cf.~Proposition~\ref{prop:one_end_all_start});
and (2) due to early termination, every non-faulty server eventually also completes and A-delivers $\acRState{e}{r}$.

Otherwise, $p_i$ A-delivers $\acState{e}{r}$ either once it completes the subsequent unreliable round 
or after a skip transition from $\acRState{e+1}{r}$ to $\acRState{e+1}{r+1}$~(\cref{sec:adeliv_msg}).
On the one hand, if $p_i$ A-delivers $\acState{e}{r}$ after completing $\acState{e}{r+1}$, 
then any other non-faulty server eventually A-delivers $\acState{e}{r}$ (cf.~Lemma~\ref{lemma:uniform_adeliv_unrel}).
On the other hand, if $p_i$ A-delivers $\acState{e}{r}$ after a skip transition, then at least one 
non-faulty server A-delivered $\acState{e}{r}$ after completing $\acState{e}{r+1}$. 
Thus, any other non-faulty server eventually A-delivers $\acState{e}{r}$ (cf.~Lemma~\ref{lemma:uniform_adeliv_unrel}).
\end{IEEEproof}

\begin{corollary}
\label{corol:uniform_same_epoch_deliv}
If two servers A-deliver round $r$, then both A-deliver $r$ in the same epoch $e$.
\end{corollary}

\begin{theorem}[Uniform set agreement]
\label{th:uniform_set_agreement}
If two servers A-deliver round $r$, then both A-deliver the same set of messages in round $r$.
\end{theorem}
\begin{IEEEproof}
Let $p_i$ and $p_j$ be two servers that A-deliver round $r$. 
Clearly, both servers A-deliver $r$ in the same epoch~$e$ (cf.~Corollary~\ref{corol:uniform_same_epoch_deliv}). 
Thus, we distinguish between $\acRState{e}{r}$ and $\acState{e}{r}$.
If $\acRState{e}{r}$, both $p_i$ and $p_j$ A-deliver the same set of messages due to the set agreement property of early termination.
If $\acState{e}{r}$, both $p_i$ and $p_j$ completed $\acState{e}{r}$, i.e., both received messages from all servers; 
thus, both A-deliver the same set of messages.
\end{IEEEproof}

\begin{theorem}[Uniform agreement]
\label{th:uniform_agreement}
If a server A-delivers $m$, then all non-faulty servers eventually A-deliver $m$.
\end{theorem}
\begin{IEEEproof}
We prove by contradiction. Let $p_i$ be a server that A-delivers $m$ in round $r$ and epoch $e$.
We assume there is a non-faulty server $p_j$ that never A-delivers $m$. 
According to Lemma~\ref{lemma:uniform_same_epoch_deliv}, $p_j$ eventually A-delivers round $r$ in epoch~$e$.
Yet, this means $p_j$ A-delivers (in round $r$) the same set of messages as $p_i$ (cf.~Theorem~\ref{th:uniform_set_agreement}), 
which contradicts the initial assumption.
\end{IEEEproof}

\begin{theorem}[Uniform total order]
\label{th:uniform_total_order}
If two servers $p_i$ and $p_j$ A-deliver messages $m_1$ and $m_2$, 
then $p_i$ A-delivers $m_1$ before $m_2$, if and only if $p_j$ A-delivers $m_1$ before $m_2$.
\end{theorem}
\begin{IEEEproof}
From construction, in \ACplus{}, every server A-delivers rounds in order (i.e., $r$ before $r+1$). 
Also, the messages of a round are A-delivered in a deterministic order. 
Moreover, according to both Lemma~\ref{lemma:uniform_same_epoch_deliv} and Theorem~\ref{th:uniform_set_agreement}, 
$p_i$ A-delivers $m_1$ and $m_2$ in the same states as $p_j$. Thus, $p_i$ and $p_j$ A-deliver 
$m_1$ and $m_2$ in the same order. 
\end{IEEEproof}


\fi

\end{document}